%% file: main.tex
\pdfoutput=1

\documentclass[a4paper,11pt]{article}
\usepackage{iftex}

\usepackage[utf8]{inputenc}
\usepackage[T1]{fontenc}

\usepackage{amsmath}
\usepackage{amssymb}
\usepackage{amsthm}
\usepackage{mathtools}
\usepackage{thmtools}
\usepackage{algorithm}
\usepackage{thm-restate}
\usepackage{algpseudocode}
\usepackage{enumitem}
\usepackage{subcaption}

\usepackage[disable]{todonotes}

\usepackage{isomath}
\usepackage{libertine}
\usepackage[libertine]{newtxmath} 
\usepackage[scaled=0.96]{zi4} 

\usepackage[DIV=12]{typearea} 

\usepackage{microtype}
\usepackage[font=small]{caption}
\usepackage[algo2e,ruled,vlined,noend,linesnumbered]{algorithm2e}

\usepackage{xcolor}
\usepackage{xspace}

\usepackage[
    style=alphabetic,
    maxalphanames=4,
    minalphanames=4,
    maxnames=20,
    minnames=20,
    backref=true,
    doi=true,
    url=true,
    backend=biber
]{biblatex}

\ifpdf
\usepackage[ocgcolorlinks]{hyperref} 
\else
\usepackage[hidelinks]{hyperref}
\fi

\input{commands}
\usepackage[capitalize]{cleveref}

\DeclareSourcemap{
	\maps[datatype=bibtex]{
		\map{
			\step[fieldset=abstract, null]
			\step[fieldset=pagetotal, null]
			\step[fieldset=timestamp, null]
			\step[fieldset=biburl, null]
			\step[fieldset=bibsource, null]
		}
		\map{
			\step[fieldsource=doi,final]
			\step[fieldset=url,null]
		}  
		\map[overwrite]{
			\pertype{inproceedings}
			\step[fieldset=editor, null]
			\step[fieldset=series, null]
			\step[fieldset=volume, null]
			\step[fieldset=publisher, null]
			\step[fieldset=isbn, null]
       }
	}
}

\newcommand{\antonis}[1]{\todo[linecolor=orange!50!black,backgroundcolor=orange!25,bordercolor=orange!50!black]{\scriptsize \textbf{AS:} #1}}

\colorlet{DarkRed}{red!50!black}
\colorlet{DarkGreen}{green!50!black}
\colorlet{DarkBlue}{blue!50!black}

\hypersetup{
	linkcolor = DarkRed,
	citecolor = DarkGreen,
	urlcolor = DarkBlue,
	bookmarksnumbered = true,
	linktocpage = true
}

\DontPrintSemicolon
\SetKwProg{Procedure}{Procedure}{}{}
\SetProcNameSty{textsc}
\SetFuncSty{textsc}
\SetKw{KwAnd}{and}
\SetKw{KwDo}{do}

\SetCommentSty{mycommfont}

\declaretheorem[numberwithin=section]{theorem}
\declaretheorem[numberlike=theorem]{lemma}

\declaretheorem[numberlike=theorem]{corollary}
\declaretheorem[numberlike=theorem]{Definition}
\declaretheorem[numberlike=theorem]{claim}
\declaretheorem[numberlike=theorem]{observation}

\newcommand{\dist}{\operatorname{dist}}
\newcommand{\distrIS}{distance-$r$ independent set\xspace}

\newcommand{\distIS}[1]{distance-#1 independent set\xspace}
\newcommand{\distDS}[1]{distance-#1 dominating set\xspace}

\newcommand{\subscript}[2]{$#1 _ #2$}
\newcommand{\old}{\text{(old)}}

\DeclareMathOperator{\polylog}{polylog}

\DeclareMathOperator*{\argmax}{argmax}

\title{Dynamic Consistent $k$-Center Clustering with Optimal Recourse}

\author{
  Sebastian Forster \thanks{Department of Computer Science, University of Salzburg, Salzburg, Austria. This work is supported by the Austrian Science Fund (FWF): P 32863-N. This project has received funding from the European Research Council (ERC) under the European Union's Horizon 2020 research and innovation programme (grant agreement No 947702).}
  \and 
  Antonis Skarlatos\samethanks[1]
}

\date{}

\ifpdf
\hypersetup{
    pdftitle = {Dynamic Consistent $k$-Center Clustering with Optimal Recourse},
    pdfauthor = {Sebastian Forster, Antonis Skarlatos}
}
\else
\fi

\begin{document}
\maketitle
\begin{abstract}
\input{trunk/abstract}
\end{abstract}

\thispagestyle{empty}
\newpage
\tableofcontents
\thispagestyle{empty}
\newpage

\setcounter{page}{1}

\input{trunk/introduction}
\input{trunk/techincal_overview}
\input{trunk/preliminaries}
\input{trunk/fully_consist_kcenter}
\input{trunk/decr_consist_kcenter}
\input{trunk/incr_consist_kcenter}
\input{trunk/acknowledgements}


\printbibliography[heading=bibintoc] 

\end{document}

%% file: commands.tex
\newcommand*\samethanks[1][\value{footnote}]{\footnotemark[#1]}

\allowdisplaybreaks[1]

\makeatletter
\g@addto@macro\bfseries{\boldmath}
\makeatother

\makeatletter
\g@addto@macro\mdseries{\unboldmath}
\g@addto@macro\normalfont{\unboldmath}
\g@addto@macro\rmfamily{\unboldmath}
\g@addto@macro\upshape{\unboldmath}
\makeatother

\makeatletter
\g@addto@macro\bfseries{\boldmath}
\def\thmhead@plain#1#2#3{%
  \thmname{#1}\thmnumber{\@ifnotempty{#1}{ }\@upn{#2}}%
  \thmnote{ {\the\thm@notefont\unboldmath(#3)}}}
\let\thmhead\thmhead@plain
\makeatother

\DeclareCiteCommand{\citem}
    {}
    {\mkbibbrackets{\bibhyperref{\usebibmacro{postnote}}}}
    {\multicitedelim}
    {}

\renewcommand*{\multicitedelim}{\addcomma\space}

\AtEveryCitekey{\clearfield{note}}

\makeatletter
\AtBeginDocument{%
    \newlength{\temp@x}%
    \newlength{\temp@y}%
    \newlength{\temp@w}%
    \newlength{\temp@h}%
    \def\my@coords#1#2#3#4{%
      \setlength{\temp@x}{#1}%
      \setlength{\temp@y}{#2}%
      \setlength{\temp@w}{#3}%
      \setlength{\temp@h}{#4}%
      \adjustlengths{}%
      \my@pdfliteral{\strip@pt\temp@x\space\strip@pt\temp@y\space\strip@pt\temp@w\space\strip@pt\temp@h\space re}}%
    \ifpdf
      \typeout{In PDF mode}%
      \def\my@pdfliteral#1{\pdfliteral page{#1}}
      \def\adjustlengths{}%
    \fi
    \ifxetex
      \def\my@pdfliteral #1{}
      \def\adjustlengths{\setlength{\temp@h}{-\temp@h}\addtolength{\temp@y}{1in}\addtolength{\temp@x}{-1in}}%
    \fi%
    \def\Hy@colorlink#1{%
      \begingroup
        \ifHy@ocgcolorlinks
          \def\Hy@ocgcolor{#1}%
          \my@pdfliteral{q}%
          \my@pdfliteral{7 Tr}
        \else
          \HyColor@UseColor#1%
        \fi
    }%
    \def\Hy@endcolorlink{%
      \ifHy@ocgcolorlinks%
        \my@pdfliteral{/OC/OCPrint BDC}%
        \my@coords{0pt}{0pt}{\pdfpagewidth}{\pdfpageheight}%
        \my@pdfliteral{F}
        \my@pdfliteral{EMC/OC/OCView BDC}%
        \begingroup%
          \expandafter\HyColor@UseColor\Hy@ocgcolor%
          \my@coords{0pt}{0pt}{\pdfpagewidth}{\pdfpageheight}%
          \my@pdfliteral{F}
        \endgroup%
        \my@pdfliteral{EMC}%
        \my@pdfliteral{0 Tr}
        \my@pdfliteral{Q}%
      \fi
      \endgroup
    }%
}
\makeatother

%% file: trunk/abstract.tex
Given points from an arbitrary metric space and a sequence of point updates sent by an adversary, what is
the minimum \emph{recourse} per update (i.e., the minimum number of changes needed to the set of centers after an update), 
in order to maintain a constant-factor approximation to a $k$-clustering problem? This question has received attention in recent years under the name \emph{consistent clustering}.

Previous works by Lattanzi and Vassilvitskii [ICLM '17] and Fichtenberger, Lattanzi, Norouzi-Fard, and Svensson [SODA '21]
studied $k$-clustering objectives, including the $k$-center and the $k$-median objectives, under only point insertions.
In this paper we study the $k$-center objective in the fully dynamic setting, where the update is either a point insertion or a point deletion. Before our work, Łącki \textcircled{r} Haeupler \textcircled{r} Grunau \textcircled{r} Rozhoň \textcircled{r} Jayaram~[SODA '24]
gave a deterministic fully dynamic constant-factor approximation algorithm for the $k$-center objective
with worst-case recourse of $2$ per point update (i.e., point insertion/point deletion).

In this work, we prove that the $k$-center clustering problem admits optimal recourse bounds by developing a deterministic fully dynamic
constant-factor approximation algorithm with worst-case recourse of $1$ per point update. Moreover our algorithm performs simple choices based on light data structures, and
thus is arguably more direct and faster than the previous one which uses a sophisticated combinatorial structure. 
Additionally to complete the picture, we develop a new deterministic decremental algorithm and a new deterministic incremental algorithm,
both of which maintain a $6$-approximate $k$-center solution with worst-case recourse of $1$ per point update. Our incremental algorithm improves over the $8$-approximation algorithm by Charikar, Chekuri, Feder, and Motwani [STOC '97].
Finally, we remark that since all three of our algorithms are deterministic, they work against an adaptive adversary.

%% file: trunk/introduction.tex
\section{Introduction}
Clustering is a fundamental and well-studied problem in computer science and artificial intelligence, and the goal is to find structure in data by grouping together similar data points.
Clustering arises in approximation algorithms, unsupervised learning, computational geometry,
community detection, image segmentation, databases, and other areas~\cite{hansen1997cluster, schaeffer2007graph, fortunato2010community, shi2000normalized, arthur2007k, tan2013data, coates2012learning}. Clustering algorithms optimize a given objective function, and specifically for $k$-clustering problems the goal is to output a set of $k$ centers that minimize a $k$-clustering objective.

One of the classical and most well-studied $k$-clustering objectives is the \emph{$k$-center objective}. In particular, given a metric space $(\mathcal{X}, \dist)$,
a set of points $P \subseteq \mathcal{X}$, and a positive integer $k \leq |P|$,
the goal of the $k$-center clustering problem is to output a subset $S \subseteq P$ of at most 
$k$ points, referred to as \emph{centers}, such that the maximum distance of any point in $P$ to its closest center is minimized. The $k$-center clustering problem admits a polynomial-time $2$-approximation algorithm~\cite{Gonzalez85, HochbaumS86}, and it is known to be NP-hard to approximate the $k$-center objective within a factor of $(2-\epsilon)$ for any constant $\epsilon > 0$~\cite{hsu1979easy}. The $k$-center clustering problem has also been studied
in the graph setting where the metric is induced by the shortest path metric of a weighted undirected graph~\cite{Thorup04, EppsteinHS20, AbboudCLM23}.

In recent years, with the explosion of data and the emergence of modern computational paradigms where data is constantly changing, there has been notable interest in developing dynamic clustering algorithms~\cite{Cohen-AddadSS16, GoranciHL18, HenzingerK20, GoranciHLSS21, HenzingerLM20, ChanGS18, BateniEFHJMW23, BhattacharyaCLP23, cruc_for_gor_yas_skar2024dynamic}. In the dynamic setting, the point set $P$ is subject to point updates sent by an adversary, and one of the goals is to maintain a set of centers with small approximation ratio and low update time. In the \emph{incremental setting} the adversary can only insert new points to $P$, in the \emph{decremental setting} the adversary can only delete points from $P$, and in the \emph{fully dynamic setting} the adversary can insert new points to $P$ and delete points from $P$. Several works focus on reducing the update time after a point update~\cite{ChanGS18, BateniEFHJMW23, BhattacharyaCLP23} and after an edge update in the graph setting~\cite{cruc_for_gor_yas_skar2024dynamic}.

In this paper, we study the dynamic \emph{consistent} clustering problem, and specifically
the dynamic \emph{consistent $k$-center} clustering problem. The main goal in this problem is to develop a dynamic algorithm that makes few changes to its set of centers after an update.
The number of changes made to the set of centers after an update (i.e., point insertion/point deletion) is the \emph{recourse} of the algorithm. Therefore, the goal is to achieve the minimum possible recourse while the approximation ratio of the maintained
solution is small. In turn, the goal can be either to minimize the amortized recourse or the worst-case recourse
in combination with either an oblivious or an adaptive adversary. Moreover, usually we want the approximation ratio of the maintained
solution to be constant.

Consistency is important from both
a theoretical and a practical point of view, and in recent
years it has received a lot of attention in the machine learning literature~\cite{LattanziV17, fichtenberger2021consistent, Cohen-AddadHPSS19, GuoKLX21, lkacki2024fully}.
From the theoretical perspective, minimizing the recourse for a fixed approximation ratio is a very interesting
combinatorial question because it shows how closely the approximate solutions align after an update occurs~\cite{lkacki2024fully}. Consistency is also closely related to the notion of low-recourse in the online algorithms literature~\cite{GuG016, MegowSVW12, BernsteinHR19, GuptaL20, BhattacharyaBLS23, EpsteinL14, SandersSS09, GuptaK0P17, GuptaKS14}. In our context, 
the irrevocability of past choices is relaxed, allowing for changes to past decisions.
From a practical perspective, imagine grouping people into $k$ teams. After a person
arrives or leaves, it makes sense not to change the set of $k$ leaders significantly, as over time the people in each group have developed trust with their respective leader.
Another example related to machine learning problems is when the $k$ centers are used for a task which is very costly to recompute, and so we aim to reduce the number of recomputations after a modification in the input.

\subsection{Related Work}
Consistent clustering was introduced by Lattanzi and Vassilvitskii~\cite{LattanziV17}, who developed an incremental $O(1)$-approximation algorithm for well-studied $k$-clustering problems (e.g., $k$-median, $k$-means, $k$-center) with $O(k^2 \log^4 n)$ total recourse (i.e., $O(k^2 \log^4 n / n)$ amortized recourse), where $n$ is the total number of points inserted by the adversary. They also provided an $\Omega(k \log n)$ lower bound on the total recourse (i.e., $\Omega(k \log n / n)$ lower bound for the amortized recourse). Later on, for the $k$-median clustering problem in the incremental setting, the total recourse was improved to $O(k \polylog n)$ by
Fichtenberger, Lattanzi, Norouzi-Fard, and Svensson~\cite{fichtenberger2021consistent}, which is optimal up to polylogarithmic factors.

Regarding the consistent $k$-center clustering problem, in~\cite{LattanziV17} the authors argue that in the incremental setting, the ``doubling algorithm'' by Charikar, Chekuri, Feder, and Motwani~\cite{CharikarCFM97} achieves the tight $O(k \log n)$ total recourse bound. To the best of our knowledge, there are no works that focus solely to the decremental setting. However recently,
Łącki \textcircled{r} Haeupler \textcircled{r} Grunau \textcircled{r} Rozhoň \textcircled{r} Jayaram~\cite{lkacki2024fully}
were the first to study the fully dynamic setting, and they developed a deterministic fully dynamic algorithm that maintains a constant-factor approximate solution with worst-case recourse of $2$. In more detail, they proved the following result.

\begin{theorem}[\cite{lkacki2024fully}]\label{th:fully_consist_priorwork}
    There is a deterministic fully dynamic algorithm that, given a point set $P$ from an arbitrary metric space subject to point insertions and point deletions and an integer $k \geq 1$, maintains a subset of points $S \subseteq P$ such that: 
    \vspace{-0.3em}
    \begin{itemize}
        \setlength\itemsep{-0.3em}
        \item The set $S$ is a $24$-approximate solution for the $k$-center clustering problem. 
        \item The worst-case recourse is at most $1$ for any point insertion
        and at most $2$ for any point deletion.
    \end{itemize}
\end{theorem}

The authors in~\cite{lkacki2024fully} mention that the update time of their algorithm is a large polynomial without stating the exact bound, and that their algorithm solves the $k$-center clustering problem for all values of $k$ simultaneously.

\paragraph{Motivation and Open Question.}
The most important factors in the consistent clustering problem are the recourse and the approximation ratio. 
Furthermore, note that worst-case guarantees are stronger than amortized guarantees, and obtaining a deterministic algorithm 
is crucial because it immediately implies robustness against an adaptive adversary. Observe that an adaptive
adversary can repeatedly keep deleting and inserting a highly significant point (e.g., a center), forcing any 
bounded-approximation algorithm to remove and add that point from/to its set of centers after every update. In turn, this implies a lower bound of $1$ on the amortized recourse. 

From that perspective, the algorithm by~\cite{lkacki2024fully} is almost optimal
because it achieves a worst-case recourse of $2$, provides a constant-factor approximation, works in the fully dynamic setting, and is deterministic---making it robust against an adaptive adversary.
For this reason, the following particularly intriguing question (also posed in~\cite{lkacki2024fully}) remains open, asking for optimal bounds:

\vspace{0.5em}
\noindent\fbox{%
    \parbox{\textwidth}{%
        \begin{center}
        \emph{Is there a deterministic fully dynamic algorithm that maintains a constant-factor approximate solution for the $k$-center clustering problem with worst-case recourse of at most $1$?}
        \end{center}
    }%
}

\subsection{Our Contributions}
In this work, we answer this question in the affirmative, which stands as our main result.

\begin{theorem} \label{th:fully_consist_intro}
    There is a deterministic fully dynamic algorithm that, given a point set $P$ from an arbitrary metric space 
    subject to point insertions and point deletions and an integer $k \geq 1$,
    maintains a subset of points $S \subseteq P$ such that: 
    \vspace{-0.3em}
    \begin{itemize}
        \setlength\itemsep{-0.3em}
        \item The set $S$ is a $50$-approximate solution for the $k$-center clustering problem.
        \item The worst-case recourse is at most $1$ for any point update.
    \end{itemize}
\end{theorem}

\paragraph{Comparison between Theorem~\ref{th:fully_consist_priorwork} and Theorem~\ref{th:fully_consist_intro}.}
Let us make a few remarks about our result in comparison with the result of~\cite{lkacki2024fully}. 
Our approximation ratio is slightly worse, but notice that our recourse is optimal and usually 
the main goal is to obtain optimal recourse with constant approximation ratio. In addition, 
our algorithm performs simple choices based on light data structures.
On the other hand, the algorithm by~\cite{lkacki2024fully} uses a more sophisticated data structure that the authors call leveled forest. Therefore, our approach is more natural and clean than the approach in~\cite{lkacki2024fully}.
Last but not least, our algorithm is faster and easily implementable.

\vspace{1em}

To complete the picture, we also study the consistent $k$-center clustering problem in the partially dynamic setting.
The motivation is to achieve optimal bounds on the recourse while obtaining a small constant for the approximation ratio.
Our result in the decremental setting is the following.

\begin{theorem}
    There is a deterministic decremental algorithm that, given a point set $P$ from an arbitrary metric space 
    subject to point deletions and an integer $k \geq 1$,
    maintains a subset of points $S \subseteq P$ such that:  \vspace{-0.3em}
    \begin{itemize}
        \setlength\itemsep{-0.3em}
        \item The set $S$ is a $6$-approximate solution for the $k$-center clustering problem.
        \item The worst-case recourse is at most $1$ for any point deletion.
    \end{itemize}
\end{theorem}

To the best of our knowledge, there are no prior works in the decremental-only setting. Usually for the $k$-center clustering problem it is harder to handle point deletions than point insertions, because for example the adversary can actually delete a center.
Note also that the algorithm by~\cite{lkacki2024fully} incurs larger recourse for point deletions, and in~\cite{ChanGS18} the most difficult updates were point deletions. 

On the other hand, most works on consistent clustering focus on the incremental setting. Nonetheless, there is still room for improvement in the incremental setting, and our respective result is as follows.

\begin{theorem}
    There is a deterministic incremental algorithm that,
    given a point set $P$ from an arbitrary metric space subject to point insertions and an integer $k \geq 1$,
    maintains a subset of points $S \subseteq P$ such that: 
    \vspace{-0em}
    \begin{itemize}
        \setlength\itemsep{-0.3em}
        \item The set $S$ is a $6$-approximate solution for the $k$-center clustering problem.
        \item The worst-case recourse is at most $1$ for any point insertion.
    \end{itemize}
\end{theorem}

Our incremental algorithm improves over the $8$-approximation ``doubling algorithm'' by Charikar, Chekuri, Feder, and Motwani~\cite{CharikarCFM97}. To obtain this improvement, we incorporate the Gonzalez's algorithm~\cite{Gonzalez85} into our incremental
algorithm. To compare the two incremental algorithms and for completeness, we demonstrate in Section~\ref{sec:doubling_alg} that with slight modifications, the ``doubling algorithm'' admits a worst-case recourse of $1$ as well. 
Finally, the fact that the incremental setting is easier than the decremental setting is evident in our algorithms, because our incremental algorithm has faster update time than our decremental algorithm. 
Specifically, the number of operations performed by our incremental algorithm after a point insertion, is a polynomial in the input parameter $k$.

There is also an interesting relationship between our three results. The merge of our incremental and decremental algorithms
with a few adjustments yields our fully dynamic algorithm, and conversely, the division of our fully dynamic algorithm with a few adjustments results in
our two partially dynamic algorithms. 

\subsection{Structure of the Paper}
In Section~\ref{sec:tech_over}, we
sketch our algorithms and give intuition about their guarantees. Some basic preliminaries follow in Section~\ref{sec:prel}. Then in Section~\ref{sec:fully}, Section~\ref{sec:decr}, and Section~\ref{sec:incr}, we provide the formal descriptions of the fully dynamic algorithm, the decremental algorithm, and the incremental algorithm, respectively. We remark that there are similarities among the three algorithms. Nevertheless, we have chosen to keep the sections independent, allowing the reader to approach each one separately at the cost of some repetition.

%% file: trunk/techincal_overview.tex
\section{Technical Overview} \label{sec:tech_over}
In this section, we give a high-level overview of our techniques and describe the intuition behind our algorithms.

\subsection{Approximation Ratio}
We start by describing the way the approximation ratio is controlled by satisfying three invariants throughout the algorithm.
Consider a \emph{$k$-center instance} which is a point set $P$ from an arbitrary metric space
and an integer parameter $k \geq 1$. We denote by $S$ the set of $k$ \emph{centers}, and assume that the distance of any point in $P$ to its closest center is at most $r$. Moreover, assume that there is a point $p_S$ in the current point set $P$ which is not a center such that the pairwise distances of points in $S \cup \{p_S\}$
are greater than $\frac{2r}{\rho}$. Then using standard arguments from the literature of the $k$-center clustering problem (see Corollary~\ref{cor:r_less2R}), we can deduce that the cost of the solution $S$ is at most $\rho$ times larger than the optimal cost of the $k$-center instance (i.e., the set $S$ is a $\rho$-approximate solution).

In our dynamic consistent $k$-center algorithms, we explicitly maintain a set of centers $S$ and a value $\delta$ called \emph{extension level}
which is associated with a \emph{radius} $r^{(\delta)}$. The goal then is to maintain a set of points $S$ that satisfies the following three invariants:
\vspace{-0.3em}
\begin{enumerate}
    \setlength\itemsep{-0.2em}
    \item The size of the set $S$ is exactly $k$.
    \item There exists a point $p_S \in P \setminus S$ such that the set $S \cup \{p_S\}$ is a \distIS{$r^{(\delta-1)}$}.
    \item The set $S$ is a \distDS{$(\alpha r^{(\delta)})$},
\end{enumerate}
\vspace{-0.3em}
\noindent
where $r^{(\delta)} = \beta r^{(\delta-1)}$, and $\alpha \geq 1, \beta > 1$ are fixed constants. 
The first invariant says that the set $S$ consist of $k$ centers. The second invariant combined with the first invariant
says that there exist $k + 1$ points whose pairwise distances are greater than $r^{(\delta-1)}$. The third invariant
says that the distance of any point to its closest center is at most $\alpha r^{(\delta)}$. Also, the relation between the 
two associated radii $r^{(\delta-1)}$ and $ r^{(\delta)}$ of the two respective extension levels $\delta-1$ and $\delta$ is $r^{(\delta)} = \beta r^{(\delta-1)}$.
Using the first and the second invariant, we can deduce that $r^{(\delta-1)} \leq 2R^*$ where $R^*$ is the optimal cost.
Together with the third invariant, we get that the set $S$ is a $(2\alpha \beta)$-approximate solution for the $k$-center clustering problem. Therefore by satisfying these three invariants for fixed constants $\alpha, \beta$, we can maintain
a $k$-center solution with constant approximation ratio.

In our setting, the point set $P$ is subject to point updates that are sent by an adaptive adversary, and in turn some invariants may not be satisfied after the update.
In this case, our algorithm either has to modify the set $S$ or increase/decrease the value of the extension level, in order to ensure small approximation ratio.
Our main goal is to achieve a \emph{worst-case recourse} of $1$, which means that the set $S$ should change by at most one center per update. In other words,
we want the symmetric difference between the set of centers before the update with the set of centers
after the update to be of size at most $2$.

\subsection{State of the Algorithm}
Before we continue describing how to handle updates, we briefly present the state of the algorithm.
As we have already explained, the algorithm maintains a set $S$ of $k$ \emph{center points}, and
a value $\delta$ which is called extension level and is associated with a radius $r^{(\delta)}$. 
Additionally for each center point $c \in S$, the algorithm \emph{internally} maintains a \emph{cluster} $C$ which consists of the points served by
the center $c$. We note that a point in $P$ is not necessarily part of the cluster of its closest center.\footnote{Thus, our algorithm internally maintains an explicit clustering consisting of $ k $ point sets, but the output it maintains only consists of a representative point from each of these clusters, the $ k $ centers.
Our overall approximation guarantee then holds for the closest-center clustering induced by these centers in which every point is assigned to the cluster of its closest center (which is in line with the $k$-center objective).}
The clusters form a partition of the point set $P$, and the idea is that only the center $c$ is responsible to serve
the points inside the corresponding cluster $C$. Furthermore, we maintain the center points of the set $S$ in some \emph{order}.
This means that the centers and the clusters have an index in $S$, and each cluster has the same index as its unique center.

A cluster is categorized as either a \emph{regular cluster}, an \emph{extended cluster}, or a \emph{zombie cluster},
and the corresponding center is respectively called either a \emph{regular center}, an \emph{extended center}, or a \emph{zombie center}.
Let us give a high-level description of each such \emph{state} of a cluster and its center.
A cluster $C$ is a regular cluster when all the points in $C$ are within distance $r^{(\delta)}$ from its center $c$. 
In the beginning of the algorithm, all clusters and centers are regular. After the insertion of a new point from the adversary, a cluster $C$ can become extended because some points from another cluster move to the cluster $C$. After the deletion of a center point $c$ from the adversary, the cluster $C$ of $c$ becomes a zombie cluster, and
the deleted center $c$ is replaced by another point $c'$ which is the zombie center of the cluster $C$.

A cluster $C$ becomes a zombie when its center is deleted but the cluster $C$ does not necessarily remain a zombie throughout the algorithm.
Once all the points in the cluster $C$ are within distance $r^{(\delta)}$ from its current center $c$, the cluster $C$ becomes regular again.
In order to control the approximation ratio, the algorithm guarantees that as long as a cluster $C$ is a zombie cluster, no additional points are added to $C$. To maintain this guarantee after a cluster becomes extended, the algorithm ensures that the distance between a zombie center
and all other centers is greater than $r^{(\delta)}$. Furthermore, if a point $p$ in a cluster $C$ is too far from its center $c$ (i.e., $\dist(p, c) > r^{(\delta)}$), and also $p$ is within distance~$r^{(\delta)}$ from another center $c'$ of a \emph{non-zombie} (i.e., either regular or extended) cluster $C'$, then the algorithm removes $p$ from~$C$ and adds it to~$C'$.

\subsection{Point Updates}
We describe here how the algorithm reacts to point updates sent by the adaptive adversary. 
In particular we demonstrate how after a point update, the set $S$ can continue satisfying the three invariants while ensuring
worst-case recourse of $1$. We remark that Invariant~1 is satisfied immediately by the way the algorithm
satisfies the other two invariants.

\subsubsection{Point Deletions}
Under point deletions, observe that all the invariants can be immediately violated. 
In the following, we explain how the algorithm satisfies Invariant~2 and Invariant~3 after a point deletion.

\paragraph{Satisfying Invariant~2.}
After a point deletion, Invariant~2 can be violated due to the deletion of a point from the set $S \cup \{p_S\}$,
where $p_S$ is the point which is the certificate for the second invariant.
Assume that the adversary deletes the point $p_S$ violating Invariant~2.
As a remedy, we introduce the following process that we call \emph{Decreasing Operation} (or \emph{Regulating Operation}), which
is called after every point update (and not just after a point deletion).
This process checks if the set $S$ is already a \distDS{$r^{(\delta)}$} and if so, the extension level $\delta$ is set to the smallest value $\delta'$ such that the set $S$ is a \distDS{$r^{(\delta')}$}. Then since $r^{(\delta-1)} < r^{(\delta)}$ we can argue that for the value $\delta - 1$,
the set $S$ is not a \distDS{$r^{(\delta-1)}$}, implying that the second invariant holds.

The situation becomes more difficult when the adversary deletes a center point $c_i$. 
Let $C_i$ be the cluster with center $c_i$, and let $i$ be their index in the ordered set $S$. Our aim is to replace the deleted center $c_i$ with another point $c'$ from the same cluster $C_i$,
such that $c'$ is at a distance greater than $r^{(\delta-1)}$ from the remaining centers (i.e., points in $S \setminus \{c_i\}$).
If this is possible, then by replacing $c_i$ with $c'$ and calling in the end the Decreasing Operation, Invariant~2 is satisfied.
Otherwise, all the points inside the cluster $C_i$ are blocked by some remaining centers, meaning that all the points in $C_i$
are within distance $r^{(\delta-1)}$ from the set $S \setminus \{c_i\}$ of the remaining center points. Notice that if a point $p$ in the cluster $C_i$ is blocked by another center of a regular cluster\footnote{The same is also true if the point $p$ is blocked by another center of an extended cluster.} then the algorithm can just move the point $p$ to the other regular cluster. Therefore we can safely assume that after the deletion of $c_i$, all the points
inside the cluster $C_i$ are blocked by some remaining center of a zombie cluster.

To that end, we demonstrate the more challenging case which is when every point inside the cluster $C_i$ is within distance $r^{(\delta)}$ from some remaining zombie center.\footnote{Note that $r^{(\delta)}$ is more strict than $r^{(\delta-1)}$ here. The reason is that we want to keep the zombie centers at a distance greater than $r^{(\delta)}$ from the rest of the centers. In this way, we can argue that a zombie cluster cannot become an extended cluster and its points cannot move to an extended cluster. In turn, this allows us to get a constant approximation ratio by setting the parameter $\alpha$ in Invariant~3 to a constant.} For a clearer explanation, we use explicit indices starting from $1$. Hence, assume that the adversary deletes the center $c_1$, and that all the points in $C_1$ are within distance $r^{(\delta)}$ from a remaining zombie center. Let $p_1$ be a point in the cluster $C_1$, and let $c_2$ be the zombie center of a zombie cluster $C_2$ such that $\dist(p_1, c_2) \leq r^{(\delta)}$. Assume that there is a point $p_2$ in the zombie cluster $C_2$ at a distance greater than
$r^{(\delta)}$ from the remaining center points in $S \setminus \{c_1\}$. In this case, 
the algorithm assigns $p_2$ to be the center of the cluster $C_2$ and
the point $c_2$ becomes the center of the cluster $C_1$. Then the algorithm adjusts the indices so that $c_1$ is the old $c_2$ and $c_2$ is the point $p_2$. But what if every point in $C_2$ is within distance $r^{(\delta)}$ from some remaining zombie center? Then let $c_3$ be the zombie center such that $\dist(p_2, c_3) \leq r^{(\delta)}$. Now this situation may continue, because all points inside the zombie cluster $C_3$ of $c_3$ could also be blocked by a zombie center. The idea then is to continue in the same way until we finally find a
zombie cluster $C_l$, for which there is a point $p_l \in C_l$ at a distance greater than $r^{(\delta)}$ from all the remaining centers (i.e., points in $S \setminus \{c_1\}$). Let us consider the following possible scenarios: 

\begin{enumerate}
    \item Assume that there is such a sequence of points $p_j, c_j$ found by the previous process (see Figure~\ref{fig:example}). In this case, the algorithm assigns the point $p_l$ to be the center
    of the cluster $C_l$ and shifts the centers accordingly with respect to the sequence found by the previous process.
    At this point, we remark that the actual action of the algorithm is to
    add the point $p_l$ to the set of centers $S$, and then adjust the indices of the existing centers.
    In this way the recourse of the algorithm is $1$, because the set $S$ changes by at most one point.
    Recall that the recourse does not count the adjustments in the implicit order of the centers.

    \item Otherwise, assume that there is no such sequence as described before. This means that the 
    previous process does not find such a sequence, but
    there is a crucial structural observation to be made for the clusters that participate in this process.
    All these clusters can be converted to regular clusters because all the points scanned by the process 
    are within distance $r^{(\delta)}$ from a center of this process. 
    Therefore we introduce the \emph{Reassigning Operation}, where the algorithm adds every point $p$ to the cluster that has its center
    within distance $r^{(\delta)}$ from $p$, and the algorithm converts every cluster that participated in the previous
    process to a regular cluster. As a result, all the points of the cluster $C_1$ of the deleted center $c_1$, move
    to a different regular cluster. Finally, the new center $c_1$ which forms a new cluster $C_1$ is set to be
    the furthest point from the remaining centers. By choosing $c_1$ like that, combined with the execution of the Decreasing Operation, Invariant~2 is satisfied.
\end{enumerate}

\begin{figure}
    \centering
    \begin{subfigure}{.33\textwidth}
      \centering
      \includegraphics[width=0.9\linewidth]{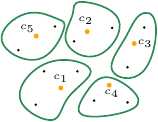}
    \end{subfigure}%
    \begin{subfigure}{.33\textwidth}
      \centering
      \includegraphics[width=0.9\linewidth]{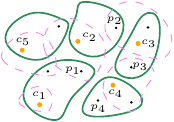}
    \end{subfigure}
    \begin{subfigure}{.33\textwidth}
      \centering
      \includegraphics[width=0.9\linewidth]{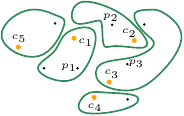}
    \end{subfigure}
     \caption{A $k$-center instance with $k = 5$. The green thick regions depict the clusters of each center. Notice that the clusters are pairwise disjoint and they form a partition of the input.\\
     \underline{Left figure}: The initial state of the instance. The adversary deletes $c_5, c_2, c_1$ one at a time. \\
     \underline{Middle figure}: The state of the instance after the three point deletions. The dashed pink regions depict the areas within a distance of $r^{(\delta)}$ from each center. At this moment the adversary deletes $c_1$, and there is a sequence of points $p_1, c_2, p_2, c_3, p_3, c_4, p_4$, such that $\dist(p_1, c_2) \leq r^{(\delta)}, \dist(p_2, c_3) \leq r^{(\delta)}, \dist(p_3, c_4) \leq r^{(\delta)}, \text{and} \dist(p_4, \{c_2, c_3, c_4, c_5\}) > r^{(\delta)}$.  \\
     \underline{Right figure}: The state of the instance after the adversary deletes $c_1$. Since such a sequence of points $p_j, c_j$ is found, the new $c_4$ is $p_4$, the indices of centers are shifted accordingly, and the clusters are updated respectively.}
    \label{fig:example}
\end{figure}

\paragraph{Satisfying Invariant~3.}
Recall that a cluster $C$ is a regular cluster when the distance between every point in $C$ and its center $c$ is 
at most $r^{(\delta)}$. When the adversary deletes the center $c$, the cluster $C$ becomes a zombie cluster 
and the algorithm replaces the old center $c$ with another point $c'$ which is called the zombie center of the cluster $C$. 
Observe that by the triangle inequality, the distance between any pair of points
in a regular cluster $C$ is at most $2r^{(\delta)}$. 
Therefore, by replacing the deleted (zombie) center $c$ 
with another zombie center in $C$, the factor $\alpha$ in Invariant~3 
increases at most by an extra factor of $2$.

However based on our previous discussion regarding Invariant~2, sometimes an immediate replacement 
for the deleted center point $c$ is not possible, because all points in $C$ are within distance $r^{(\delta)}$ from some remaining zombie center. In this case, the algorithm tries to detect a sequence as described before in the discussion about Invariant~2. The observation here is that by the way the algorithm reassigns the center points, 
in the worst-case the center $c$ becomes a point which used to be a center that was blocking a point $p$ in $C$. In turn, for the updated centers we have that $\dist(p, c) \leq r^{(\delta)}$, and by using the triangle inequality again we can deduce that the factor $\alpha$ in Invariant~3 is increased at most by an extra factor of $3$.
In turn, the approximation ratio is increased at most by an extra factor of $3$.
Similar arguments work also in the case where the cluster $C$ is not regular.

\subsubsection{Point Insertions}
Under point insertions, observe that the invariant that can be immediately violated is Invariant~3.

\paragraph{Satisfying Invariant~3.}
After a point insertion, Invariant~3 can be violated due to the insertion of a
new point at a distance greater than $\alpha r^{(\delta)}$ from all the current centers.
Nonetheless, we react already when the new inserted point is at a distance greater than $r^{(\delta)}$ from the set
of centers $S$. In this scenario, the idea is to increase the extension level to some $\delta'$, and in turn the associated radius to some $r^{(\delta')}$, through the \emph{Increasing Operation} (or \emph{Doubling Operation})
such that all points in $P$ are within distance $r^{(\delta')}$ from the unmodified (yet) set of centers $S$. 
However, we do not want to violate Invariant~2, and as a result this goal may not be feasible. 
Nevertheless, there is an important structural observation to be made for the center points and
the last value of the extension level $\delta$ that does not violate Invariant~2.
If there is a point at a distance greater than $r^{(\delta)}$ from the set $S$ and also the extension level $\delta$ cannot increase anymore without violating Invariant~2, then the minimum pairwise distance of center points in $S$ must
be at most $r^{(\delta)}$.

In this case, the algorithm finds a pair of center points $c_s, c_i \in S$ as follows: $c_s$ is the center of least index in $S$ such that there is another center $c_i$
within distance $r^{(\delta)}$ from $c_s$ (i.e., $\dist(c_s, c_i) \leq r^{(\delta)}$).\footnote{Note that the distance $\dist(c_s, c_i)$ is not necessarily equal to $\min_{x, y \in S,\, x \neq y} \dist(x, y)$. \label{footnote:xy_ci_cs}} Then the points from the cluster $C_i$ of $c_i$ move to the cluster $C_s$ of $c_s$, and the cluster $C_s$ becomes an extended cluster. Moreover by setting the parameter $\beta$ appropriately, we can argue that if the Increasing Operation is called then
only the new inserted point $p^+$ could be at a distance greater than $r^{(\delta)}$ from the set $S \setminus \{c_i\}$. Therefore, by
assigning $p^+$ to be the $i$-th center of the set $S$, Invariant~3 is satisfied.
We remark that the reason the center points $c_s, c_i$ are selected in this way is because
we can argue that the cluster $C_i$ of $c_i$ is a regular cluster, and so the factor $\alpha$ in Invariant~3 (and in turn the approximation ratio) increases at most by an extra factor of $2$.

\paragraph{Point Insertion and Zombie Cluster.}
We have already described the scenarios where the adversary inserts a new point at a distance greater than $r^{(\delta)}$ from all the current centers. Moreover, if the new point is inserted within distance $r^{(\delta)}$ from a regular or an extended center, then the
algorithm does not need to modify the set of centers since the new point can be added to the corresponding cluster. But what if the new point is within distance $r^{(\delta)}$ from a zombie center? If the extension level can increase without violating Invariant~2, 
then by doing so, all invariants are satisfied.

To that end, the more intricate case is when the two following events happen simultaneously:
(1) The new inserted point is at a distance greater than $r^{(\delta)}$ from all the regular and
extended clusters, but it is within distance $r^{(\delta)}$ from a zombie center $c_z$. 
(2) The minimum pairwise distance of center points is at most $r^{(\delta)}$. 
This is a subtle case because in order to control the approximation ratio, we do not want the zombie center to be responsible to serve new points, as then we can argue that the zombie center $c$ of a 
zombie cluster $C$ is close enough to any point in $C$.
Furthermore note that due to event (2), any increase to the extension level would violate Invariant~2 which is undesirable. 

In this case (see Figure~\ref{fig:example_1}), the algorithm chooses the pair of center points $c_s, c_i \in S$ as described before. Next, the cluster $C_s$ of $c_s$ becomes an extended cluster, and the exempted center point $c_i$ is set to be
the point $c_z$. In turn, the area of radius $r^{(\delta)}$ around $c_i$ becomes the new regular cluster $C_i$, and the
remaining points from the cluster $C_z$ of $c_z$ form the new $z$-th zombie cluster.
The next step is to find the new position of the zombie center $c_z$. To do this, we pretend that the point $c_z$
has been deleted by the adversary, and the algorithm proceeds accordingly as described before in the point deletion.

\begin{figure}
    \centering
    \begin{subfigure}{.33\textwidth}
      \centering
      \includegraphics[width=0.9\linewidth]{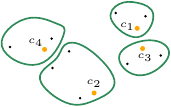}
    \end{subfigure}%
    \begin{subfigure}{.33\textwidth}
      \centering
      \includegraphics[width=0.9\linewidth]{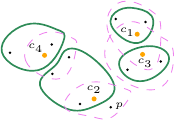}
    \end{subfigure}
    \begin{subfigure}{.33\textwidth}
      \centering
      \includegraphics[width=0.9\linewidth]{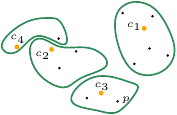}
    \end{subfigure}
     \caption{A $k$-center instance with $k = 4$. The green thick regions depict the clusters of each center. Notice that the clusters are pairwise disjoint and they form a partition of the input. \\
     \underline{Left figure}: The initial state of the instance. \\
     \underline{Middle figure}: The moment when the adversary inserts $p$ at a distance smaller than $r^{(\delta)}$ from the zombie center $c_2$. The dashed pink regions depict the areas within a distance of $r^{(\delta)}$ from each center. \\
     \underline{Right figure}: The state of the instance after the algorithm processes the insertion of $p$. At first, the cluster of $c_1$ becomes an extended cluster,
     the point $c_3$ is exempted from being center, and the new $c_3$ is the old $c_2$. Then the zombie cluster $C_2$ of $c_2$ is updated to consist of points from the old $C_2$ that are at a distance greater than $r^{(\delta)}$ from
     the new $c_3$. Finally, for the relocation of $c_2$ the algorithm proceeds as if the adversary had deleted the old point $c_2$ (without
     actually deleting it, since this point is $c_3$ now).}
    \label{fig:example_1}
\end{figure}

\paragraph{Gonzalez Operation.}
Finally, we describe the main reason why our incremental-only algorithm improves over the $8$-approximation ``doubling algorithm'' by~\cite{CharikarCFM97}.
Note that due to the absence of point deletions, the notion of zombie clusters (which appear when the center of a cluster is deleted), is not needed.
Apart from that, the main difference of our incremental algorithm to our fully dynamic algorithm is the way we maintain the order of the centers, and the way we choose the aforementioned pair of points $c_s, c_i \in S$.

Specifically in our incremental algorithm, we maintain an order of the centers with respect to their
pairwise distances, by running Gonzalez's algorithm~\cite{Gonzalez85} -- the heuristic that repeatedly selects the point farthest from the existing centers as the next center -- on the set $S$ of current centers. Next analogously to the fully dynamic algorithm, the Increasing (Doubling) Operation is performed after the insertion of a point, and then the algorithm chooses the aforementioned pair of centers $c_s, c_i \in S$ as follows: the distance between them is the minimum over all the pairwise distances of center points and $c_i$ has maximum index according to Gonzalez's algorithm. In this scenario, the distance between the new point and the set $S$ is greater than $r^{(\delta)}$, while the minimum pairwise distance among center points is at most $r^{(\delta)}$. Therefore, there exist $k + 1$  points whose pairwise distances are at least $\dist(c_s, c_i)$, and so $\dist(c_s, c_i) \leq 2R^*$. Using the properties of Gonzalez's algorithm, we show that actually any center point that has been exempted from being a center is within distance $2R^*$ from its closest center point. Also by Invariant~2, we can deduce that $r^{(\delta)} \leq 4R^*$. As a result, every point is within distance $6R^*$ from its closest center, and thus the $6$-approximation follows.

%% file: trunk/preliminaries.tex
\section{Preliminaries} \label{sec:prel}
Let $P$ be a point set from an arbitrary metric space $(\mathcal{X}, \dist)$. The notation $(\mathcal{X}, \dist)$ is omitted
whenever it is clear from the context.
For a point $p \in P$ and a subset of points $S \subseteq P$,
the \emph{distance} $\dist(p, S)$ between $p$ and $S$ is equal to $\min_{q \in S} \dist(p, q)$, namely
the distance from $p$ to its closest point in $S$.

\begin{Definition}[k-center clustering problem]
    Given a point set $P$ from an arbitrary metric space and an integer $k \geq 1$, 
    the goal is to output a subset of points $S \subseteq P$ of size at most~$k$,
    such that the value $\max_{p \in P} \dist(p, S)$ is minimized.
\end{Definition}

Consider a \emph{$k$-center instance} $(P, k)$, which is the pair of the given input point set and the
integer parameter $k$. Let $S$ be a feasible solution for the instance (i.e., a subset of at most $k$ points).
The points of the set $S$ are called \emph{centers} or \emph{center points}. 
The \emph{cost} (also referred to as \emph{radius}) of a feasible solution $S$ is equal to $\max_{p \in P} \dist(p, S)$.
We denote by $R^* \coloneqq \min_{|S| \leq k} \max_{p \in P} \dist(p, S)$  the optimal cost (also referred to as optimal radius) of the given instance, and by $S^*$ any optimal solution with cost $R^*$ (i.e., $\max_{p \in P} \dist(p, S^*) = R^*$).

In the \emph{dynamic setting}, the point set $P$ is subject to point updates sent by an adaptive adversary.
The number of changes made to the set of centers $S$ after a point update (i.e., point insertion/point deletion) is the \emph{recourse}. More formally, assuming that the set of centers is always of size $k$ (which is guaranteed by our algorithms), the recourse is half the size of the symmetric difference between the set of centers before the update with the set of centers after the update. 
\input{trunk/static_kcenter}

%% file: trunk/static_kcenter.tex
\subsection{Static $k$-Center Algorithms}
We present here some relevant definitions and results about two well-known
algorithms that we utilize later in our dynamic consistent $k$-center algorithms.
These two algorithms are the \emph{Gonzalez algorithm}~\cite{Gonzalez85} and the algorithm by Hochbaum and Shmoys~\cite{HochbaumS86},
and both of them achieve a $2$-approximation for the $k$-center clustering problem in the static setting.

\begin{Definition}[Distance-$r$ Independent Set]
    Given a point set $P$ from an arbitrary metric space and a parameter $r > 0$, 
    a \emph{distance-$r$ independent set} $S$ in $P$ is a subset of points such that
    the pairwise distances of points in $S$ are strictly more than $r$.
    In other words, it holds that $\min_{p, q \in S} \dist(p, q) > r$.
\end{Definition}

\begin{Definition}[Distance-$r$ Dominating Set]
    Given a point set $P$ from an arbitrary metric space and a parameter $r > 0$, 
    a \emph{distance-$r$ dominating set} $S$ in $P$ is a subset of points such that
    the distance between $S$ and every point in $P$ is at most $r$.
    In other words, it holds that $\max_{p \in P} \dist(p, S) \leq r$.
\end{Definition} 

The following results can be attributed to Hochbaum and Shmoys~\cite{HochbaumS86}
who presented a reduction from the $k$-center clustering problem to the maximal \distrIS problem.
In the analysis of our algorithms, we use the next lemmas to argue that a \distrIS of size at least $k + 1$ implies that our maintained solution has constant approximation ratio.

\begin{lemma}[\cite{HochbaumS86}] \label{lem:2r_atmostk}
    Consider a $k$-center instance $(P, k)$ with optimal radius $R^*$.
    Then for every $r \geq 2R^*$, it holds that any \distrIS is of size at most~$k$.
\end{lemma}

\begin{corollary} \label{cor:r_less2R}
    Consider a $k$-center instance $(P, k)$ with optimal radius $R^*$, and
    assume that there exists a \distrIS of size at least $k + 1$. Then
    it must be true that $r < 2R^*$.
\end{corollary}
\begin{proof}
    Supposing that $r \geq 2R^*$, leads to a contradiction due to Lemma~\ref{lem:2r_atmostk}, and so the claim follows.
\end{proof}

\begin{lemma} \label{lem:k+1_atleastr_lesseq2R}
    Consider a $k$-center instance $(P, k)$ with optimal radius $R^*$, and let $S$
    be a set of size at least $k + 1$ such that the pairwise distances of points in $S$ are at least $r$.
    Then it must be true that $r \leq 2R^*$.
\end{lemma}
\begin{proof}
    Consider an optimal solution $S^*$ of size $k$ with radius $R^*$. By the pigeonhole principle there must exist
    a center $c^* \in S^*$ and two different points $p, q \in S$ such that $\dist(p, c^*) \leq R^*$ and $\dist(q, c^*) \leq R^*$.
    By the triangle inequality, we have that $\dist(p, c^*) + \dist(q, c^*) \geq \dist(p, q) \geq r$.
    Thus, without loss of generality we can assume that $\dist(p, c^*) \geq \frac{r}{2}$. In turn, this implies that $R^* \geq \frac{r}{2}$, as needed.
\end{proof}

The next theorem is the algorithm by Hochbaum and Shmoys~\cite{HochbaumS86} which
says that one can compute a $2$-approximate solution for the $k$-center clustering problem
via $O(\log |P|)$ invocations of an algorithm that computes a maximal \distrIS.

\begin{theorem}[\cite{HochbaumS86}] \label{th:2appr_alg}
    Consider a $k$-center instance $(P, k)$, where $P$ is a point set from an arbitrary metric space,
    and let $D$ be the sorted set of the ${|P| \choose 2}$ pairwise distances of points in $P$. 
    By applying binary search on the set $D$ and calling a maximal \distrIS algorithm $\mathcal{A}$ at each iteration, we can find
    a $2$-approximate solution for the $k$-center instance. The running time of the algorithm is $O(T_{\mathcal{A}} \cdot \log |P|)$,
    where $T_{\mathcal{A}}$ is the running time of $\mathcal{A}$.
\end{theorem}
\begin{proof}
    Let $r$ be the value of a binary search iteration on the set $D$.
    At each binary search iteration, we execute an algorithm $\mathcal{A}$ that computes a maximal \distrIS in $P$, and let $S$ be the
    output of $\mathcal{A}$. If the size of $S$ is greater than $k$, then the binary search moves to larger values in the set $D$,
    otherwise the binary search moves to smaller values in the set $D$.
    
    Let $r_\text{min}$ be the smallest $r$ such that $\mathcal{A}$
    returns a maximal \distrIS $S_\text{min}$ of size at most~$k$. By Lemma~\ref{lem:2r_atmostk} we have that
    that $r_\text{min} \leq 2R^*$, since the value $2R^*$ is such a candidate.
    Notice that by the maximality of $S_\text{min}$, we can infer that $S_\text{min}$ is also a \distDS{$r_\text{min}$}. In turn, this implies that
    the returned set $S_\text{min}$ is a $2$-approximate solution.
    The value of $r_\text{min}$ is found after $O(\log |D|) = O(\log |P|)$ executions of $\mathcal{A}$.
\end{proof}

%% file: trunk/fully_consist_kcenter.tex
\section{Fully Dynamic Consistent $k$-Center Clustering} \label{sec:fully}
In the fully dynamic setting, we are given as input a point set $P$ from an arbitrary metric space which undergoes point insertions
and point deletions. The goal is to fully dynamically maintain a feasible solution for the $k$-center clustering problem with 
small approximation ratio and low recourse. In this section, we develop a deterministic fully dynamic algorithm that 
maintains a $50$-approximate solution for the $k$-center clustering problem
with a worst-case recourse of $1$ per update, as demonstrated in the following theorem.
Recall that by $R^*$ we denote the current optimal radius of the given instance.

\begin{restatable}{theorem}{fullyconsist} \label{th:fully_consist}
    There is a deterministic fully dynamic algorithm that, given a point set $P$ from an arbitrary metric space 
    subject to point insertions and point deletions and an integer $k \geq 1$,
    maintains a subset of points $S \subseteq P$ such that: 
    \vspace{-0.3em}
    \begin{itemize}
        \setlength\itemsep{-0.3em}
        \item The set $S$ is a \distDS{$50R^*$} in $P$ of size $k$.
        \item The set $S$ changes by at most one point per update.
    \end{itemize}
\end{restatable}

Our algorithm maintains a set $S = \{c_1, \ldots, c_k\}$ of $k$ \emph{center points} in some \emph{order}
and a value $\delta$ called \emph{extension level}. Further, the algorithm associates a \emph{radius} $r^{(\delta)}$ with 
each value of $\delta$. Each center point $c_i$ is the unique center of the \emph{cluster} $C_i$ which
contains exactly the points served by $c_i$. In our algorithm, every point belongs to exactly
one cluster (i.e., every point is served by exactly one center), and so
the idea is that the center $c_i$ is responsible to serve every point inside the cluster $C_i$.
A cluster $C_i$ is categorized as either a \emph{regular cluster}, an \emph{extended cluster}, or a \emph{zombie cluster}.
The corresponding center is called \emph{regular center} (\emph{regular center point}), \emph{extended center} (\emph{extended center point}), 
or \emph{zombie center} (\emph{zombie center point}), respectively. 
We refer to this as the \emph{state} of the cluster and its center.
Each cluster $C_i$ that belongs in the
union of regular and extended clusters is called \emph{non-zombie cluster}, and its corresponding center is called \emph{non-zombie center} (\emph{non-zombie center point}). 
We define the sets $S_\text{reg}$, $S_\text{ext}$, $S_\text{zmb}$, and $S_\text{nz} = S_\text{reg} \cup S_\text{ext}$
as the center points of $S$ that are regular, extended, zombie, and non-zombie centers respectively.

A regular cluster $C_i$ contains only points that are
within distance $r^{(\delta)}$ from the corresponding regular center $c_i$. However, an extended or a zombie cluster $C_i$ can also contain points that are at distances greater than $r^{(\delta)}$ from their corresponding center $c_i$. 
We remark that the implicit order of center points in $S$ may change
after a point update, but we guarantee that the set $S$ itself changes by at most one point per update (i.e., the worst-case recourse is $1$).
Throughout the algorithm we want to satisfy the following three invariants:
\begin{enumerate}[label=\textbf{Invariant \arabic*}, itemindent=3em, itemsep=-0.2em]
    \item\label{inv_1_fully} \hspace{-0.4em}: The size of the set $S$ is exactly $k$.
    \item\label{inv_2_fully}\hspace{-0.4em}: There is a point $p_S \in P \setminus S$ such that the set $S \cup \{p_S\}$ is a \distIS{$r^{(\delta-1)}$}.
    \item\label{inv_3_fully}\hspace{-0.4em}: The set $S$ is a \distDS{$5r^{(\delta)}$}.
\end{enumerate}

\subsection{Fully Dynamic Algorithm}
At first we describe the preprocessing phase of the algorithm, and afterwards we describe
the way the algorithm handles point updates. A pseudocode of the algorithm is provided in Algorithm~\ref{alg:fully_consist}.

\subsubsection{Preprocessing Phase}
In the preprocessing phase, we run the algorithm of Theorem~\ref{th:2appr_alg} to obtain a $2$-approximate solution $S$, and let
$r \coloneqq \max_{p \in P} \dist(p, S)$. Next, we repeat the following process with the goal to satisfy all the three invariants. 
While the set $S$ has size strictly less than
$k$, we add to $S$ an arbitrary point $c$ that is at a distance greater than $\frac{r}{5}$ from the current centers (i.e., $\dist(c, S) > \frac{r}{5}$).
If there is no such point, the value of $r$ is divided by five, and the process continues until the size of $S$ reaches $k$. 
Once the size of $S$ becomes $k$, the value of $r$ continues to be divided by five as long as the set $S$ is a \distDS{$\frac{r}{5}$}.

In the special case where the point set $P$ contains fewer than $k$ points, we assume that the preprocessing phase is ongoing and the
set $S$ consists of all points of $P$. Notice that in this scenario the three invariants may not be satisfied. However, this is not an issue as every point serves as a center. Eventually, when the size of $P$ becomes at least $k + 1$, we execute once the previous steps of the preprocessing phase.

The value of the associated radius $r^{(\delta')}$ for each $\delta'$ is set to $5^{\delta'} \cdot r$, where $r$ is the final value after the preprocessing phase has finished. Initially, the extension level $\delta$ of the algorithm is set to zero, all the center points of $S$ are regular centers, and all the clusters are regular clusters. Also, each cluster is set to contain points in $P$ that are within distance $r^{(0)}$ from its center, such that the clusters
form a partition of the point set $P$.

\paragraph{Analysis of the preprocessing phase.}
After the preprocessing phase has finished, the size of $S$ is $k$ and the set $S$ is a \distDS{$r^{(\delta)}$}, 
where $\delta$ is the extension level which is initially zero.
Additionally, by construction there exists a point $p_S \in P \setminus S$ such that the set $S \cup \{p_S\}$
is a \distIS{$r^{(\delta-1)}$}. Hence, in the beginning all the three invariants are satisfied.
Combining \ref{inv_1_fully} with \ref{inv_2_fully}, we can infer that there exists a \distIS{$r^{(\delta-1)}$} of size at least $k + 1$.
Thus it holds that $r^{(\delta-1)} < 2R^*$ by Corollary~\ref{cor:r_less2R}, and since $r^{(\delta)} = 5r^{(\delta-1)}$ by construction, it follows that the set $S$ is initially a $10$-approximate solution.\footnote{Actually, the set $S$ is initially a $2$-approximate solution, because we run at first the $2$-approximation algorithm of Theorem~\ref{th:2appr_alg} and the addition of more points to $S$ can only reduce its approximation ratio. However, we use this simplistic analysis as a warm-up to demonstrate the utility of the invariants.}

\subsubsection{Point Updates}
Let $\delta$ be the current extension level of the algorithm, and recall that $r^{(\delta)}$ is the associated radius. Under point updates, the algorithm proceeds as follows.

\paragraph{Point deletion.} 
Consider the deletion of a point $p$ from the point set $P$, and let $P \coloneqq P \setminus \{p\}$. If $p$ is not a center point (i.e., $p \notin S$),
then the point $p$ is removed from the cluster it belongs to, and 
the algorithm just calls the Decreasing Operation which is described below. Note that the set $S$ remains unmodified and the recourse is zero.
Otherwise assume that the deleted point $p$ is the $i$-th center point of $S$ (i.e., $p \in S$ and specifically $p = c_i$).
The algorithm removes the point $p$ from the set $S$ and from the cluster $C_i$ it belongs to.
Since the center $c_i$ of $C_i$ is deleted, the cluster $C_i$ becomes a zombie cluster
(if the deleted point $c_i$ is a zombie center, then its corresponding cluster $C_i$ is already a zombie cluster). 
Next, the algorithm calls the Decreasing Operation\footnote{The Decreasing Operation is called at this moment in order to ensure that all the points of the zombie cluster $C_i$ that are within distance $r^{(\delta)}$
from a non-zombie center, move to the corresponding non-zombie cluster.} and then continues based on the two following cases:

\begin{enumerate}[label=C\arabic*.]
    \item If there is another point $c'$ in $C_i$ at a distance greater than $r^{(\delta)}$ from the remaining centers of $S$ (i.e., if $\exists c' \in C_i$ such that $\dist(c', S) > r^{(\delta)}$),\footnote{Recall that at this moment $c_i \notin S$ and $c_i \notin C_i$. Moreover, all points initially in $C_i$ that are within distance $r^{(
    \delta)}$ from a non-zombie center have been removed from $C_i$ and added to the corresponding non-zombie cluster. \label{footnote:c_i_del_decr}} then the point 
    $c'$ is assigned as the center of the zombie cluster $C_i$. In particular the point $c'$ is added to~$S$, and as the point $c'$ is now the $i$-th center point of the ordered set $S$, the algorithm sets $c_i \coloneqq c'$ and
    the point $c_i = c'$ becomes the zombie center of the zombie cluster $C_i$. 
    The recourse in this case is $1$.

    \item Otherwise, every point $c' \in C_i$ is within distance $r^{(\delta)}$ from the remaining centers of $S$. In this case, the algorithm tries to detect a sequence of $2l-1$ points of the following form:
    \begin{equation} \label{eq:seq_for_replac_fully}
        p_{t_1}, c_{t_2}, p_{t_2}, \;\ldots\;, c_{t_l}, p_{t_l}, \tag{seq}
    \end{equation}    
    where $t_1, \dots, t_l$ are indices different from each other,
    $t_1 = i, p_{t_1} \in C_i$, for every $j: 2 \leq j \leq l$ the points $p_{t_j}, c_{t_j}$ satisfy the following conditions:
    \begin{itemize}
        \item $c_{t_j} \in S \cap S_\text{zmb}$ is a zombie center such that $\dist(p_{t_{j-1}}, c_{t_j}) \leq r^{(\delta)}$,\footref{footnote:c_i_del_decr}

        \item $t_j$ is the index of $c_{t_j}$ in the ordered set $S$,

        \item $p_{t_j} \in C_{t_j}$ such that $\dist(p_{t_j}, c_{t_j}) > r^{(\delta)}$, \antonis{Is $\dist(p_{t_j}, c_{t_j}) > r^{(\delta)}$ necessary?}
    \end{itemize}
    and $\dist(p_{t_l}, S) > r^{(\delta)}$ for the last index $t_l$.

    \begin{enumerate}[label=C2\alph*.]
        \item If the algorithm finds such a sequence, then the point $p_{t_l}$ is assigned as the center of the zombie cluster $C_{t_l}$,
        and the indices in the ordered set $S$ are shifted to the right with respect to the sequence (see Figure~\ref{fig:example}).
        In particular the point $p_{t_l}$ is added to $S$, and let $q_{t_j}$ be temporarily the center point $c_{t_j}$ 
        for every $j: 2 \leq j \leq l$. Next, the algorithm sets $c_{t_l}$ to $p_{t_l}$, and $c_{t_j}$ to $q_{t_{j+1}}$ for every $j: 1 \leq j \leq l - 1$.
        The clusters are updated respectively, namely the point $q_{t_j}$ is removed from $C_{t_j}$ and the center point $c_{t_{j-1}}$ is added to $C_{t_{j-1}}$ for every $j: 2 \leq j \leq l$,\footnote{Note that actually the center point $c_{t_{j-1}}$ is now the point $q_{t_j}$.} 
        and note that the new center point $c_{t_l}$ is already in $C_{t_l}$.\footnote{$c_{t_l}$ belongs to $C_{t_l}$
        because the algorithm sets $c_{t_l}$ to $p_{t_l}$, and $p_{t_l}$ belongs to $C_{t_l}$ based on the third condition of~(\ref{eq:seq_for_replac_fully}).}
        After the change of the order, for all $j: 1 \leq j \leq l$ the center point $c_{t_j}$
        is the zombie center of the zombie cluster $C_{t_j}$. Notice that this operation affects only the order 
        of the centers in $S$ and not the set $S$ itself, which means that the recourse in this case is $1$.

        \item Otherwise, no such sequence exists that satisfies~(\ref{eq:seq_for_replac_fully}). In this case, the algorithm first calls the Reassigning Operation and then updates
        the $i$-th cluster, both described as follows: 

        \vspace{-0.5em}
        \subparagraph{Reassigning Operation.} \antonis{check that there are no repetitions and it runs in polynomial time.}
        Let $Q$ be a queue initially containing the points of the $i$-th cluster $C_i$.\footref{footnote:c_i_del_decr}
        During the Reassigning Operation, the following steps are performed iteratively as long as the queue $Q$ is not empty:
        \begin{itemize}
            \item First, the algorithm picks a point $q \in Q$ and finds a center point $c_{j'} \in S$ such that $\dist(q, c_{j'}) \leq r^{(\delta)}$.\footref{footnote:c_i_del_decr} Such a center point $c_{j'}$ must exist, since there is no sequence satisfying~(\ref{eq:seq_for_replac_fully}). Let $C_j$ be the cluster that contains the point $q$.

            \item The algorithm then removes the point $q$ from $C_j$ and $Q$, and adds $q$ to the cluster $C_{j'}$. Additionally, the cluster $C_{j'}$ becomes a regular cluster and its center $c_{j'}$ becomes a regular center.

            \item Next, all points in $C_{j'}$ that are at distance greater than $r^{(\delta)}$ from their center $c_{j'}$ are added to the queue $Q$ (in~\cref{algline:add_farpnt_to_Q}).
        \end{itemize}
        If the updated queue $Q$ is now empty then the Reassigning Operation ends. Otherwise, it is repeated using the updated queue~$Q$.
        
        \vspace{0.5em}
        Once the Reassigning Operation terminates, the algorithm finds a point $c'$ at maximum distance from the set 
        $S$ of the remaining centers, and builds the new $i$-th regular cluster with $c'$ as its center. In particular the point $c'$ is added to $S$, and as the point $c'$ is now the $i$-th center point of the ordered set $S$, the algorithm sets $c_i \coloneqq c'$ and the cluster $C_i$ is updated to contain 
        initially only the center point $c_i = c'$. The cluster $C_i$ is designated as a regular cluster and its center point $c_i$ as a regular center. The recourse in this subcase is $1$.
    \end{enumerate}
\end{enumerate}

\paragraph{Point insertion.}
Consider the insertion of a point $p^+$ into the point set $P$, and let $P \coloneqq P \cup \{p^+\}$. 
If $p^+$ is within distance $r^{(\delta)}$ from a non-zombie center point $c_j$ of $S$, then the point $p^+$ is added
to the corresponding non-zombie cluster $C_j$, and the algorithm just calls the Decreasing Operation which is described below.
Note that the set $S$ remains unmodified and the recourse is zero. 
Otherwise if the inserted point $p^+$ is at a distance greater than $r^{(\delta)}$
from all the non-zombie centers, then the algorithm continues based on the two following cases:

\begin{enumerate}[label=C\arabic*.]
    \setcounter{enumi}{2}
    \item If there exist two different center points such that the distance between them is at most $r^{(\delta)}$ 
    (i.e., $\min_{x, y \in S,\, x \neq y} \dist(x, y) \leq r^{(\delta)}$), 
    then let $c_s \in S$ be the center point of least index which is within distance $r^{(\delta)}$ from a different center point $c_i \in S$.\footref{footnote:xy_ci_cs}
    That is, the center points $c_s, c_i \in S$ satisfy the following conditions:
    \begin{equation} \label{choose_cs_ci_fully}
        \dist(c_s, c_i) \leq r^{(\delta)}, \, s\text{ is the minimum possible index in $S = \{c_1, \ldots, c_k\}$, and $s < i$.} \tag{ext}
    \end{equation}
    The algorithm exempts the center point $c_i$ from being a center and
    extends the cluster $C_s$. Specifically, the point $c_i$ is removed from the set $S$,
    which means that the corresponding cluster $C_i$ of $c_i$ loses its center.
    For that reason, the algorithm adds all the points inside the cluster $C_i$ to the cluster $C_s$.\footnote{In the analysis (see Lemma~\ref{lem:rem_ci_reg}) we show that $c_i$ must be a regular center, and thus the merging of these two clusters does not explode the approximation ratio.} The center point $c_s$ becomes an extended center and its cluster $C_s$ becomes an extended cluster. Next, the algorithm replaces in $S$ the exempted center point $c_i$ with another
    point, in the following way:
    
    \begin{enumerate}[label=C3\alph*.]
        \item If the new inserted point $p^+$ is within distance $r^{(\delta)}$ from a zombie center $c_z \in S$
        which is the $z$-th center point of the ordered set $S$,\footnote{Recall that at this moment the point which was the $i$-th center before the update has been removed from the set of centers $S$.\label{footnote:ci_exempt}} then the algorithm assigns the point $c_z$ to be the $i$-th center
        and assigns another point to serve as the $z$-th center in the ordered set $S$, as follows (see Figure~\ref{fig:example_1}). 
        Let $C'_z$ be temporarily the zombie cluster $C_z$. The algorithm sets $c_i \coloneqq c_z$,
        and the center point $c_i$ becomes a regular center. Moreover, the corresponding cluster $C_i$ is constructed as the set of points from $C'_z \cup \{p^+\}$ that are within distance $r^{(\delta)}$ from the updated regular center $c_i$, and also $C_i$ becomes a regular cluster.
        Afterwards, the zombie cluster $C_z$ is updated to consist of the remaining points from $C_z'$ whose distances are strictly more than $r^{(\delta)}$ from the center point $c_z$ (note that at this moment we have that $c_z = c_i$). To find the point which will serve as the $z$-th center in the ordered set $S$, the algorithm proceeds accordingly based on cases C1 and C2,
        after replacing in the algorithm's description the set $C_i$ with the set $C_z$ and the index~$i$ 
        with the index $z$ (see also Line~\ref{algline:call_replac} and the function Replace() in Algorithm~\ref{alg:fully_consist}).

        \item Otherwise, the new point $p^+$ is at a distance greater than $r^{(\delta)}$ from the set~$S$ of the remaining centers (i.e., $\dist(p^+, S) > r^{(\delta)}$).\footref{footnote:ci_exempt} In this case, the point $p^+$ is added to $S$ and as the point $p^+$ is now the $i$-th center point of the ordered set $S$, the algorithm sets $c_i \coloneqq p^+$. Additionally, the corresponding cluster $C_i$ is updated to be equivalent to the set $\{p^+\}$. The center point $c_i$ becomes a regular center and its corresponding cluster $C_i = \{p^+\}$ becomes a regular cluster. The recourse in this subcase is $1$.
    \end{enumerate} 

    \item Otherwise, the pairwise distances between the center points (i.e., points in $S$) are greater than $r^{(\delta)}$. 
    In this case, if there is another point $p \in P \setminus S$ at a distance strictly greater than $r^{(\delta)}$ from the set~$S$, 
    then the algorithm calls the Increasing Operation which is described below.
    Next for the updated value of the extension level~$\delta$,\footnote{Note that the value of the extension level $\delta$ may not change.} if every point is within distance $r^{(\delta)}$ from the set~$S$ (i.e., $\max_{p \in P} \dist(p, S) \leq r^{(\delta)}$), then the new inserted point $p^+$ is added to the cluster of its closest center. Note that the set $S$ remains unmodified and the recourse is zero. Otherwise, let $p$ be a point of maximum distance from the set~$S$, which means
    that $\dist(p, S) > r^{(\delta)}$. Subsequently, the algorithm finds the center points $c_s$ and $c_i$ as defined previously in~(\ref{choose_cs_ci_fully}), and observe that such centers must exist by the construction of the Increasing Operation. The algorithm proceeds as in case C3 (and eventually as in case C3b) by replacing in the algorithm's description the point $p^+$ with the point $p$ (see also Lines~\ref{algline:c3_aft_c4_beg}-\ref{algline:c3_aft_c4_end} in Algorithm~\ref{alg:fully_consist}).

    \subparagraph{Increasing Operation.} 
    At the start of the Increasing Operation, all clusters and centers become regular. Then the extension level $\delta$ is increased by one until either the set $S$ stops being a \distIS{$r^{(\delta)}$}
    or the set $S$ becomes a \distDS{$r^{(\delta)}$}. Recall that we have~$r^{(\delta)} = 5r^{(\delta-1)}$.
\end{enumerate}

\vspace{0.5em}
\noindent
In the end of each point update, the fully dynamic algorithm calls the Decreasing Operation.

\paragraph{Decreasing Operation.}
    During the Decreasing Operation the algorithm performs three tasks:
    \begin{enumerate}
        \item First, if the set $S$ is a \distDS{$r^{(\delta)}$} then all clusters and centers become regular, and $\delta$ is updated to the smallest value $\delta'$ such that the set $S$ is a \distDS{$r^{(\delta')}$}.\footnote{Note that the value of the extension level $\delta$ can either decrease or remain the same.}

        \item Second, every cluster $C_j$ (where $1 \leq j \leq k$) in which all points $q \in C_j$ lie within distance $r^{(\delta)}$ from its center $c_j$, becomes a regular cluster and its center $c_j$ becomes a regular center (in~\cref{algline:cluster_reg_2_fully}).

        \item Finally, every point $q \in P$ that satisfies the following conditions:
        \begin{itemize}
            \item $q$ is at a distance greater than $r^{(\delta)}$ from the center $c_j$ of the cluster $C_j$ that contains $q$,
            \item $q$ is within distance $r^{(\delta)}$ from another non-zombie center $c_{j'}$ of the non-zombie cluster $C_{j'}$,
        \end{itemize}
        is removed from $C_j$ and added to $C_{j'}$.\footnote{In case there is more than one non-zombie center point within distance $r^{(\delta)}$ from the point $q$, the algorithm adds $q$ to only one of the corresponding non-zombie clusters arbitrarily.}
    \end{enumerate}

\begin{algorithm}[H]
    \DontPrintSemicolon
    \caption{\textsc{fully dynamic consistent $k$-center}{}}
    \label{alg:fully_consist}

    \SetAlgoLined
    \SetArgSty{textrm}

    \tcp{Let $S = \{c_1, \dots, c_k\}$ be the ordered set of $k$ centers}

    \tcp{The value $r^{(\delta')}$ for each $\delta'$ is set to $5^{\delta'} \cdot r$, for the final value of $r$ in the preprocessing phase}
    \tcp{Let $\delta$ be the current extension level}

    \vspace{0.5em}
    
    \SetKwFunction{FIncreasingOp}{IncreasingOp}
    \Procedure{\FIncreasingOp{}} {
        For all $c_j \in S: c_j$ becomes regular center and $C_j$ becomes regular cluster \label{algline:cluster_reg_3_fully}
        \vspace{0.5em}
        
        \While{$\min_{x, y \in S,\,x \neq y} \dist(x, y) > r^{(\delta)} \textbf{ and } \max_{p \in P} \dist(p, S) > r^{(\delta)}$} {
            $\delta \gets \delta + 1$ \tcp{It holds that $r^{(\delta)} = 5 r^{(\delta-1)}$} \label{algline:incr_ext_lvl}
        }
    }

    \vspace{1em}

    \SetKwFunction{FDecreasingOp}{DecreasingOp}
    \Procedure{\FDecreasingOp{}} {
        \If {$\max_{q \in P} \dist(q, S) \leq r^{(\delta)}$} { \label{algline:decr_radius_1_fully}
            For all $c_j \in S: c_j$ becomes regular center and $C_j$ becomes regular cluster \label{algline:cluster_reg_1_fully}
        
            \vspace{0.5em}
            
            Find the smallest $\delta'$ such that $S$ is a \distDS{$r^{(\delta')}$}
            
            $\delta \gets \delta'$ \label{algline:decr_radius_2_fully}
        }
        \vspace{0.5em}
        
        \For{$j \in \{1, \ldots, k\}$} {
            \If{$\forall q \in C_j: \dist(q, c_j) \leq r^{(\delta)}$} {
                $C_j$ becomes regular cluster and $c_j$ becomes regular center \label{algline:cluster_reg_2_fully}
            }
        }
        
        \vspace{0.5em}
        
        \For{$q \in P$} {
            Let $C_j$ be the cluster such that $q \in C_j$
            \vspace{0.5em}
            
            \If{$\dist(q, c_j) > r^{(\delta)}$ \textbf{and} $\exists j': (C_{j'}$ is \textbf{not} zombie cluster \textbf{and} $\dist(q, c_{j'}) \leq r^{(\delta)}$)} { \label{algline:switch_cluster_1_fully_if}
                Remove $q$ from $C_j$ and add $q$ to $C_{j'}$ 
                \label{algline:switch_cluster_1_fully}
            }
        }
    }
    
    \vspace{1em}
    
    \SetKwFunction{FReassigningOp}{ReassigningOp}
    \Procedure{\FReassigningOp{$Q$}} {
        \While{$Q$ is not empty} {
            Pick $q$ from $Q$, and let $C_j$ be the cluster such that $q \in C_j$

            Find center $c_{j'} \in S$ such that $\dist(q, c_{j'}) \leq r^{(\delta)}$ \label{algline:find_c_fully}

            Remove $q$ from $C_j$ and $Q$, and add $q$ to $C_{j'}$ \label{algline:switch_cluster_2_fully}

            Cluster $C_{j'}$ becomes regular cluster and $c_{j'}$ becomes regular center \label{algline:cluster_reg_reas_oper}

            Add all points from $\{q_{j'} \in C_{j'} \mid \dist(q_{j'}, c_{j'}) > r^{(\delta)}\}$ to $Q$ \label{algline:add_farpnt_to_Q}
        }
    }
\end{algorithm}

\begin{algorithm}[H]
    \DontPrintSemicolon
    \SetAlgoLined
    \SetArgSty{textrm}
    
    \tcp{Continuation of Algorithm~\ref{alg:fully_consist}}
    \vspace{0.8em}

    \SetKwFunction{FReplace}{Replace}
    \Procedure{\FReplace{$C, n$}} {
        \If(\tcp*[h]{case C1}){$\exists c' \in C: \dist(c', S) > r^{(\delta)}$} { \label{algline:caseC1_fully}
            Add $c'$ to $S$ \label{algline:add_to_S_2}

            $c_n \gets c'$
            
            $c_n$ becomes zombie center
        } 
        \vspace{0.5em}
        
        \Else(\tcp*[h]{case C2}){
            Find a sequence $p_{t_1}, c_{t_2}, p_{t_2}, \;\ldots\;, c_{t_l}, p_{t_l}$ 
            of $2l-1$ points such that: 

            $t_1, \dots, t_l$ are indices different from each other,
            
            $t_1 = n, p_{t_1} \in C$, for every $j: 2 \leq j \leq l$: \label{algline:seq_init}
            
            (1) $c_{t_j} \in \{c \in S: c \text{ is a zombie center}\}$ and $\dist(p_{t_{j-1}}, c_{t_j}) \leq r^{(\delta)}$,

            (2) $t_j$ is the index of $c_{t_j}$ in $S = \{c_1, \ldots, c_k\}$,
            
            (3) $p_{t_j} \in C_{t_j}$ and $\dist(p_{t_j}, c_{t_j}) > r^{(\delta)}$,
            
            and $\dist(p_{t_l}, S) > r^{(\delta)}$ for the last index $t_l$ \label{algline:ptl_S}
            
            \vspace{0.5em}
            
            \If(\tcp*[h]{case C2a}){such a sequence is found} {
                Add $p_{t_l}$ to $S$ \label{algline:add_to_S_3}

                \vspace{0.5em}
                
                \For{$j \in \{1, \ldots, l-1\}$} {
                    $c_{t_j} \gets c_{t_{j+1}}$ \label{algline:shift_centers_1}

                    Remove $c_{t_j}$ from $C_{t_{j+1}}$ and add $c_{t_j}$ to $C_{t_j}$ \label{algline:shift_centers}

                    $c_{t_j}$ becomes zombie center
                }
                \vspace{0.5em}
                $c_{t_l} \gets p_{t_l}$ \tcp{$c_{t_l}$ is already in $C_{t_l}$}

                $c_{t_l}$ becomes zombie center
            }

            \vspace{0.5em}
            \Else(\tcp*[h]{case C2b}) {
                \FReassigningOp{$C$}
                
                Let $c'$ be a point of maximum distance from the set $S$

                Add $c'$ to $S$ \label{algline:add_to_S_4}

                $c_n \gets c'$, and $c_n$ becomes regular center \label{algline:reset_ci_C2b_fully}
            
                $C_n \gets \{c_n\}$, and $C_n$ becomes regular cluster \label{algline:reset_Cli_C2b_fully}
            }
        }
    }

    \vspace{1em}
    
    \SetKwFunction{FDeletePoint}{DeletePoint}
    \Procedure{\FDeletePoint{p}} {
        Remove $p$ from $P$
        \vspace{0.5em}
        
        \If{$p \notin S$} {
            Remove $p$ from the cluster it belongs to
            
            \FDecreasingOp{}

            \vspace{0.5em}
            \textbf{exit} from \FDeletePoint{}
        }
        \vspace{0.5em}
        \Else(\tcp*[h]{$p \in S$}) {   
            Let $i$ be the index of $p$ in $S = \{c_1, \dots, c_k\}$, namely $c_i = p$
            
            Remove $p$ from $S$ and $C_i$ \label{algline:rem_from_S_2}
            \vspace{0.5em}
            
            \If{$c_i$ is \textbf{not} zombie center} {
                Cluster $C_i$ becomes zombie cluster
            }
            \vspace{0.5em}

            \FDecreasingOp{}
            
            \FReplace{$C_i, i$}
        }

        \vspace{0.5em}
        \FDecreasingOp{}
    }
\end{algorithm}

\begin{algorithm}[H]
    \DontPrintSemicolon
    \SetAlgoLined
    \SetArgSty{textrm}

    \tcp{Continuation of Algorithm~\ref{alg:fully_consist}}
    \vspace{1em}
    
    \SetKwFunction{FInsertPoint}{InsertPoint}
    \Procedure{\FInsertPoint{$p^+$}} {
        Add $p^+$ to $P$
        
        $S_\text{nz} = \{c \in S \mid c \text{ is a non-zombie center}\}$
        \vspace{0.5em}
        
        \If{$\dist(p^+, S_\text{nz}) \leq r^{(\delta)}$} { \label{algline:ifnonzombie}
            Find a non-zombie center $c_j$ with index $j$ in  $S = \{c_1, \dots, c_k\}$ such that: $\dist(p^+, c_j) \leq r^{(\delta)}$
            
            \vspace{0.5em}
            Add $p^+$ to $C_j$
            
            \FDecreasingOp{}

            \vspace{0.5em}
            \textbf{exit} from \FInsertPoint{}
        }

        \vspace{0.5em}
        \tcp{In the following it holds that $\dist(p^+, S_\text{nz}) > r^{(\delta)}$}
        \vspace{0.5em}
        
        \If(\tcp*[h]{case C3}){$\min_{x, y \in S,\,x \neq y} \dist(x, y) \leq r^{(\delta)}$} { \label{algline:ci_cj_leq_rdelta_fully}
            Find $c_s, c_i \in S$ such that:
            $\dist(c_s, c_i) \leq r^{(\delta)}$,
            $s$ is the min possible index in $S = \{c_1, \dots, c_k\}$,
            and $s < i$ \label{algline:determ_i_fully}
            \vspace{0.5em}
            
            Add all points from $C_i$ to $C_s$
    
            $c_s$ becomes extended center and $C_s$ becomes extended cluster \label{algline:cs_extend_1}
    
            Remove $c_i$ from $S$ \label{algline:rem_from_S_1}
            \vspace{0.5em}  
            
            \If(\tcp*[h]{case C3a}){$\dist(p^+, S) \leq r^{(\delta)}$} {
                Find a zombie center $c_z$ with index $z$ in $S = \{c_1, \dots, c_k\}$ such that: $\dist(p^+, c_z) = \dist(p^+, S)$ \label{algline:find_cz}
    
                \vspace{0.5em}
                $c_i \gets c_z$
                
                $C_i \gets \{q \in C_z \cup \{p^+\} \mid \dist(q, c_i) \leq r^{(\delta)}\}$
    
                $c_i$ becomes regular center and $C_i$ becomes regular cluster \label{algline:C3a_reg}
                
                $C_z \gets C_z \setminus C_i$ \label{algline:Clz_updated}
                \vspace{0.5em}
                
                \FReplace{$C_z, z$} \label{algline:call_replac}
            }
            \vspace{0.5em}
            \Else(\tcp*[h]{case C3b}){  
                Add $p^+$ to $S$ \label{algline:add_to_S_1}
        
                $c_i \gets p^+$, and $c_i$ becomes regular center
            
                $C_i \gets \{c_i\}$, and $C_i$ becomes regular cluster
            } 
        }
        \vspace{0.5em}
        
        \Else(\tcp*[h]{case C4}) {
            \If{$\max_{p \in P} \dist(p, S) > r^{(\delta)}$} {
                \FIncreasingOp{}
            }
        
            \vspace{0.5em}
            Let $p \in P$ be a point of maximum distance from the set $S$ \label{algline:p_max_C4}
            \vspace{0.5em}
            
            \If{$\dist(p, S) > r^{(\delta)}$} { \label{algline:C3b_after_C4}
                Find $c_s, c_i \in S$ such that:
                $\dist(c_s, c_i) \leq r^{(\delta)}$,
                $s$ is the min possible index in $S$,
                and $s < i$ \label{algline:determ_i_fully_1}
                \vspace{0.5em}
                
                \tcp{case C3b after case C4} \label{algline:c3_aft_c4_beg}
                
                Add $p$ to $S$ and remove $c_i$ from $S$ \label{algline:update_S_1}
    
                Add all points from $C_i$ to $C_s$
    
                $c_s$ becomes extended center and $C_s$ becomes extended cluster \label{algline:cs_extend_2}
        
                $c_i \gets p$, and $c_i$ becomes regular center
            
                $C_i \gets \{p\}$, and $C_i$ becomes regular cluster \label{algline:c3_aft_c4_end}
            } 
            \vspace{0.5em}
            \Else{
                Find a center $c_j$ with index $j$ in $S = \{c_1, \dots, c_k\}$ such that: $\dist(p^+, c_j) = \dist(p^+, S)$

                Add $p^+$ to $C_j$
            }
        }

        \vspace{0.5em}
        \FDecreasingOp{}
    }
\end{algorithm}

\subsection{Analysis of the Fully Dynamic Algorithm} 
Our goal in this section is to prove Theorem~\ref{th:fully_consist} by analyzing
the fully dynamic Algorithm~\ref{alg:fully_consist}. Throughout the analysis, we use the following notation during a point update (i.e., a point insertion or a point deletion).
Let $S_\text{old}$ be the set of centers before the algorithm has processed the update,
and $S$ be the set of centers after the algorithm has processed the update.
Similarly, let $\delta_\text{old}$ be the value of the extension level before the algorithm has processed the update,
and $\delta$ be the value of the extension level after the algorithm has processed the update.

The proof of Theorem~\ref{th:fully_consist} almost immediately follows from the three invariants. Therefore
in the following we prove that after a point update,\footnote{Whenever we say “after a point update”, we mean after the algorithm has completed all updates to its data structures.} 
all the three invariants are satisfied in our fully dynamic Algorithm~\ref{alg:fully_consist}.

\subsubsection{Proof of Invariants $1$ and $2$}
\begin{lemma} \label{lem:size_S_worst_recourse_fully} 
    After a point update, the size of $S$ is $k$ (i.e., \ref{inv_1_fully} is satisfied). Moreover, the worst-case recourse per update is $1$.
\end{lemma}
\begin{proof}
    The proof follows an induction argument. In the preprocessing phase, the algorithm constructs a set $S$ of $k$ centers. We analyze the two types of point updates separately.
    
    Under a point deletion, the set of centers is modified if and only if the deleted
    point $p$ is a center point. In this case, the algorithm removes $p$ from the set of centers in Line~\ref{algline:rem_from_S_2}. Then, due to either case C1 in Line~\ref{algline:add_to_S_2},
    subcase C2a in Line~\ref{algline:add_to_S_3}, or subcase C2b in Line~\ref{algline:add_to_S_4},
    one point is added to the set of centers. Thus under point deletions, the size of the set of centers remains the same (i.e., $|S| = k$), and the worst-case recourse in a point deletion is at most $1$.
    
    Under a point insertion, the set of centers is modified only due to either case C3 or case C4. During case C3, one point is removed from the set of centers in Line~\ref{algline:rem_from_S_1}, and
    one point is added to the set of centers due to either subcase C3a, or
    subcase C3b in Line~\ref{algline:add_to_S_1}. Observe that in subcase C3a, the algorithm calls
    the Replace() function in Line~\ref{algline:call_replac} and proceeds accordingly either with case C1 or case C2. Hence as we argued before in the point deletion, one point is added to the set of centers in these cases. Finally during case C4, one point is removed from the set of centers and one point is added to the set of centers in Line~\ref{algline:update_S_1}. Therefore after a point update, the size of the set of centers
    remains the same (i.e., $|S| = k$), and the worst-case recourse in a point update is at most $1$.

    Notice that from the previous cases, only one case can occur per point update, which justifies our assertions regarding
    the worst-case recourse.
\end{proof}

\begin{observation} \label{obs:decrOper}
    Let $\delta_1$ and $\delta_2$ be the value of the extension level before and after the execution of the Decreasing Operation respectively.
    Then it holds that $\delta_2 \leq \delta_1$, and in turn for the associated radii it holds that $r^{(\delta_2)} \leq r^{(\delta_1)}$.
    Furthermore during the Decreasing Operation, a center point either becomes a regular center or its state is not affected.
\end{observation}
\begin{proof}
    During the Decreasing Operation, the value of the extension level is updated
    in Line~\ref{algline:decr_radius_2_fully} only if the set $S$ is a \distDS{$r^{(\delta_1)}$}
    as indicated in Line~\ref{algline:decr_radius_1_fully}. Hence, as $\delta_1$ is a candidate in Line~\ref{algline:decr_radius_2_fully} and the algorithm finds the smallest such $\delta'$, it follows that $\delta_2 \leq \delta_1$ and by construction $r^{(\delta_2)} \leq r^{(\delta_1)}$.

    Finally, observe that a cluster and its center can potentially become regular either in Line~\ref{algline:cluster_reg_1_fully} or in Line~\ref{algline:cluster_reg_2_fully}. If this does not occur, then their state remains unchanged.
\end{proof}

\begin{lemma} \label{lem:invar_2_satisf_fully}
    After a point update, there exists a point $p_S \in P \setminus S$ such that the set $S \cup \{p_S\}$ is a \distIS{$r^{(\delta-1)}$} (i.e., \ref{inv_2_fully} is satisfied). Moreover, the value of the radius $r^{(\delta)}$ is at most $10R^*$.
\end{lemma}
\begin{proof}
    We first prove by induction on the number of point updates that \ref{inv_2_fully} is satisfied. For the base case,
    the statement holds by the preprocessing phase. For the induction step, we analyze separately the two possible types of point updates.

    \paragraph{Point deletion.} 
    Let $p$ be the deleted point from the point set $P$, and let $P$ be the updated point set without the deleted point $p$. If $p$ is not a center point (i.e., $p \notin S_\text{old}$), then the algorithm
    only calls the Decreasing Operation. In turn, we have that $S = S_\text{old}$ and
    by Observation~\ref{obs:decrOper} it holds that $\delta \leq \delta_\text{old}$ and $r^{(\delta-1)} \leq r^{(\delta_\text{old}-1)}$.
    Since by induction hypothesis the set $S_\text{old}$ is a \distIS{$r^{(\delta_\text{old}-1)}$}, we can conclude that the
    set $S$ is a \distIS{$r^{(\delta-1)}$}. Moreover, by the way the Decreasing Operation works in Lines~\ref{algline:decr_radius_1_fully}-\ref{algline:decr_radius_2_fully}, there must exist a point $p_S \in P \setminus S$ such that the updated 
    set $S \cup \{p_S\}$ is a \distIS{$r^{(\delta-1)}$}, as needed. Otherwise
    assuming that the deleted point $p$ is the $i$-th center point (i.e., $p \in S_\text{old}$ and $p = c_i$), we analyze the two possible cases of
    the algorithm.

    \begin{enumerate}
        \item Assume that the algorithm proceeds with case C1, and let $c'$ be the corresponding new zombie center of the zombie cluster $C_i$. Then the distance between the point $c'$ and the set $S_\text{old} \setminus \{p\}$ is greater than $r^{(\delta_\text{old})}$, and by induction hypothesis 
        the pairwise distances between points in $S_\text{old}$ are strictly greater than $r^{(\delta_\text{old} - 1)}$. 
        By construction we have that $S = (S_\text{old} \setminus \{p\}) \cup \{c'\}$, and since the algorithm
        calls the Decreasing Operation in the end, it holds that $r^{(\delta - 1)} \leq r^{(\delta_\text{old} - 1)}$
        by Observation~\ref{obs:decrOper}. Hence, the pairwise distances between points in $S$ are strictly greater than $r^{(\delta - 1)}$. Moreover by the way the Decreasing Operation works in Lines~\ref{algline:decr_radius_1_fully}-\ref{algline:decr_radius_2_fully}, 
        there must exist a point $p_S \in P \setminus S$ at a distance greater than 
        $r^{(\delta - 1)}$ from the updated set $S$, and so the claim follows.
            
        \item If the algorithm proceeds with case C2, then we analyze the two subcases separately:
        \begin{enumerate}
            \item Assume that the algorithm proceeds with subcase C2a.
            Since there is a sequence satisfying~(\ref{eq:seq_for_replac_fully}) it holds that $\dist(p_{t_l}, S_\text{old} \setminus \{p\}) > r^{(\delta_\text{old})}$ (see also Line~\ref{algline:ptl_S}), and by induction hypothesis the 
            pairwise distances between points in $S_\text{old}$ are strictly greater than $r^{(\delta_\text{old} - 1)}$. By construction
            we have that $S = (S_\text{old} \setminus \{p\}) \cup \{p_{t_l}\}$, and by Observation~\ref{obs:decrOper} 
            it holds that $r^{(\delta - 1)} \leq r^{(\delta_\text{old} - 1)}$. In turn,
            it follows that the pairwise distances between points in $S$ are strictly greater than $r^{(\delta - 1)}$. Eventually, by
            the way the Decreasing Operation works in Lines~\ref{algline:decr_radius_1_fully}-\ref{algline:decr_radius_2_fully}, there must exist a
            point $p_S \in P \setminus S$ at a distance greater than $r^{(\delta - 1)}$ from the updated set $S$,
            as needed.
            
            \item Assume that the algorithm proceeds with case C2b, and let $c'$ be the corresponding new center of maximum
            distance from the set $S_\text{old} \setminus \{p\}$. By construction we have that 
            $S = (S_\text{old} \setminus \{p\}) \cup \{c'\}$. Also due to the Decreasing Operation
            we have that $r^{(\delta-1)} < \max_{q \in P} \dist(q, S)$, and since $S_\text{old} \setminus \{p\} \subseteq S$ 
            it holds that $\max_{q \in P} \dist(q, S) \leq \max_{q \in P} \dist(q, S_\text{old} \setminus \{p\})$. Therefore
            it follows that: 
            \begin{equation*} 
                r^{(\delta-1)} \;<\; \max_{q \in P} \dist(q, S_\text{old} \setminus \{p\}) \;=\; \dist(c', S_\text{old} \setminus \{p\}).
            \end{equation*}
            Together with the induction hypothesis and Observation~\ref{obs:decrOper}, we can infer
            that the set $S$ is a \distIS{$r^{(\delta-1)}$}. Finally, once again by the way the Decreasing Operation works in Lines~\ref{algline:decr_radius_1_fully}-\ref{algline:decr_radius_2_fully}, there must exist a
            point $p_S \in P \setminus S$ such that the updated set $S \cup \{p_S\}$ is a \distIS{$r^{(\delta-1)}$},
            and so \ref{inv_2_fully} is satisfied.
            \end{enumerate}
    \end{enumerate}
    
    \paragraph{Point insertion.} 
    In the rest of the proof, we use the following observation.
    
    \begin{observation} \label{obs:incrOper}
        During the Increasing Operation the algorithm increases the extension level to some $\delta'$ where $\delta_\text{old} < \delta' \leq \delta$, only when the set $S_\text{old}$ is a \distIS{$r^{(\delta'-1)}$} and the set $S_\text{old}$ is not a \distDS{$r^{(\delta'-1)}$}.
    \end{observation}
    
    Let $p^+$ be the inserted point into the point set $P$, and let $P$ be the updated point set including the inserted point~$p^+$. If $p^+$ is within distance
    $r^{(\delta_\text{old})}$ from a non-zombie center point of $S_\text{old}$, then 
    the algorithm only calls the Decreasing Operation. In turn, we have that $S = S_\text{old}$ and by Observation~\ref{obs:decrOper} it holds that
    $\delta \leq \delta_\text{old}$ and $r^{(\delta-1)} \leq r^{(\delta_\text{old}-1)}$.
    By the induction hypothesis there exists a point $p_S \in P \setminus S_\text{old}$ such that the set $S_\text{old} \cup \{p_S\}$ is a \distIS{$r^{(\delta_\text{old}-1)}$}. Since the point $p_S$ is still part of the point set $P$, 
    we can conclude that $p_S \in P \setminus S$ and that the set $S \cup \{p_S\}$
    is a \distIS{$r^{(\delta-1)}$}, as needed.
    Otherwise for the remaining cases, we analyze the algorithm based on whether the set of centers is modified or not.

    \begin{enumerate}
    \item If the set of centers is not modified (i.e., $S = S_\text{old}$), by construction the algorithm proceeds with case C4.
    If all the points are within distance $r^{(\delta_\text{old})}$ from the set $S$, then the algorithm just calls in the end the Decreasing Operation,
    and the claim follows as before using the induction hypothesis and Observation~\ref{obs:decrOper}. Otherwise if there exists a point at a distance greater than $r^{(\delta_\text{old})}$ from the set $S$, then the algorithm calls the Increasing Operation. Based on Observation~\ref{obs:incrOper}, there must exist a point $p_S \in P \setminus S$ at a distance greater than $r^{(\delta-1)}$ from the set $S$ and the set $S$ must be a \distIS{$r^{(\delta-1)}$}.
    In the end the algorithm calls the Decreasing Operation, and using Observation~\ref{obs:decrOper} we can conclude that \ref{inv_2_fully} is satisfied.
    
    \item Otherwise assuming that the set of centers is modified, we analyze case C3 and case C4 of the algorithm separately:
    \begin{enumerate} 
    \item Assume that the algorithm proceeds with case C3 in Line~\ref{algline:ci_cj_leq_rdelta_fully}, where
    an older center point $c_i$ is exempted from being a center. To avoid confusion with the reset of $c_i$,
    let $c_i'$ be the old point $c_i$ determined in Line~\ref{algline:determ_i_fully}. By construction the algorithm
    removes $c_i'$ from the set of centers and adds a point $c'$ to the set of centers, which means that $S = (S_\text{old} \setminus \{c_i'\}) \cup \{c'\}$. Also by construction of the algorithm and Observation~\ref{obs:decrOper}, we have that $\delta \leq \delta_\text{old}$ and $r^{(\delta-1)} \leq r^{(\delta_\text{old}-1)}$. We split the analysis in two subcases:
    \begin{itemize}
        \item If the algorithm proceeds with case C3a, the updated center point $c_i$ becomes the old $z$-th center point determined
        in Line~\ref{algline:find_cz}, the corresponding zombie cluster $C_z$ is updated, and then the algorithm calls the function Replace() in
        Line~\ref{algline:call_replac}. In turn, the algorithm proceeds accordingly either with case C1 or with case C2, aiming to locate the new $z$-th center point. To that end, the claim follows by using the same arguments
        as before in the analysis of cases C1 and C2 in the point deletion, by replacing in the proof the (zombie cluster) $C_i$ with $C_z$ and the (deleted) point $p$ with $c_i'$.

        \item If the algorithm proceeds with subcase C3b, then by construction the point $c'$ is actually the new inserted point~$p^+$
        and it holds that $\dist(c', S_\text{old} \setminus \{c_i'\}) > r^{(\delta_\text{old})}$.
        By the induction hypothesis the set $S_\text{old}$ is a \distIS{$r^{(\delta_\text{old}-1)}$}, and as $S = (S_\text{old} \setminus \{c_i'\}) \cup \{c'\}$
        it follows that the set $S$ is a  \distIS{$r^{(\delta-1)}$}. Moreover the algorithm calls the Decreasing Operation in the end, and by the way it works in Lines~\ref{algline:decr_radius_1_fully}-\ref{algline:decr_radius_2_fully} there must exist a point $p_S \in P \setminus S$ such that the updated set $S \cup \{p_S\}$ is a \distIS{$r^{(\delta-1)}$},
        and so \ref{inv_2_fully} is satisfied.
    \end{itemize}
    
    \item Alternatively assume that the algorithm proceeds with case C4 and after with case C3b. In this situation, the Increasing Operation is called and Line~\ref{algline:update_S_1} is executed. Based on Observation~\ref{obs:incrOper}, the set $S_\text{old}$ is a \distIS{$r^{(\delta-1)}$}.
    Also by construction we have that $S = (S_\text{old} \setminus \{c_i'\}) \cup \{p\}$,
    where $c_i'$ is the old point $c_i$ determined in Line~\ref{algline:determ_i_fully_1}, and $p$ is a point in $P$ determined in Line~\ref{algline:p_max_C4} such that
    $\dist(p, S_\text{old}) > r^{(\delta)}$. Therefore by setting $p_S = c_i'$ and by Observation~\ref{obs:decrOper}, we
    can conclude that \ref{inv_2_fully} is satisfied.
    \end{enumerate}
    \end{enumerate}

    \paragraph{Upper bound on $r^{(\delta)}.$}
    Since \ref{inv_2_fully} is satisfied, there exists a point $p_S \in P \setminus S$ such that the set $S \cup \{p_S\}$ is a \distIS{$r^{(\delta-1)}$}.
    Based on Lemma~\ref{lem:size_S_worst_recourse_fully} the size of the set $S \cup \{p_S\}$ is $k + 1$, and using
    Corollary~\ref{cor:r_less2R} we get that $r^{(\delta-1)} < 2R^*$.
    Therefore as we have that $r^{(\delta)} = 5r^{(\delta-1)}$ by construction, the claim follows. 
\end{proof}

\subsubsection{Auxiliary Results}
We continue with the proofs of some useful results which are utilized for the proof of \ref{inv_3_fully}.
Essentially, we show that whenever a cluster $C$ becomes an extended cluster then it cannot be already a zombie cluster, and that $C$ can become an extended cluster only because of another regular cluster. In turn, these two factors along with the fact that no additional points
enter a zombie cluster while it remains a zombie cluster\footnote{The only way a new point can enter a zombie cluster is if it is itself a center point. This case can be safely ignored, as the distance from a center to itself is zero. \label{footnote:zmb_new_points}}, allow us to upper bound the approximation ratio. 
Note though that an extended cluster can potentially become a zombie cluster.

\antonis{This Lemma is new.}
\begin{lemma} \label{lem:zomb_no_new_point}
    After a point update, consider a zombie cluster $C_z$ and let $C_z^\old$ be $C_z$ before the update. Then it holds that $C_z \setminus \{c_z\} \subseteq C_z^\old$, where $c_z \in S_\text{zmb}$ is the zombie center of $C_z$. In other words, no additional points enter a zombie cluster besides its zombie center.
\end{lemma}
\begin{proof}
    Notice that a zombie cluster receives new points only due to case C2a \antonis{right?}, where each new point added must be its updated zombie center. Indeed, during the Decreasing Operation a zombie cluster cannot receive new points. In fact, it may lose some points due to~\cref{algline:switch_cluster_1_fully}, which are then switched to a non-zombie cluster. Also during the Reassigning Operation, any cluster that receives new points in~\cref{algline:switch_cluster_2_fully} becomes a regular cluster. Finally note that during case C3a, a subset of a zombie cluster may become a regular cluster. Afterwards, since the algorithm proceeds with either case C1 or C2, the zombie cluster can only receive a new point if it is its zombie center.
\end{proof}

\begin{lemma} \label{lem:all_zomb_more_r}
    After a point update, all the zombie centers are at a distance greater than $r^{(\delta)}$ from the rest of the centers.
    In other words, for every zombie center point $c_z \in S_\text{zmb}$ it holds that $\dist(c_z, S \setminus \{c_z\}) > r^{(\delta)}$.
\end{lemma}
\begin{proof}
    The proof is by induction on the number of point updates. For the base case, the statement holds trivially since after the preprocessing phase all the center points are regular centers. For the induction step, we analyze separately the two possible types of point updates.
    
    \paragraph{Point deletion.}
    Let $p$ be the deleted point from the point set $P$, and let $P$ be the updated point set without the deleted point $p$. If $p$ is not a center point (i.e., $p \notin S_\text{old}$), then the algorithm
    just calls the Decreasing Operation. In turn, we have that $S = S_\text{old}$ and
    by Observation~\ref{obs:decrOper} it holds that $\delta \leq \delta_\text{old}$ and $r^{(\delta)} \leq r^{(\delta_\text{old})}$.
    Together with the induction hypothesis, we can infer that the center points which were zombie centers before the update, satisfy the statement after the point update as well.
    Also based on Observation~\ref{obs:decrOper}, the Decreasing Operation does not convert any center point to a zombie center, and so the claim holds for this case.
    Otherwise assuming that the deleted point $p$ is the $i$-th center point (i.e., $p \in S_\text{old}$ and $p = c_i$), 
    we analyze the following cases of the algorithm:

    \begin{enumerate}
        \item Assume that the algorithm proceeds either with case C1 or with case C2a. If the algorithm proceeds with case C1, let $\overline{c}$ be the corresponding new zombie center $c'$ of the zombie cluster $C_i$. Otherwise if the algorithm proceeds with case C2a, let $\overline{c}$ be the corresponding
        point $p_{t_l}$ from the sequence that satisfies~(\ref{eq:seq_for_replac_fully}).
        By construction, it holds that $\dist(\overline{c}, S_\text{old} \setminus \{p\}) > r^{(\delta_\text{old})}$ 
        (see also Line~\ref{algline:caseC1_fully} for C1, and Line~\ref{algline:ptl_S} for C2a), and also we have that $S = (S_\text{old} \setminus \{p\}) \cup \{\overline{c}\}$.
        The algorithm calls the Decreasing Operation and by Observation~\ref{obs:decrOper} it holds that $r^{(\delta)} \leq 
        r^{(\delta_\text{old})}$. By the induction hypothesis, every center point $c_z$ which was a zombie center in $S_\text{old}$ before the update is at a distance greater than $r^{(\delta_\text{old})}$ from the set $S_\text{old} \setminus \{c_z\}$. Hence every center point $c_z$ which was a zombie center in $S_\text{old}$ 
        before the update is at a distance greater than $r^{(\delta)}$ from the set $S \setminus \{c_z\}$ as well.
        Moreover, notice that $\overline{c}$ is the only point in the metric space which can be converted from a non-zombie center to a zombie center after the update. Finally observe that during case C2a, the shifting of indices in Line~\ref{algline:shift_centers} is along zombie centers (see the first condition of~(\ref{eq:seq_for_replac_fully})), and thus the claim follows.
            
        \item Assume that the algorithm proceeds with case C2b, and let $c'$ be the corresponding new center of maximum distance from the set $S_\text{old} \setminus \{p\}$. By construction we have that 
        $S = (S_\text{old} \setminus \{p\}) \cup \{c'\}$, and observe that
        the Reassigning Operation does not convert any center point to a zombie center.
        If the distance between the point $c'$ and the set $S_\text{old} \setminus \{p\}$
        is greater than $r^{(\delta_\text{old})}$, then the claim follows by using the same arguments as before.
        Otherwise assume that $\dist(c', S_\text{old} \setminus \{p\}) \leq r^{(\delta_\text{old})}$. By construction we have that
        $S_\text{old} \setminus \{p\} \subseteq S$, which implies that:
        \begin{equation*} 
           \max_{q \in P} \dist(q, S) \;\leq\; \max_{q \in P} \dist(q, S_\text{old} \setminus \{p\}) \;=\; \dist(c', S_\text{old} \setminus \{p\}) \;\leq\; r^{(\delta_\text{old})}.
        \end{equation*}
        Therefore due to the Decreasing Operation in Lines~\ref{algline:decr_radius_1_fully}-\ref{algline:cluster_reg_1_fully}, all the center points are converted to regular centers, and thus the claim holds trivially.
        \end{enumerate}
        
    \paragraph{Point insertion.}
    Let $p^+$ be the inserted point into the point set $P$, and let $P$ be the updated point set including the inserted point $p^+$.
    If $p^+$ is within distance $r^{(\delta_\text{old})}$ from a non-zombie center point of $S_\text{old}$, 
    then the algorithm just calls the Decreasing Operation. In turn, we have that $S = S_\text{old}$ and by Observation~\ref{obs:decrOper}
    it holds that $\delta \leq \delta_\text{old}$ and $r^{(\delta)} \leq r^{(\delta_\text{old})}$. Together with the induction hypothesis, we can infer that the center points which were zombie centers before the update, satisfy the statement after the point update as well.
    Also observe that the Decreasing Operation does not convert any center point to a zombie center, and so the claim holds for this case.
    
    Otherwise assume that the new inserted point $p^+$ is at a distance greater than $r^{(\delta_\text{old})}$ from all the non-zombie center points of $S_\text{old}$. In this situation, the algorithm proceeds either with case C3 or with case C4, and we analyze the two cases separately:

    \begin{enumerate}
    \item Assume that the algorithm proceeds with case C4. If the algorithm does not call the Increasing Operation, then all the points must be within distance $r^{(\delta_\text{old})}$ from the set $S_\text{old}$, and
    the set of centers is not modified (i.e., $S = S_\text{old}$). In the end, the algorithm calls the Decreasing Operation which in Lines~\ref{algline:decr_radius_1_fully}-\ref{algline:cluster_reg_1_fully} converts all the center points to regular centers, and so the claim holds trivially.
    Otherwise if the algorithm calls the Increasing Operation, observe that all the center points are converted to regular centers in Line~\ref{algline:cluster_reg_3_fully}. If the algorithm proceeds further
    with case C3b after case C4 in Line~\ref{algline:C3b_after_C4}, notice that no center point is converted
    to a zombie center. In the end, the algorithm calls the Decreasing Operation which
    does not convert any center point to a zombie center based on Observation~\ref{obs:decrOper}, and so the claim holds again trivially.
    
    \item Assume that the algorithm proceeds with case C3 in Line~\ref{algline:ci_cj_leq_rdelta_fully}, where an older center point $c_i$ is exempted from being a center, and a center point $c_s$
    is converted to extended center in Line~\ref{algline:cs_extend_1}. 
    To avoid confusion with the reset of $c_i$,
    let $c_i'$ be the old point $c_i$ determined in Line~\ref{algline:determ_i_fully}. By construction the algorithm
    removes the point $c_i'$ from the set of centers and adds a point $c'$ to the set of centers, which means that $S = (S_\text{old} \setminus \{c_i'\}) \cup \{c'\}$. Also by construction of the algorithm and Observation~\ref{obs:decrOper}, we have that $\delta \leq \delta_\text{old}$ and $r^{(\delta)} \leq r^{(\delta_\text{old})}$. We split the analysis in two subcases:
    \begin{itemize}
        \item If the algorithm proceeds with case C3a, the updated center point $c_i$ becomes a regular center in Line~\ref{algline:C3a_reg},
        the corresponding zombie cluster $C_z$ is updated, and then the algorithm calls the function Replace() in Line~\ref{algline:call_replac}. 
        In turn, the algorithm proceeds accordingly either with case C1 or with case C2, with the goal to locate the new $z$-th center point. To that end, the claim follows by using the same arguments
        as before in the analysis of the point deletion, by replacing in the proof the (zombie cluster) $C_i$ with $C_z$ and the (deleted) point $p$ with $c_i'$.
        
        \item If the algorithm proceeds with subcase C3b, then by construction the 
        point $c'$ is actually the new inserted point $p^+$ and it holds that
        $\dist(c', S_\text{old} \setminus \{c_i'\}) > r^{(\delta_\text{old})}$.
        By the induction hypothesis, every center point $c_z$ which was a zombie center in $S_\text{old}$ 
        before the update is at a distance greater than $r^{(\delta_\text{old})}$ from the set $S_\text{old} \setminus \{c_z\}$. Hence every center point $c_z$ which was a zombie center in $S_\text{old}$ 
        before the update is at a distance greater than $r^{(\delta)}$ from the set $S \setminus \{c_z\}$ as well. 
        Note that the new point $c'$ becomes a regular center after the update. Therefore by construction and  Observation~\ref{obs:decrOper}, there is no center point which is converted from a non-zombie to a zombie center,
        and this concludes the claim.
    \end{itemize}
    \end{enumerate}
\end{proof}

\antonis{This is new. Check for correctness.}
\begin{observation} \label{obs:rad_change_reg}
    Whenever the value of the extension level changes (and in turn the value of the associated radius),  all centers become regular.
\end{observation}
\begin{proof}
    The extension level can change only due to~\cref{algline:incr_ext_lvl} during the Increasing Operation or due to~\cref{algline:decr_radius_2_fully} during the Decreasing Operation. In both cases, all center points become regular in~\cref{algline:cluster_reg_3_fully} and~\cref{algline:cluster_reg_1_fully} respectively.
\end{proof}

\antonis{Updated proof. Check for correctness.}
\begin{lemma} \label{lem:rem_ci_reg} 
    After a point insertion, assume that the fully dynamic algorithm removes a center point $c_i$ from the set of centers.
    Then the point $c_i$ must have been a regular center before the algorithm processes the point insertion.
\end{lemma}
\begin{proof}
    To avoid confusion with the reset of $c_i$, let $c_i^\old$ be the old point $c_i$ before the algorithm processes the point insertion. With this notation, we have to prove that the point $c_i^\old$ must have been a regular center before the point insertion. Notice that the point $c_i^\old$ must satisfy (\ref{choose_cs_ci_fully}) and is determined either in Line~\ref{algline:determ_i_fully} or in Line~\ref{algline:determ_i_fully_1}. Hence, it holds
    that $\dist(c_i^\old, S_\text{old} \setminus \{c_i^\old\}) \leq r^{(\delta_\text{old})}$. 
    Let $c_s$ be the corresponding center point which is determined either in Line~\ref{algline:determ_i_fully} or in Line~\ref{algline:determ_i_fully_1} and it is responsible for removing $c_i^\old$ from the set of centers. Since $c_s, c_i^\old$ satisfy~(\ref{choose_cs_ci_fully}), we have $s < i$ and $\dist(c_s, c_i^\old) \leq r^{(\delta_\text{old})}$. Based on Lemma~\ref{lem:all_zomb_more_r}, the point $c_i^\old$ could not have been a zombie center before the point insertion, and so it must have been either a regular or an extended center.

    Suppose to the contrary that the point $c_i^\old$ is an extended center when it is removed from the set of centers, and let $\tau$ be the most recent update at which $c_i^\old$ becomes an extended center. Observe that the only way for a center to become an extended center is either in Line~\ref{algline:cs_extend_1} or in Line~\ref{algline:cs_extend_2}, and let $\delta_\tau$ be the value of the extension level at the $\tau$-th update.
    
    \begin{claim} \label{clm:ext_lvl_same}\antonis{This is new, please check!}
        Based on the assumptions made so far, it holds that $\delta_\tau = \delta_\text{old}$.
    \end{claim}
    \begin{proof}
        By the definition of $\tau$, the point $c_i^\old$ remains an extended center from the $\tau$-update until just before the current point insertion.
        Suppose to the contrary that $\delta_\tau \neq \delta_\text{old}$, due to an earlier update. 
        Based on~\cref{obs:rad_change_reg}, the old center point $c_i^\old$ must have become regular at the first update at which the extension level differs from $\delta_\tau$. However, this contradicts the fact that $c_i^\old$ remains extended from the $\tau$-update until just before the current point insertion, implying that the extension level must have remained the same throughout this interval.
    \end{proof}

    Therefore by~\cref{clm:ext_lvl_same}, it follows that $r^{(\delta_\tau)} = r^{(\delta_\text{old})}$.
    Note that the point $c_s$ must have become a center before the $\tau$-th update, as otherwise due to~\cref{algline:ifnonzombie}, the point $c_s$ could not have been a center at the current update.
    Finally, we can conclude the following:
    \begin{itemize}
        \item Both points $c_s$ and $c_i^\old$ belong to the set of centers at the $\tau$-th update.
        \item $\dist(c_s, c_i^\old) \leq r^{(\delta_\tau)}$.
        \item $s < i$.
    \end{itemize}
    Nevertheless, this contradicts the condition that $i$ is the minimum index satisfying~(\ref{choose_cs_ci_fully}) at the $\tau$-th update. As a result, the point $c_i^\old$ cannot be an extended center before the current update, and thus $c_i^\old$ must have been a regular center.
\end{proof}

\subsubsection{Proof of \ref{inv_3_fully}}
Consider a point update sent by the adaptive adversary.
Let $S^{(\text{nz})}$ denote the ordered set of the most recent non-zombie centers after the fully dynamic algorithm has processed the point update. Similarly,
let $S^{(\text{nz})}_\text{old}$ denote the ordered set of the most recent non-zombie centers
before the point update (see also Figure~\ref{fig:example_nz}).
In other words, the $i$-th point in $S^{(\text{nz})}$ is
the last non-zombie center of cluster $C_i$ in the metric space.

Note that the set $S^{(\text{nz})}$ may be different than the set of non-zombie centers in $S$ (i.e., $S^{(\text{nz})} \neq S_\text{nz}$).
Similarly, the set $S^{(\text{nz})}_\text{old}$ may be different than the set of non-zombie center in $S_\text{old}$.
Furthermore, the sets $S^{(\text{nz})}_\text{old}, S^{(\text{nz})}$ may even contain points in the metric space that are no longer part of the point set $P$ due to prior point deletions. We remark that the sets $S^{(\text{nz})}_\text{old}, S^{(\text{nz})}$ are auxiliary sets which are used in the rest of the analysis. \antonis{Should we prove in general that a point belongs to exactly one cluster?}

\begin{figure}
    \centering
    \begin{subfigure}{.33\textwidth}
      \centering
      \includegraphics[width=0.4\linewidth]{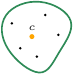}
    \end{subfigure}%
    \begin{subfigure}{.33\textwidth}
      \centering
      \includegraphics[width=0.4\linewidth]{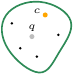}
    \end{subfigure}
    \begin{subfigure}{.33\textwidth}
      \centering
      \includegraphics[width=0.4\linewidth]{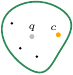}
    \end{subfigure}
     \caption{The cluster is initially regular or extended, and after the point deletion becomes a zombie cluster. \\
     \underline{Left figure}: The center $c$ (which is initially either regular or extended) is also the most recent non-zombie center. The adversary deletes $c$. \\
     \underline{Middle figure}: The zombie center $c$ is a new point, and now $q$ is the most recent non-zombie center. Note that $q$ is not part of the point set $P$ anymore. The adversary deletes $c$. \\
     \underline{Right figure}: The new zombie center $c$ is a new point, and $q$ is still the most recent non-zombie center.
    \label{fig:example_nz}}
\end{figure}

\begin{observation} \label{obs:DecrOperProp}
    Whenever the fully dynamic algorithm decreases the value of the extension level (i.e., $\delta < \delta_\text{old}$) during 
    the Decreasing Operation, the current set of centers $S$ is a \distDS{$r^{(\delta)}$} for 
    the updated value $\delta$ of the extension level, and also it holds that $S^{(\text{nz})} = S_\text{reg} = S$. In particular, every point $q \in P$ is within distance $r^{(\delta)}$ from the center point $c_j$,
    where $c_j \in S$ is the regular center of the regular cluster $C_j$ that contains the point $q$.
\end{observation}

The next lemma proves that \ref{inv_3_fully} is satisfied over the course of the algorithm. The main statement of Lemma~\ref{lem:C_covers_fully} is property $P_1$, while property $P_2$ serves as an auxiliary result for the validity of $P_1$. We point out that
for the validity of both $P_1$ and $P_2$, it is important that our fully dynamic algorithm does not add new points to a zombie cluster, as stated in~\cref{lem:zomb_no_new_point}.\footref{footnote:zmb_new_points}

\begin{lemma} \label{lem:C_covers_fully}
    After a point update, consider a point $q \in P$ and let $C_j$ be the cluster that contains it.
    Recall that $c_j \in S$ is the center of the cluster $C_j$, and let $c^{(\text{nz})}_j \in S^{(\text{nz})}$ be the most recent non-zombie center of $C_j$. Then the following two properties $P_1$ and $P_2$ are satisfied:
    
    \begin{enumerate}[label=\subscript{P}{\arabic*}.]
        \item \begin{enumerate}
        \item If $c_j$ is a regular center (i.e., $c_j \in S_\text{reg}$),
        then $\dist(q, c_j) \leq r^{(\delta)}$.

        \item If $c_j$ is an extended center (i.e., $c_j \in S_\text{ext}$),
        then  $\dist(q, c_j) \leq 2r^{(\delta)}$.

        \item If $c_j$ is a zombie center (i.e., $c_j \in S_\text{zmb}$),
        then $\dist(q, c_j) \leq 5r^{(\delta)}$.
        \end{enumerate}  

        \item If $q$ is not a center (i.e., $q \in P \setminus S$), then 
        $\dist(q, c^{(\text{nz})}_j) \leq 2r^{(\delta)}$.
    \end{enumerate}
\end{lemma}

\begin{proof}
    The proof of the validity of $P_1$ and $P_2$ is by induction on the number of point updates. The base case is when the preprocessing phase has finished. Initially, every point is within distance $r^{(0)}$ from its center, all centers are regular (i.e., $S = S_\text{reg} = S^{(\text{nz})}$), and the extension level $\delta$ is zero. Hence at the beginning, both properties $P_1$ and $P_2$ are satisfied. For the induction step, we analyze the two possible types of point updates
    separately. Notice that in the following we can assume that $\delta \geq \delta_\text{old}$ and $r^{(\delta)} \geq r^{(\delta_\text{old})}$, as any decrease to the extension level implies the validity of both properties $P_1$ and $P_2$ based on Observation~\ref{obs:DecrOperProp}.

    \paragraph{Point deletion.} Let $p$ be the deleted point from the point set $P$, and let $P$ be the updated point set without the deleted point $p$. 
    During the Reassigning Operation, we assume in Line~\ref{algline:find_c_fully} that there exists such a point $c_{j'}$. 
    The next claim proves the correctness of the Reassigning Operation by arguing that this instruction is executed properly.

    \begin{claim} \label{clm:reassOper_q_atmost_r_c}
        During the Reassigning Operation, for every point $q \in Q$ there exists a center $c_{j'} \in S_\text{old} \setminus \{p\}$ such that
        $\dist(q, c_{j'}) \leq r^{(\delta_\text{old})}$ in Line~\ref{algline:find_c_fully}.
    \end{claim}
    \begin{proof}
        Notice that $p$ is the deleted point with index $i$ in $S_\text{old}$, and let $C_i'$
        be the $i$-th cluster $C_i$ in the beginning of the Reassigning Operation. Before the Reassigning Operation, the algorithm removes $p$ from its cluster and calls the Decreasing Operation. As a result, the set $C_i'$ neither contains the old center $p$ nor
        any point within distance $r^{(\delta_\text{old})}$ from a non-zombie center (see Lines~\ref{algline:switch_cluster_1_fully_if}-\ref{algline:switch_cluster_1_fully}).\footnote{Since the center
        $c_i$ is deleted at this moment, we assume that $\dist(q, c_i)$ in Line\ref{algline:switch_cluster_1_fully_if}
        returns $\infty$.} Initially the queue $Q$ is initialized to $C_i'$, and observe that the algorithm must have proceeded with case C2b in order to call the Reassigning Operation. Since the algorithm did not proceed with case C1, every point in $C_i'$ must be within distance $r^{(\delta_\text{old})}$ from some center in $S_\text{old} \setminus \{p\}$. Consequently, Line~\ref{algline:find_c_fully} is executed properly for these points in $C_i'$.
        
        Regarding the rest of the points in the queue $Q$, note that the algorithm proceeds with case C2b when there is no sequence satisfying~(\ref{eq:seq_for_replac_fully}). 
        During the Reassigning Operation, due to Lines~\ref{algline:switch_cluster_1_fully_if}-\ref{algline:switch_cluster_1_fully} of the Decreasing
        Operation, notice that whenever a point in $Q$ is within distance $r^{(\delta_\text{old})}$ from a center point $c$, then $c$ must a zombie center.
        Suppose to the contrary that at some moment during the Reassigning Operation, a point $q \in Q$ is at a distance greater than $r^{(\delta_\text{old})}$ from any other center in $S_\text{old} \setminus \{p\}$. 
        Then such a point $q$ could serve as the point $p_{t_l}$ in (\ref{eq:seq_for_replac_fully}), because all the center points along the sequence are zombie centers and also $\dist(q, S_\text{old} \setminus \{p\}) > r^{(\delta_\text{old})}$. However, this
        contradicts the fact that there is no sequence satisfying~(\ref{eq:seq_for_replac_fully}), and thus
        Line~\ref{algline:find_c_fully} is executed correctly throughout the algorithm.
    \end{proof}

    \noindent
    We analyze the validity of the two properties $P_1$ and $P_2$ under a point deletion separately. 
    
    \subparagraph{Validity of $P_2$.}
    Since after a point deletion all remaining points in $P$ belong to the point set $P_\text{old}$ before the update, the induction hypothesis can be applied to every point in $P \setminus S_\text{old}$. Notice that under point deletions, the algorithm exempts an older center point in $S_\text{old}$ from being a center only if the adversary removes it from the point set.\footnote{Recall that by ``exempt'', we mean only the removal of the center point from the set of centers, and not the shifting of indices.} Therefore, the induction hypothesis can actually be applied for all the points in $P \setminus S$ for the updated set $S$, since $P \setminus S \subseteq P_\text{old} \setminus S_\text{old}$.
    
    Consider an arbitrary point $q \in P \setminus S$, and let $C_j$ be the unique cluster that contains it before the point deletion. 
    If the point $q$ switches clusters during the Decreasing Operation in 
    Line~\ref{algline:switch_cluster_1_fully} or the Reassigning Operation in Line~\ref{algline:switch_cluster_2_fully}, then the validity of $P_2$ for the point $q$ follows by construction. Observe that a point could switch clusters also during case C2a in Line~\ref{algline:shift_centers}, but note that by construction
    these are center points and thus the claim for $P_2$ is not affected.
    Otherwise assume that the point $q$ remains part of the cluster $C_j$ after the point deletion, and let
    $c^{(\text{old})}_j \in S^{(\text{nz})}_\text{old}$ be the most recent non-zombie center of $C_j$ before the point deletion.
    We continue the analysis based on whether the point $c^{(\text{old})}_j$ is still the most recent non-zombie center of $C_j$
    after the point deletion:

    \begin{itemize}
        \item If $c^{(\text{old})}_j \in S^{(\text{nz})}$, 
        then by the induction hypothesis and as $q \in C_j$, it follows that $\dist(q, c_j^{(\text{nz})}) \leq 2r^{(\delta_\text{old})}$.
        Recall that we can assume that $r^{(\delta)} \geq r^{(\delta_\text{old})}$ based
        on Observation~\ref{obs:DecrOperProp}, and thus the point $q$ satisfies the statement of $P_2$.
        
        \item Otherwise, let $c^{(\text{new})}_j$ be the new non-zombie center of the cluster $C_j$ (i.e., $c^{(\text{new})}_j \in S^{(\text{nz})}$).
        Since the point $q$ does not switch clusters, either Line~\ref{algline:cluster_reg_1_fully} or Line~\ref{algline:cluster_reg_2_fully} is executed during the Decreasing Operation, or Line~\ref{algline:cluster_reg_reas_oper} is executed during the Reassigning Operation. Also observe that the point $q$ does not enter the queue $Q$ in~\cref{algline:add_farpnt_to_Q}, as that would cause $q$ to switch clusters.\footnote{Notice that during the Reassigning Operation, some points may switch clusters while others may not.} In turn, all the points in $C_j$ must be within distance $r^{(\delta)}$ from the point $c^{(\text{new})}_j \in S$ (which is also a regular center in the updated set $S$). Therefore as the point $q$ belongs to $C_j$, the validity of $P_2$ follows. 
    \end{itemize}

    \subparagraph{Validity of $P_1$.}
    If the deleted point $p$ is not a center point (i.e., $p \notin S_\text{old}$), then 
    the point $p$ is removed from the cluster it belongs to, and the algorithm calls the Decreasing Operation. 
    Consider an arbitrary point $q \in P$. If the point $q$ ends up in a non-zombie cluster due to either 
    Line~\ref{algline:cluster_reg_1_fully}, Line~\ref{algline:cluster_reg_2_fully}, or
    Line~\ref{algline:switch_cluster_1_fully} during the Decreasing Operation, then the point $q$ satisfies the statement of $P_1$.
    Otherwise the point $q$ remains in the same cluster whose state and center do not change, and the validity of $P_1$
    follows by the induction hypothesis and the fact that $r^{(\delta_\text{old})} \leq r^{(\delta)}$. Otherwise assuming that the deleted point $p$ is the $i$-th center point (i.e., $p \in S_\text{old}$ and $p = c_i$),
    we analyze the two possible cases of the algorithm:

    \begin{enumerate}
        \item Assume that the algorithm proceeds with case C1, and let $c'$ be the corresponding new zombie center. Also let $C_i^{(\text{old})}$ be the corresponding old cluster before the point deletion.
        By the validity of $P_2$ before the update, there is a point $c^{(\text{old})}_i \in S^{(\text{nz})}_\text{old}$ such that $\dist(q, c^{(\text{old})}_i) \leq 2r^{(\delta_\text{old})}$ for every point $q \in C_i^{(\text{old})} \setminus S_\text{old}$. Recall that the point $c^{(\text{old})}_i$ is the most recent non-zombie center of $C_i^{(\text{old})}$ before the update. Based on the way the algorithm selects the point $c'$, it holds that $c' \in C_i^{(\text{old})} \setminus S_\text{old}$, which in turn means that $\dist(c', c^{(\text{old})}_i) \leq 2r^{(\delta_\text{old})}$.
        Thus by the triangle inequality, it follows that $\dist(q, c') \leq \dist(q, c^{(\text{old})}_i) + \dist(c', c^{(\text{old})}_i) \leq 4r^{(\delta_\text{old})}$ 
        for every point $q \in C_i^{(\text{old})} \setminus S_{\text{old}}$. 
        Since the updated zombie cluster $C_i$ does not receive new points based on~\cref{lem:zomb_no_new_point}, we have $C_i \setminus S \subseteq C_i^{(\text{old})} \setminus S_\text{old}$. In turn together with the fact that $r^{(\delta_\text{old})} \leq r^{(\delta)}$, we get that
        $\dist(q, c') \leq \dist(q, c^{(\text{old})}_i) + \dist(c', c^{(\text{old})}_i) \leq 4r^{(\delta)}$ for every point $q \in C_i \setminus S$. 
        Moreover, notice that every point $q \in P$ that switches clusters during the Decreasing Operation in Line~\ref{algline:switch_cluster_1_fully}
        satisfies the statement of $P_1$ by construction. 
        
        Finally, consider an arbitrary point $q \in P \setminus C_i$ which belongs to the same cluster $C_j$ before and after the point deletion.
        By construction we have $S = (S_\text{old} \setminus \{p\}) \cup \{c'\}$, which means that the center point $c_j \in S_\text{old} \setminus \{p\}$ remains the center of the cluster $C_j$ before and after the point deletion. 
        If the center $c_j$ becomes regular during the Decreasing Operation either in Line~\ref{algline:cluster_reg_1_fully} or in Line~\ref{algline:cluster_reg_2_fully}, the point $q$ satisfies the statement of $P_1$. Otherwise by the induction hypothesis, 
        it holds that $\dist(q, c_j) \leq \alpha r^{(\delta_\text{old})}$ where $\alpha$ is either $1$, $2$, or $5$ based
        on whether $c_j$ was a regular, extended, or zombie center, respectively. 
        By construction and using Observation~\ref{obs:decrOper}, we can deduce that the state of $c_j \in S$ does not change (e.g., if $c_j$ was extended then it remains extended), and together with the fact that $r^{(\delta_\text{old})} \leq r^{(\delta)}$, the validity of $P_1$ follows.

        \item Assuming that the algorithm proceeds with case C2, we analyze the two subcases separately:
        \begin{enumerate} 
            \item Assume that the set of centers is modified due to case C2a, and let 
            $p_{t_1}, c_{t_2}, p_{t_2}, \;\ldots\;, c_{t_l}, p_{t_l}$ be the corresponding sequence satisfying~(\ref{eq:seq_for_replac_fully}). 
            If any point $q \in P$ switches clusters during the Decreasing Operation in Line~\ref{algline:switch_cluster_1_fully}, then $q$ satisfies the statement of $P_1$. As explained before in the proof about the validity of $P_2$, a point could switch clusters also in Line~\ref{algline:shift_centers}, but this happens only for center points and thus the claim for $P_1$ is not affected.
            
            Hence it remains to examine points which do not switch clusters.
            Consider an arbitrary point $q \in P$, and let $C_{j'}$ be the cluster containing it before and after the point deletion.
            If the most recent non-zombie center of $C_{j'}$ is updated during the Decreasing Operation either in Line~\ref{algline:cluster_reg_1_fully} or in Line~\ref{algline:cluster_reg_2_fully}, then the point $q$ is within distance $r^{(\delta)}$ from the regular center $c_{j'}$ of the regular cluster $C_{j'}$, as needed for $q$. Now regarding the point $q \in C_{j'}$, assume that the most recent non-zombie center $c^{(\text{nz})}_{j'}$ of the cluster $C_{j'}$ does not change after the point deletion, namely $c^{(\text{nz})}_{j'} \in S^{(\text{nz})}_\text{old} \cap S^{(\text{nz})}$. Recall that we have $S = (S_\text{old} \setminus \{p\}) \cup \{p_{t_l}\}$,
            and consider the following two cases:
            \begin{itemize}
            \item If the index $j'$ is equal to an index $t_j$ of the sequence that satisfies~(\ref{eq:seq_for_replac_fully}) where $1 \leq j \leq l$,
            then by the validity of $P_2$ it holds that $\dist(q, c^{(\text{nz})}_{t_j}) \leq 2r^{(\delta)}$. 
            Based on the way the algorithm sets each $c_{t_j}$ and as the sequence satisfies the conditions of (\ref{eq:seq_for_replac_fully}), 
            there must have been a point $p_{t_j}$ in $C_{t_j}$ such that $\dist(c_{t_j}, p_{t_j}) \leq r^{(\delta_\text{old})}$.\footnote{Note that for $j = l$, 
            the center point $c_{t_l}$ is actually the point $p_{t_l}$, and so $\dist(c_{t_l}, p_{t_l}) = 0$.}
            Since before the point deletion such a point $p_{t_j}$ was in $C_{t_j}$ and the point $c^{(\text{nz})}_{t_j} = c^{(\text{nz})}_{j'}$ 
            was the most recent non-zombie
            center of the cluster $C_{t_j}$, by the validity  of $P_2$ (before the update) we conclude that $\dist(p_{t_j}, c^{(\text{nz})}_{t_j}) \leq 2r^{(\delta_\text{old})}$.
            In turn, by the triangle inequality we have  $\dist(c_{t_j}, c^{(\text{nz})}_{t_j}) \leq \dist(c_{t_j}, p_{t_j}) + \dist(p_{t_j}, c^{(\text{nz})}_{t_j}) \leq 3r^{(\delta)}$ (recall that $r^{(\delta_\text{old})} \leq r^{(\delta)}$). Thus by using the triangle inequality again, it holds that: 
            \[
                \dist(q, c_{t_j}) \;\leq\; \dist(q, c^{(\text{nz})}_{t_j}) \;+\; \dist(c^{(\text{nz})}_{t_j}, c_{t_j}) \;\leq\; 5r^{(\delta)}.
            \] 
            Finally, notice that in this case the cluster $C_{t_j}$ is a zombie cluster and since its most recent non-zombie center
            is not updated, it means that the center point $c_{t_j}$ remains a zombie center. Therefore
            for any point $q \in P$ which belongs to a cluster whose index is equal to an index $t_j$ as defined before in~(\ref{eq:seq_for_replac_fully}), the statement of $P_1$ holds.
            
            \item If the index $j'$ is not equal to any of the previous $t_j$ as defined before in~(\ref{eq:seq_for_replac_fully}), we have by construction that the center $c_{j'}$ belongs to both sets $S_\text{old}$ and $S$. Moreover the position of $c_{j'}$ in the metric space does not change, and similarly the state of $c_{j'}$ remains the same. 
            Therefore by the induction hypothesis and using that $r^{(\delta_\text{old})} \leq r^{(\delta)}$,
            we can infer that the point $q$ satisfies the statement of $P_1$, and this concludes the validity of $P_1$.
            \end{itemize}
            
            \item Assume that the set of centers is modified due to case C2b, and recall that $p$ is the deleted point with index $i$ in $S_\text{old}$. To avoid confusion with the reset of the center $c_i$ and its cluster $C_i$ in Lines~\ref{algline:reset_ci_C2b_fully} and \ref{algline:reset_Cli_C2b_fully},
            let $C_i'$ be the cluster $C_i$ before it is updated. 
            The algorithm calls the Reassigning Operation, and note that Line~\ref{algline:find_c_fully} is executed properly based on Claim~\ref{clm:reassOper_q_atmost_r_c}.
            Since the queue $Q$ in the Reassigning Operation is initialized to $C_i' \setminus \{p\}$, all points in $C_i'$ are added to a different regular cluster $C_{j'}$ (see Line~\ref{algline:switch_cluster_2_fully}). In turn, all points at a distance greater than $r^{(\delta)}$ from the center $c_{j'}$ of the affected cluster $C_{j'}$ are added to $Q$ in~\cref{algline:add_farpnt_to_Q}. As a result, by construction and by~\cref{clm:reassOper_q_atmost_r_c} using that $r^{(\delta_\text{old})} \leq r^{(\delta)}$, each point $q$ in any cluster affected during the Reassigning Operation ends up in a regular cluster $C_{j'}$ such that $\dist(q, c_{j'}) \leq r^{(\delta)}$, as required.
            Whenever a point $q \in P$ switches clusters, it ends up in a non-zombie cluster due to either Line~\ref{algline:cluster_reg_1_fully}, Line~\ref{algline:cluster_reg_2_fully}, Line~\ref{algline:switch_cluster_1_fully} or Line~\ref{algline:switch_cluster_2_fully}, and so $q$ satisfies the statement of $P_1$ by construction.
      
            Hence, it remains to consider an arbitrary point $q \in P \setminus C_i'$ that remains in the same cluster after the point deletion. If the cluster containing $q$ changes state in this subcase, it becomes a regular cluster (see Observation~\ref{obs:decrOper}), and by construction the point $q$ is within distance $r^{(\delta)}$ from its regular center. To that end, assume that the state of the cluster containing $q$ remains unchanged.
            By construction we have  $S = (S_\text{old} \setminus \{p\}) \cup \{c'\}$, and the centers with index different than $i$ in $S_\text{old}$ do not
            change positions in the metric space. Therefore by the induction hypothesis and using that $r^{(\delta_\text{old})} \leq r^{(\delta)}$,
            we can infer that such a point $q$ satisfies the statement of $P_1$, and this concludes the validity of $P_1$.
        \end{enumerate}
    \end{enumerate}
    
    \paragraph{Point insertion.} 
    Recall that based on Observation~\ref{obs:DecrOperProp}, we can assume that $\delta \geq \delta_\text{old}$ and $r^{(\delta)} \geq r^{(\delta_\text{old})}$. 
    Let $p^+$ be the inserted point into the point set $P$, and let $P$ be the updated point set including the inserted point $p^+$. If $p^+$ is within distance
    $r^{(\delta_\text{old})}$ from a non-zombie center $c_j$ of $S_\text{old}$, then the algorithm 
    adds the point $p^+$ to the cluster $C_j$, and calls the Decreasing Operation. 
    Since $c_j$ is a non-zombie center we have  $c^{(\text{nz})}_j = c_j$. 
    Thus the point $p^+$ is within distance $r^{(\delta)}$ from the most recent non-zombie center $c^{(\text{nz})}_j$, and so both $P_1$ and $P_2$ are satisfied for the new point $p^+$. Consider an arbitrary point $q \in P \setminus \{p^+\}$. If the point $q$ ends up in a non-zombie cluster due to either Line~\ref{algline:cluster_reg_1_fully},
    Line~\ref{algline:cluster_reg_2_fully}, or Line~\ref{algline:switch_cluster_1_fully} during the Decreasing Operation, then $q$ satisfy the statements of $P_1$ and $P_2$ by construction.
    Otherwise, the point $q$ remains in the same cluster whose state does not change, and the validity of both $P_1$ and $P_2$ follows by
    the induction hypothesis and the fact that $r^{(\delta_\text{old})} \leq r^{(\delta)}$.

    Otherwise assume that the new inserted point $p^+$ is at a distance greater than $r^{(\delta_\text{old})}$ from all the non-zombie centers of $S_\text{old}$. In this situation, the algorithm proceeds either with case C3 or with case C4, and we analyze the two cases separately:

    \begin{enumerate}
    \item Assume that the algorithm proceeds with case C3 in Line~\ref{algline:ci_cj_leq_rdelta_fully}, where an older center point $c_i$ is exempted from being a center. 
    To avoid confusion with the reset of $c_i$, let $c_i'$ be the old point $c_i$ determined in Line~\ref{algline:determ_i_fully} and let $C_i'$ be the corresponding cluster of $c_i'$ before the update. By construction the algorithm
    removes $c_i'$ from the set of centers and adds a point $c'$ to the set of centers, which means that $S = (S_\text{old} \setminus \{c_i'\}) \cup \{c'\}$. 
    
    At first, we show that all points in $C_i'$ satisfy the statements of both $P_1$ and $P_2$. Consider an arbitrary point $q \in C_i'$ that belonged to the old version of the cluster $C_i$.
    By Lemma~\ref{lem:rem_ci_reg} the point $c_i'$ was a regular center before the update, so by the induction hypothesis of $P_1$ the point $q$ is within distance $r^{(\delta_\text{old})}$ from  $c_i'$.
    By construction the point $q$ is added to the cluster $C_s$ of the extended center $c_s$ determined in Line~\ref{algline:determ_i_fully},
    and as (\ref{choose_cs_ci_fully}) is satisfied we have $\dist(c_s, c_i') \leq r^{(\delta_\text{old})}$.
    Therefore by the triangle inequality, all points in $C_i'$ end up in the cluster $C_s$ and are 
    within distance $2r^{(\delta_\text{old})} \leq 2r^{(\delta)}$ from the extended center $c_s$. Since the center $c_s$ is the most recent non-zombie 
    center of the cluster $C_s$, the statements of both $P_1$ and $P_2$ are satisfied for all these points 
    inside $C_s$, and in turn for all the points in $C_i'$. Furthermore, note that the new point 
    $p^+$ joins a regular cluster $C_j$ and its distance from the corresponding regular center 
    $c_j$ is at most $r^{(\delta_\text{old})} \leq r^{(\delta)}$, as needed for $p^+$. 
    
    For the analysis of the remaining points in $P \setminus (C_i' \cup \{p^+\})$, observe
    that the algorithm replaces the exempted point $c_i'$ with another point $c'$ by proceeding either with case C3a or case C3b.
    We split the analysis in two subcases: 
    \begin{enumerate}
        \item If the algorithm proceeds with case C3a, then the updated center $c_i$ becomes regular and is placed in the position of the old $z$-th center determined in Line~\ref{algline:find_cz}. By construction the points that join the updated
        cluster $C_i$ satisfy the statement of both $P_1$ and $P_2$. Moreover, the corresponding zombie cluster $C_z$ is updated in Line~\ref{algline:Clz_updated},
        and the algorithm calls the function Replace() in Line~\ref{algline:call_replac}. In turn, the algorithm proceeds accordingly either with case C1 or with case C2, aiming to locate the new $z$-th center point. To that end, for the remaining points 
        the claims follow by using the same arguments as before in the analysis of cases C1 and C2 in the point deletion, by replacing in the proof the (zombie cluster) $C_i$ with $C_z$ and the (deleted) point $p$ with $c_i'$.
        
        \item If the algorithm proceeds with case C3b, then by construction the point $c'$ is actually the new inserted point $p^+$, and it becomes the $i$-th center of the updated $S$. Together with the induction hypothesis and the construction of the Decreasing Operation, the claims for both $P_1$ and $P_2$ hold.
    \end{enumerate}
    
    \item Alternatively assume that the algorithm proceeds with case C4. If the algorithm does not call the Increasing Operation, then
    all points must be within distance $r^{(\delta_\text{old})}$ from the set $S_\text{old}$, and
    the set of centers is not modified (i.e., $S = S_\text{old}$). In the end the algorithm calls the Decreasing Operation
    which in Lines~\ref{algline:decr_radius_1_fully}-\ref{algline:cluster_reg_1_fully} converts all centers to regular, and so the claims of both $P_1$ and $P_2$ hold.
    Otherwise if the algorithm calls the Increasing Operation, all centers are converted to regular in Line~\ref{algline:cluster_reg_3_fully}, and in this case the algorithm increases the extension level by at least one. Namely, we have 
    $\delta-1 \geq \delta_\text{old}$ and in turn $r^{(\delta-1)} \geq r^{(\delta_\text{old})}$. By the induction hypothesis of $P_1$, every point in $P \setminus \{p^+\}$ is within distance $5r^{(\delta_\text{old})} \leq 5r^{(\delta-1)}$ from the center of the cluster it belongs to. 
    Since $r^{(\delta)} = 5r^{(\delta-1)}$, after the execution of the Increasing Operation
    every point in $P \setminus \{p^+\}$ is within distance $r^{(\delta)}$ from the center of the regular cluster it belongs to.    
    
    Let $p$ be a point of maximum distance from the set $S_\text{old}$ determined in Line~\ref{algline:p_max_C4}.
    If the point $p$ is not the new inserted point $p^+$, then it must be true that $\dist(p^+, S_\text{old}) \leq r^{(\delta)}$, and by construction the set of centers is not modified (i.e., $S = S_\text{old}$). Hence, as the point $p^+$ is also within distance $r^{(\delta)}$ from its regular center, the validity of both $P_1$ and $P_2$ follows.
    Otherwise if the algorithm proceeds with case C3b after case C4 in Line~\ref{algline:C3b_after_C4},
    then the algorithm assigns $p^+$ to be the $i$-th center of the updated set $S$.
    Note that by the way the Increasing Operation works, there exist such a pair of points $c_s, c_i$ in Line~\ref{algline:determ_i_fully_1}
    such that $\dist(c_s, c_i) \leq r^{(\delta)}$. Finally, by using the same arguments as before regarding the extended cluster in Line~\ref{algline:cs_extend_2}, we can conclude that
    both statements $P_1$ and $P_2$ hold.
    \end{enumerate}
\end{proof}

Consequently, we can deduce that the set of centers $S$ maintained by the algorithm is a \distDS{$5r^{(\delta)}$}, and thus \ref{inv_3_fully} is satisfied throughout the algorithm.

\subsubsection{Analysis of the Update Time and Finishing the Proof of~\cref{th:fully_consist}}
The last piece to complete the proof of our main result is to show that the update time of the fully dynamic algorithm is a polynomial in the size of the input. Note that the only ambiguous step of the algorithm is in the beginning of case C2. The next lemma demonstrates that this step can be computed in polynomial time as well.

\begin{lemma} \label{lem:runtime_seq}
    After a point update, the fully dynamic algorithm can either find a sequence satisfying~(\ref{eq:seq_for_replac_fully}) or
    determine that there is no such sequence in $O(|P| \cdot k)$ time.
\end{lemma}
\begin{proof}
    The problem of either detecting a sequence satisfying~(\ref{eq:seq_for_replac_fully}) or
    determining that there is no such sequence can be reduced to a simple graph search problem, as follows. 
    We build a directed graph $G = (P, E)$, such that the edge set is: 
    \begin{align*}
        E \coloneqq\; &\{(p, c) \in P \times P \mid p \notin S \land c \in S \land \dist(p, c) \leq r^{(\delta)}\} \; \cup \\
        &\{(c, p) \in P \times P \mid c \in S \land p \notin S \land p \text{ is in the cluster of } c \land \dist(c, p) > r^{(\delta)}\}.
    \end{align*}
    Namely, the edges either correspond to a point and a center whose distances are at most $r^{(\delta)}$,
    or to a center and a point of its cluster whose distances are greater than $r^{(\delta)}$.
    Let $V_n$ be the set of vertices that correspond to the points of the first cluster $C_n$ in Line~\ref{algline:seq_init}.
    Then the problem of whether there is a sequence satisfying~(\ref{eq:seq_for_replac_fully}) is
    reduced to the problem of whether there exists a path from a vertex in $V_n$
    to another vertex with no outgoing edges. This graph search problem can be solved with DFS. Finally, observe that the size of $E$ and the time to build the graph and run DFS is $O(|P| \cdot k)$.
\end{proof}

\antonis{Should we mention that our runtime is independent of the aspect ratio?}
\begin{lemma} \label{lem:runtime_fully}
    The fully dynamic Algorithm~\ref{alg:fully_consist} performs $O(|P| \cdot k)$ operations per point update.
\end{lemma}
\begin{proof}
    After a point update, the algorithm can perform every task in $O(|P| \cdot k)$ operations.
    Note that the Reassigning Operation can be implemented with a similar idea as in Lemma~\ref{lem:runtime_seq}.
\end{proof}

By combining Lemma~\ref{lem:size_S_worst_recourse_fully}, Lemma~\ref{lem:invar_2_satisf_fully}, and Lemma~\ref{lem:C_covers_fully}, 
we can now finish the proof of Theorem~\ref{th:fully_consist} which we restate for convenience.

\fullyconsist*
\begin{proof}
    Consider Algorithm~\ref{alg:fully_consist} and its analysis.
    Based on Lemma~\ref{lem:C_covers_fully}, we can infer that the set $S$ is a \distDS{$5r^{(\delta)}$}, and by Lemma~\ref{lem:invar_2_satisf_fully} we have  $r^{(\delta)} \leq 10R^*$. Therefore using also Lemma~\ref{lem:size_S_worst_recourse_fully}, the claim follows.
\end{proof}

%% file: trunk/decr_consist_kcenter.tex
\section{Decremental Consistent $k$-Center Clustering} \label{sec:decr}
In the decremental setting, we are given as input a point set $P$ from an arbitrary metric space which undergoes point deletions. 
The goal is to decrementally maintain a feasible solution for the $k$-center clustering problem with 
small approximation ratio and low recourse. In this section, we develop a deterministic decremental algorithm that 
maintains a $6$-approximate solution for the $k$-center clustering problem
with worst-case recourse of $1$ per update, as demonstrated in the following theorem.
Recall that by $R^*$ we denote the current optimal radius of the given instance.

\begin{restatable}{theorem}{decrconsist} \label{th:decr_consist}
    There is a deterministic decremental algorithm that, given a point set $P$ from an arbitrary metric space 
    subject to point deletions and an integer $k \geq 1$,
    maintains a subset of points $S \subseteq P$ such that:  \vspace{-0.3em}
    \begin{itemize}
        \setlength\itemsep{-0.3em}
        \item The set $S$ is a \distDS{$6R^*$} in $P$ of size $k$.
        \item The set $S$ changes by at most one point per update.
    \end{itemize}
\end{restatable}

Our algorithm maintains a set $S = \{c_1, \ldots, c_k\}$ of $k$ \emph{center points} in some \emph{order} and a \emph{radius} $\hat{r}$. 
Each center point $c_i$ is the unique center of the \emph{cluster} $C_i$ which
contains exactly the points served by $c_i$. In our algorithm, every point belongs to exactly
one cluster (i.e., every point is served by exactly one center), and so
the idea is that the center $c_i$ is responsible to serve every point inside the cluster $C_i$.
A cluster $C_i$ is categorized as either a \emph{regular cluster} or a \emph{zombie cluster}.
The corresponding center is called \emph{regular center} (\emph{regular center point}) or \emph{zombie center} (\emph{zombie center point}), respectively. We refer to this as the \emph{state} of the cluster and its center.

A regular cluster $C_i$ contains only points that are
within distance $\hat{r}$ from the corresponding regular center $c_i$. However, a zombie cluster $C_i$ can also contain points that are at distances greater than $\hat{r}$ from their corresponding zombie center $c_i$. 
We remark that the implicit order of points in $S$ may change
after an update, but we guarantee that the set $S$ itself changes by at most one point per update (i.e., the worst-case recourse is $1$).
Throughout the algorithm we want to satisfy the following three invariants. 

\vspace{1em}
\noindent
\textbf{Invariant~1}: The size of the set $S$ is exactly $k$.

\noindent
\textbf{Invariant~2}: There exists a point $p_S \in P \setminus S$ such that the pairwise distances of points in $S \cup \{p_S\}$ are at least $\hat{r}$.

\noindent
\textbf{Invariant~3}: The set $S$ is a \distDS{$3\hat{r}$}.

\subsection{Decremental Algorithm}
At first we describe the preprocessing phase of the algorithm, and afterwards we describe
the way the algorithm handles point deletions. A pseudocode of the algorithm is provided in Algorithm~\ref{alg:dec_consist}.

\subsubsection{Preprocessing Phase}
Let $D$ be the ordered set of pairwise distances of points in $P$ as defined in Theorem~\ref{th:2appr_alg}.
In the preprocessing phase, we run the algorithm of Theorem~\ref{th:2appr_alg} to obtain a $2$-approximate solution $S$, and let  $r_1 \coloneqq \max_{p \in P} \dist(p, S)$.
Observe that $r_1$ must be included in the set $D$. Let $r_0$ be the value just before $r_1$ in the ordered set $D$.
Next, we repeat the following process with the goal to satisfy all the three invariants. While the set $S$ has size strictly less than
$k$, we add to $S$ an arbitrary point $c$ that is at a distance at least $r_1$ from the current centers (i.e., $\dist(c, S) \geq r_1$).
If there is no such point, then we set $r_1 \coloneqq r_0$ and $r_0$ to the value just before the updated $r_1$ in the ordered set $D$,
and the process continues until the size of $S$ reaches $k$. 
Once the size of $S$ becomes $k$, the value of $r_1$ is set to the smallest value in the ordered set $D$
such that the set $S$ is a \distDS{$r_1$}.

Initially, the value of the radius $\hat{r}$ is set to $r_1$,
all the center points of $S$ are regular centers, and all the clusters are regular clusters. 
Also, each cluster is set to contain points in $P$ that are within distance $\hat{r}$ from its center, such that the clusters form a partition of the point set $P$.

\paragraph{Analysis of the preprocessing phase.}
After the preprocessing phase has finished, the size of $S$ is $k$ and the set $S$ is a \distDS{$\hat{r}$}.
Additionally, by construction there exists a point $p_S \in P \setminus S$ such that the pairwise distances
of points in $S \cup \{p_S\}$ are at least $\hat{r}$. Hence, in the beginning all the three invariants are satisfied.
Combining Invariant~1 with Invariant~2, we can infer that there exist at least $k + 1$ points whose pairwise distances
are at least $r_1$. Thus it holds that $\hat{r} \leq 2R^*$ by Lemma~\ref{lem:k+1_atleastr_lesseq2R}, 
which means that the set $S$ is initially a $2$-approximate solution.

\subsubsection{Point Deletion} 
    Let $\hat{r}$ be the current radius of the algorithm.
    Under point deletions, the algorithm proceeds as follows.
    
    Consider the deletion of a point $p$ from the point set $P$, and let $P \coloneqq P \setminus \{p\}$. If $p$ is not a center point (i.e., $p \notin S$),
    then the point $p$ is removed from the cluster it belongs to, and the algorithm just calls the Regulating Operation which is described below. 
    Note that the set $S$ remains unmodified and the recourse is zero.
    Otherwise assume that the deleted point $p$ is the $i$-th center point of $S$ (i.e., $p \in S$ and specifically $p = c_i$).
    The algorithm removes the point $p$ from the set $S$ and from the cluster $C_i$ it belongs to.
    If the deleted point $c_i$ is a zombie center, then its corresponding cluster $C_i$ is already a zombie cluster. 
    Otherwise as the center $c_i$ of $C_i$ is deleted, the regular cluster $C_i$ becomes a zombie cluster. The algorithm continues based on the two following cases:

    \begin{enumerate}[label=C\arabic*.]
        \item If there is another point $c'$ in $C_i$ at a distance greater than $\hat{r}$ from the remaining centers of $S$
        (i.e., if $\exists c' \in C_i$ such that $\dist(c', S) > \hat{r}$),\footnote{Recall that at this moment $c_i \notin S$ and $c_i \notin C_i$. \label{footnote:c_i_del_decr_decSec}} then the point 
        $c'$ is assigned to be the center of the zombie cluster $C_i$. In particular the point $c'$ is added to $S$, and as the point $c'$ is now the $i$-th center point of the ordered set $S$, the algorithm sets $c_i \coloneqq c'$ and
        the point $c_i = c'$ becomes the zombie center of the zombie cluster $C_i$. 
        The recourse in this case is $1$.

        \item Otherwise, every point $c' \in C_i$ is within distance $\hat{r}$ from the remaining centers of $S$. In this case, the algorithm tries to detect a sequence of $2l-1$ points of the following form:
        \begin{equation} \label{eq:seq_for_replac}
            p_{t_1}, c_{t_2}, p_{t_2}, \;\ldots\;, c_{t_l}, p_{t_l}
        \end{equation}
        where $t_1, \dots, t_l$ are indices different from each other,
        where $t_1 = i, p_{t_1} \in C_i$, and for every $j: 2 \leq j \leq l$ the points $p_{t_j}, c_{t_j}$ satisfy the following conditions:
        \begin{itemize}
            \item $c_{t_j} \in S$ such that $\dist(p_{t_{j-1}}, c_{t_j}) \leq \hat{r}$.\footref{footnote:c_i_del_decr_decSec}

            \item $t_j$ is the index of $c_{t_j}$ in the ordered set $S$.

            \item $p_{t_j} \in C_{t_j}$ such that $\dist(p_{t_j}, c_{t_j}) > \hat{r}$.
        \end{itemize}
        and $\dist(p_{t_l}, S) > \hat{r}$.

        \begin{enumerate}[label=C2\alph*.]
            \item If the algorithm detects such a sequence, then the point $p_{t_l}$ is assigned to be the center of the zombie cluster $C_{t_l}$,
            and the indices in the ordered set $S$ are shifted to the right with respect to the sequence.
            In particular the point $p_{t_l}$ is added to $S$, and let $q_{t_j}$ be temporarily the center point $c_{t_j}$ for every $j: 2 \leq j \leq l$. 
            Next, the algorithm sets $c_{t_l}$ to $p_{t_l}$, and $c_{t_j}$ to $q_{t_{j+1}}$ for every $j: 1 \leq j \leq l - 1$.
            The clusters are updated respectively, namely the point $q_{t_j}$ is removed from $C_{t_j}$ and the center point $c_{t_{j-1}}$ is added to $C_{t_{j-1}}$ for every $j: 2 \leq j \leq l$,\footnote{Notice that actually the center point $c_{t_{j-1}}$ is now the point $q_{t_j}$.} 
            and note that the center point $c_{t_l}$ is already in $C_{t_l}$.
            After the change of the order, for all $j: 1 \leq j \leq l$ the center point $c_{t_j}$
            is the zombie center of the zombie cluster $C_{t_j}$. Notice that this operation affects only the order 
            of the centers in $S$ and not the set $S$ itself, which means that the recourse in this case is $1$.

            \item Otherwise, let $c'$ be a point of maximum distance from the set $S$ of the remaining centers.
            In this case, the algorithm first calls the Reassigning Operation which is described below, 
            and then builds a new regular cluster with the point $c'$ as its center. 
            In particular the point $c'$ is added to $S$, and as the point $c'$ is now the $i$-th center point of the ordered set $S$, the algorithm sets $c_i \coloneqq c'$
            and the cluster $C_i$ is updated to contain initially only the center point $c_i = c'$. 
            The cluster $C_i$ is a regular cluster and the center point $c_i$ is a regular center.
            The recourse in this subcase is $1$.
            
            \subparagraph{Reassigning Operation.}
            Let $Q$ be a set containing initially the points in $C_i$, where $C_i$ is the older version of the $i$-th cluster.\footref{footnote:c_i_del_decr_decSec}
            During the Reassigning Operation, the following steps are performed iteratively
            as long as the set $Q$ is not empty. First, the algorithm picks a point $q \in Q$ and finds a center point $c_{j'} \in S$ such that $\dist(q, c_{j'}) \leq \hat{r}$.\footref{footnote:c_i_del_decr_decSec} Such a center point $c_{j'}$ must exist, since there is no sequence satisfying~(\ref{eq:seq_for_replac}).
            \antonis{should we prove it?} Let $C_j$ be the cluster that contains the point $q$.
            Then, the algorithm removes the point $q$ from $C_j$ and from $Q$, and adds $q$ to the cluster $C_{j'}$. Also the cluster $C_{j'}$ becomes a regular cluster and its center $c_{j'}$ becomes a regular center. Next, all the points from $C_{j'}$ whose
            distances are greater than $\hat{r}$ from the corresponding center point $c_{j'}$ (i.e., points from the set 
            $\{q_{j'} \in C_{j'} \mid \dist(q_{j'}, c_{j'}) > \hat{r}\}$) are added to the set $Q$. Now, if the updated set $Q$ is empty then this procedure ends,
            otherwise it is repeated using the updated set $Q$.
        \end{enumerate}
    \end{enumerate}

    \vspace{0.2em}
    \noindent
    In the end of each point deletion, the algorithm calls the Regulating Operation. 

    \subparagraph{Regulating Operation.}
        During the Regulating Operation the algorithm performs three tasks.
        (1) At first, if the set $S$ is a \distDS{$\hat{r}$} then all the clusters become regular clusters, all the centers become regular centers, and
        the algorithm updates $\hat{r}$ to the smallest value $r'$ such that the set $S$ is a \distDS{$r'$}.
        Note that the value of the radius $\hat{r}$ may not decrease. 
        (2) Then, the algorithm iterates over all indices $j: 1 \leq j \leq k$,
        and if every point $q \in C_j$ of cluster $C_j$ is within distance $\hat{r}$ from the center $c_j$ (i.e., if for all $q \in C_j: \dist(q, c_j) \leq \hat{r}$), then the cluster $C_j$ becomes a regular cluster
        and its center $c_j$ becomes a regular center.
        (3) Finally for every point $q \in P$, if $q \in C_j$ is at distance greater than $\hat{r}$ from the center $c_j$
        of the cluster $C_j$ which contains $q$, and also $q$ is within distance $\hat{r}$ from another regular center $c_{j'}$ of the regular cluster $C_{j'}$, then the algorithm removes the point $q$ from $C_j$ and adds it to $C_{j'}$.\footnote{In case there is more than one regular center point within distance $\hat{r}$ from the point $q$, the algorithm adds $q$ to only one of the corresponding regular clusters arbitrarily.}
        
\begin{algorithm}[H]
    \DontPrintSemicolon
    \caption{\textsc{decremental consistent $k$-center}{}}
    \label{alg:dec_consist}

    \SetAlgoLined
    \SetArgSty{textrm}
    
    \tcp{Let $S = \{c_1, \dots, c_k\}$ be the ordered set of $k$ centers}

    \tcp{Let $\hat{r}$ be the current radius}

    \vspace{0.5em}

    \SetKwFunction{FRegulatingOp}{RegulatingOp}
    \Procedure{\FRegulatingOp{}} {
        \If {$\max_{q \in P} \dist(q, S) \leq \hat{r}$} { \label{algline:decr_radius_1}
            For all $c_j \in S: c_j$ becomes regular center and $C_j$ becomes regular cluster \label{algline:cluster_reg_2}
        
            \vspace{0.5em}
            
            $\hat{r} \gets \max_{q \in P} \dist(q, S)$ \tcp{The smallest value such that $S$ is a \distDS{$\hat{r}$}} \label{algline:decr_radius_2}
        }

        \vspace{0.5em}
        
        \For{$j \in \{1, \ldots, k\}$} {
            \If{$\forall q \in C_j: \dist(q, c_j) \leq \hat{r}$} {
                $C_j$ becomes regular cluster and $c_j$ becomes regular center \label{algline:cluster_reg_1}
            }
        }
        
        \vspace{0.5em}
        
        \For{$q \in P$} {
            Let $C_j$ be the cluster such that $q \in C_j$
            \vspace{0.5em}
            
            \If{$\dist(q, c_j) > \hat{r}$ \textbf{and} $\exists j': (C_{j'}$ is regular cluster \textbf{and} $\dist(q, c_{j'}) \leq \hat{r}$)} { 
                Remove $q$ from $C_j$ and add $q$ to $C_{j'}$ \label{algline:switch_cluster_1}
            }
        }
    }

    \vspace{1em}
    
    \SetKwFunction{FReassigningOp}{ReassigningOp}
    \Procedure{\FReassigningOp{}} {
        $Q \gets C_i$ \tcp{The deleted point $p$ is the $i$-th center point of $S$, and it has been removed from $S$ and $C_i$}
        \vspace{0.5em}

        \While{$Q$ is not empty} {
            Pick $q$ from $Q$, and let $C_j$ be the cluster such that $q \in C_j$

            Find $c_{j'} \in S$ such that $\dist(q, c_{j'}) \leq \hat{r}$ \label{algline:find_c_dec}

            Remove $q$ from $C_j$ and $Q$, and add $q$ to $C_{j'}$ \label{algline:switch_cluster_2}

            Cluster $C_{j'}$ becomes regular cluster and $c_{j'}$
            becomes regular center

            Add all points from $\{q_{j'} \in C_{j'}: \dist(q_{j'}, c_{j'}) > \hat{r}\}$ to $Q$
        }
    }
\end{algorithm}

\begin{algorithm}[H]
    \DontPrintSemicolon
    \SetArgSty{textrm}
    \tcp{Continuation of Algorithm~\ref{alg:dec_consist}}
    \vspace{1em}
    
    \SetKwFunction{FDeletePoint}{DeletePoint}
    \Procedure{\FDeletePoint{p}} {
        Remove $p$ from $P$
        \vspace{0.5em}
        
        \If{$p \notin S$} {
            Remove $p$ from the cluster it belongs to
            
            \FRegulatingOp{}{}

            \vspace{0.5em}
            \textbf{exit} from \FDeletePoint{}
        }
        \vspace{0.5em}
        \Else(\tcp*[h]{$p \in S$}) { 
            Let $i$ be the index of $p$ in $S = \{c_1, \dots, c_k\}$, namely $c_i = p$

            Remove $p$ from $S$ and $C_i$
            \vspace{0.5em}
            
            \If{$c_i$ is \textbf{not} zombie center} {
                Cluster $C_i$ becomes zombie cluster
            }
            \vspace{0.5em}
            
            \If(\tcp*[h]{case C1}){$\exists c' \in C_i: \dist(c', S) > \hat{r}$} {
                Add $c'$ to $S$
    
                $c_i \gets c'$
                
                $c_i$ becomes zombie center
            } 
            \vspace{0.5em}
            
            \Else(\tcp*[h]{case C2}){
                Find a sequence $p_{t_1}, c_{t_2}, p_{t_2}, \;\ldots\;, c_{t_l}, p_{t_l}$ 
                of $2l-1$ points such that: 
    
                $t_1, \dots, t_l$ are indices different from each other,
                
                $t_1 = i, p_{t_1} \in C_i$, and for every $j: 2 \leq j \leq l$: \label{algline:seq_init_dec}
                
                (1) $c_{t_j} \in S$ and $\dist(p_{t_{j-1}}, c_{t_j}) \leq \hat{r}$ 
    
                (2) $t_j$ is the index of $c_{t_j}$ in $S = \{c_1, \ldots, c_k\}$
                
                (3) $p_{t_j} \in C_{t_j}$ and $\dist(p_{t_j}, c_{t_j}) > \hat{r}$
                
                and $\dist(p_{t_l}, S) > \hat{r}$ \label{algline:ptl_S_dec}
                \vspace{0.5em}
                
                \If(\tcp*[h]{case C2a}){such a sequence is found} {
                    Add $p_{t_l}$ to $S$
    
                    \vspace{0.5em}
                    
                    \For{$j \in \{1, \ldots, l-1\}$} {
                        $c_{t_j} \gets c_{t_{j+1}}$
    
                        Remove $c_{t_j}$ from $C_{t_{j+1}}$ and add $c_{t_j}$ to $C_{t_j}$ \label{algline:shift_centers_dec}
                    }
                    \vspace{0.5em}
                    $c_{t_l} \gets p_{t_l}$ \tcp{$c_{t_l}$ is already in $C_{t_l}$}
    
                    $c_i$ becomes zombie center
                }
    
                \vspace{0.5em}
                \Else(\tcp*[h]{case C2b}) {
                    \FReassigningOp{}
                    
                    Let $c'$ be a point of maximum distance from the set $S$
    
                    Add $c'$ to $S$
    
                    $c_i \gets c'$, and $c_i$ becomes regular center \label{algline:reset_ci_C2b_dec}
                
                    $C_i \gets \{c_i\}$, and $C_i$ becomes regular cluster \label{algline:reset_Cli_C2b_dec}
                }
            }
    
            \vspace{0.5em}
        }
    
        \vspace{0.5em}
        \FRegulatingOp{}
    }
\end{algorithm}

\subsection{Analysis of the Decremental Algorithm}
    Our goal in this section is to prove Theorem~\ref{th:decr_consist} by analyzing Algorithm~\ref{alg:dec_consist}.
    Throughout the analysis, we use the following notation during a point deletion.
    Let $S_\text{old}$ be the set of centers before the algorithm has processed the update,
    and $S$ be the set of centers after the algorithm has processed the update.
    Similarly, let $\hat{r}_\text{old}$ be the value of the radius before the algorithm has processed the update,
    and $\hat{r}$ be the value of the radius after the algorithm has processed the update.

    The proof of Theorem~\ref{th:decr_consist} almost immediately follows from the three invariants. Therefore in the following we prove that after a point deletion, 
    all the three invariants are satisfied in Algorithm~\ref{alg:dec_consist}.

    \begin{lemma} \label{lem:size_S_worst_recourse_dec}
        After a point deletion, the size of $S$ is $k$ (i.e., Invariant~1 is satisfied). Moreover,
        the worst-case recourse per update is $1$.
    \end{lemma}
    \begin{proof}
        The proof follows an induction argument. In the preprocessing phase, the algorithm constructs a set of centers $S$ of size $k$.
        Under a point deletion, a center point can be removed from the set of centers only if it is also removed from the point set $P$.
        In this case, the old center point is replaced by another point.
        Thus after a point deletion, the size of the set of centers remains the same, and the worst-case recourse in an update is at most $1$.
    \end{proof}
    
    \begin{observation} \label{obs:regOper}
        Let $\hat{r}_1$ and $\hat{r}_2$ be the value of the radius before and after the execution of the Regulating Operation respectively.
        Then it holds that $\hat{r}_2 \leq \hat{r}_1$.
    \end{observation}
    
    \begin{lemma} \label{lem:inv_2}
        After a point deletion, there exists a point $p_S \in P \setminus S$ such that the pairwise distances of points in $S \cup \{p_S\}$ are at least $\hat{r}$ (i.e., Invariant~2 is satisfied). Moreover,
        the value of the radius $\hat{r}$ is at most $2R^*$.
    \end{lemma}
    \begin{proof}
        We first prove by induction on the number of deleted points that Invariant~2 is satisfied.
        For the base case, the statement holds by the preprocessing phase.
        For the induction step, let $p$ be the deleted point from the point set $P$, and let $P$ be the updated point set without the deleted point $p$. If $p$ is not a center point (i.e., $p \notin S_\text{old}$), then the algorithm only calls the Regulating Operation.
        In turn, we have that $S = S_\text{old}$ and by Observation~\ref{obs:regOper} it holds that $\hat{r} \leq \hat{r}_\text{old}$.
        Since by induction hypothesis the pairwise distances of points in $S_\text{old}$ are at least $\hat{r}_\text{old}$,
        we can conclude the pairwise distances of points in $S$ are at least $\hat{r}$.    
        Moreover, by the way the Regulating Operation works in Lines~\ref{algline:decr_radius_1}-\ref{algline:decr_radius_2}, there must exist a
        point $p_S \in P \setminus S$ of distance at least $\hat{r}$ from the updated set $S$, as needed. Otherwise assuming that the deleted point $p$ is the $i$-th center point (i.e., $p \in S_\text{old}$ and $p = c_i$), we analyze the two possible cases of the algorithm.

        \begin{enumerate}
            \item Assume that the algorithm proceeds with case C1, and let $c'$ be the corresponding new zombie center of the zombie cluster $C_i$.
            Then the distance between the point $c'$ and the set $S_\text{old} \setminus \{p\}$ is greater than $\hat{r}_\text{old}$, and by induction hypothesis 
            the pairwise distances of points in $S_\text{old}$ are at least $\hat{r}_\text{old}$. By construction we have that $S = (S_\text{old} \setminus \{p\}) \cup \{c'\}$, and since the algorithm calls the Regulating Operation in the end, it holds that $\hat{r} \leq \hat{r}_\text{old}$ by Observation~\ref{obs:regOper}. 
            Hence, the pairwise distances of points in $S$ are at least $\hat{r}$.
            Moreover by the way the Regulating Operation works in Lines~\ref{algline:decr_radius_1}-\ref{algline:decr_radius_2}, 
            there must exist a point $p_S \in P \setminus S$ of distance at least $\hat{r}$ from the updated set $S$, and so the claim follows.

            \item If the algorithm proceeds with case C2, then we analyze the two subcases separately:
            \begin{enumerate}
            \item Assume that the algorithm proceeds with subcase C2a.
            Since there is a sequence satisfying~(\ref{eq:seq_for_replac}) it holds that $\dist(p_{t_l}, S_\text{old} \setminus \{p\}) > \hat{r}_\text{old}$ (see also Line~\ref{algline:ptl_S_dec}), and by induction hypothesis the pairwise distances of points in $S_\text{old}$ are at least $\hat{r}_\text{old}$. By construction we have that $S = (S_\text{old} \setminus \{p\}) \cup \{p_{t_l}\}$, and since the algorithm calls the Regulating Operation in the end, it holds that $\hat{r} \leq \hat{r}_\text{old}$ by Observation~\ref{obs:regOper}. In turn, it follows that the pairwise distances of points in $S$ are at least $\hat{r}$. Eventually, by the way the Regulating Operation works in Lines~\ref{algline:decr_radius_1}-\ref{algline:decr_radius_2}, there must exist a
            point $p_S \in P \setminus S$ of distance at least $\hat{r}$ from the updated set $S$,
            as needed.
            
            \item Assume that the algorithm proceeds with subcase C2b, and let $c'$ be the corresponding new center
            of maximum distance from the set $S_\text{old} \setminus \{p\}$. 
            By construction we have that $S = (S_\text{old} \setminus \{p\}) \cup \{c'\}$. 
            Also due to the Regulating Operation we have that $\hat{r} \leq \max_{q \in P} \dist(q, S)$, and since $S_\text{old} \setminus \{p\} \subseteq S$ it holds that $\max_{q \in P} \dist(q, S) \leq \max_{q \in P} \dist(q, S_\text{old} \setminus \{p\})$. Therefore            
            it follows that: 
            \begin{equation} \label{ineq:c_more_r_dec}
                \hat{r} \;\leq\; \max_{q \in P} \dist(q, S_\text{old} \setminus \{p\}) \;=\; \dist(c', S_\text{old} \setminus \{p\}).
            \end{equation}

            Together with the induction hypothesis and Observation~\ref{obs:regOper}, we can infer that the pairwise distances of points in $S$ are at least $\hat{r}$. Finally, once again by
            the way the Regulating Operation works in Lines~\ref{algline:decr_radius_1}-\ref{algline:decr_radius_2}, there must exist a point $p_S \in P \setminus S$ of distance at least $\hat{r}$ from the updated set $S$,
            and so Invariant~2 is satisfied.
            \end{enumerate}
        \end{enumerate}

        Since Invariant~2 is satisfied, there exists a point $p_S \in P \setminus S$ such that the pairwise distances of points in $S \cup \{p_S\}$ are at least $\hat{r}$. By Lemma~\ref{lem:size_S_worst_recourse_dec}, the size of the set $S \cup \{p_S\}$ is $k + 1$.
        Therefore based on Lemma~\ref{lem:k+1_atleastr_lesseq2R}, it holds that $\hat{r} \leq 2R^*$.
    \end{proof}
    
    Let $S^{(\text{reg})}_\text{old}$ and $S^{(\text{reg})}$ be the ordered set of the most recent regular center points,
    before and after the algorithm has processed the update respectively.
    In other words, the $i$-th point in $S^{(\text{reg})}_\text{old}$ corresponds to the
    last regular center point of the cluster $C_i$ in the metric space, just before the update. 
    Note that the set $S^{(\text{reg})}_\text{old}$ may be different than the set of regular
    center points of $S_\text{old}$, and that the sets $S^{(\text{reg})}_\text{old}, S^{(\text{reg})}$ may even contain points from the metric space that are not part of the point set $P$ anymore because of point deletions. 
    We remark that the sets $S^{(\text{reg})}_\text{old}, S^{(\text{reg})}$ are auxiliary sets which are used in the following lemma. 
    
    \begin{lemma} \label{lem:C_covers_dec}
        After a point deletion, consider a point $q \in P$ and let $C_j$ be the cluster that it.
        Recall that $c_j \in S$ is the center of the cluster $C_j$, and let $c^{(\text{reg})}_j \in S^{(\text{reg})}$
        be the most recent regular center of  $C_j$. Then the two properties $P_1$ and $P_2$ hold:
        \begin{enumerate}[label=\subscript{P}{\arabic*}.]
            \item The distance between the point $q$ and the point $c_j$ is at most $3\hat{r}$ (i.e., $\dist(q, c_j) \leq 3\hat{r}$).

            \item If the point $q$ is not a center point (i.e., $q \in P \setminus S$), then it holds that $\dist(q, c^{(\text{reg})}_j) \leq \hat{r}$.
        \end{enumerate}
    \end{lemma}
    \begin{proof} 
        The proof of the validity of $P_1$ and $P_2$ is by induction on the number of deleted points.
        For the base case after the preprocessing phase has finished, we have that every point is within distance $\hat{r}$
        from its center and that all the center points are regular centers (i.e., $S =
        S^{(\text{reg})}$), and thus both $P_1$ and $P_2$ are satisfied. 
        For the induction step, let $p$ be the deleted point from the point set $P$, and let $P$ be the updated point set without the deleted point $p$. In the analysis we use the following observation.

        \begin{observation} \label{obs:RegOperProp}
            Whenever the algorithm decreases the value of the radius (i.e., $\hat{r} < \hat{r}_\text{old}$) during 
            the Regulating Operation in the end of the update, the updated set of centers $S$ is a \distDS{$\hat{r}$} for 
            the updated value $\hat{r}$ of the radius, and also we have that $S^{(\text{reg})} = S$.
            More specifically every point $q \in P$ is within distance $\hat{r}$ from the center point $c_j$,
            where $c_j \in S$ is the regular center of the regular cluster $C_j$ that contains the point $q$.
        \end{observation}

        \noindent
        Notice that in the rest of the proof we can assume that $\hat{r} = \hat{r}_\text{old}$, since by construction
        the radius cannot increase and also any decrease to the radius implies the validity of both $P_1$ and $P_2$
        based on Observation~\ref{obs:RegOperProp}. During the Reassigning Operation, we assume in Line~\ref{algline:find_c_dec} that there exists such a point $c_{j'}$. 
        The next claim shows the correctness of the Reassigning Operation by arguing that this instruction is executed properly.

        \begin{claim} \label{clm:reassOper_q_atmost_r_c_dec}
            During the Reassigning Operation, for every point $q \in Q$ there exists a center point $c_{j'} \in S_\text{old} \setminus \{p\}$ such that
            $\dist(q, c_{j'}) \leq \hat{r}_\text{old}$ in Line~\ref{algline:find_c_dec}.
        \end{claim}
        \begin{proof}
            Recall that $p$ is the deleted point with index $i$ in $S_\text{old}$, and let $C_i'$
            be the cluster $C_i$ before it is updated due to the point deletion. 
            Initially the set $Q$ is initialized to $C_i' \setminus \{p\}$, and observe that the algorithm
            must have proceeded with case C2b in order to call the Reassigning Operation. Thus by construction, every point in $C_i' \setminus \{p\}$ must be within distance $\hat{r}$ from another center point of $S_\text{old} \setminus \{p\}$, and for these points Line~\ref{algline:find_c_dec} is executed properly. Observe that the algorithm proceeds with case C2b
            when there is no sequence satisfying~(\ref{eq:seq_for_replac}). 
            Suppose to the contrary that at some moment a point $q \in Q$ during the Reassigning Operation is at a distance greater than $\hat{r}_\text{old}$ from any other center point of $S_\text{old} \setminus \{p\}$. Then such a point could serve as the point $p_{t_l}$ in (\ref{eq:seq_for_replac}), contradicting the fact that there is no sequence satisfying~(\ref{eq:seq_for_replac}).
            As a result, Line~\ref{algline:find_c_dec} is executed properly all the time.
        \end{proof}

        \noindent
        We analyze the validity of the two properties $P_1$ and $P_2$ under a point deletion separately. 
    
        \paragraph{Validity of $P_2$.}
        The induction hypothesis can be applied for every point in $P \setminus S_\text{old}$. Notice that under point deletions, 
        an older center point in $S_\text{old}$ is exempted from being a center only if the adversary removes it from the point set $P$.
        Therefore, the induction hypothesis can actually be applied for all the points in $P \setminus S$ for the updated set $S$.
    
        Consider an arbitrary point $q \in P \setminus S$, and let $C_j$ be the unique cluster that contains it before the point deletion. 
        If the point $q$ switches clusters during the Regulating Operation in 
        Line~\ref{algline:switch_cluster_1} or during the Reassigning Operation in Line~\ref{algline:switch_cluster_2}, then the validity of $P_2$ for the point $q$ follows by construction. Notice that a point could switch clusters also during case C2a in Line~\ref{algline:shift_centers_dec}, but observe that by construction these are center points and thus the claim for $P_2$ is not affected.
        Otherwise assume that the point $q$ remains part of the cluster $C_j$ after the point deletion, and let $c^{(\text{old})}_j \in S^{(\text{reg})}_\text{old}$ be the most recent regular center of $C_j$ before the point deletion. We continue the analysis based on whether
        the point $c^{(\text{old})}_j$ is still the most recent regular center of the cluster $C_j$ after the point deletion or not.
        \begin{itemize}
            \item If $c^{(\text{old})}_j \in S^{(\text{reg})}$, then since we have that $q \in C_j$ 
            it holds that $\dist(q, c_j^{(\text{reg})}) \leq \hat{r}_\text{old}$
            by induction hypothesis. Recall that we can assume that $\hat{r} = \hat{r}_\text{old}$
            based on Observation~\ref{obs:RegOperProp}, and thus the point $q$ satisfies the statement of $P_2$.
            
            \item Otherwise, let $c^{(\text{new})}_j$ be the new regular center of the cluster $C_j$ (i.e., $c^{(\text{new})}_j \in S^{(\text{reg})}$).
            Since the point $q$ does not switch clusters, by construction during the Regulating Operation either Line~\ref{algline:cluster_reg_2} or Line~\ref{algline:cluster_reg_1} is executed. In turn, all the points in $C_j$ must be within distance $\hat{r}$ from the point $c^{(\text{new})}_j \in S$
            (which is also a regular center point in the updated set $S$). Therefore as the point $q$ belongs to $C_j$, the validity of $P_2$ follows. 
        \end{itemize}
        
        \paragraph{Validity of $P_1$.}
        Recall that based on Observation~\ref{obs:RegOperProp}, we can assume that $\hat{r} = \hat{r}_\text{old}$. 
        If the deleted point $p$ is not a center point (i.e., $p \notin S_\text{old}$), then 
        the point $p$ is removed from the cluster it belongs to, and the algorithm calls the Regulating Operation. 
        Consider an arbitrary point $q \in P$. If the point $q$ ends up in a regular cluster due to either 
        Line~\ref{algline:cluster_reg_2}, Line~\ref{algline:cluster_reg_1} or
        Line~\ref{algline:switch_cluster_1}, then the point $q$ satisfies the statement of $P_1$ by construction.
        Otherwise the point $q$ remains in the same cluster whose state does not change, and the validity of $P_1$
        follows by induction hypothesis.
        
        Otherwise assuming that the deleted point $p$ is the $i$-th center point (i.e., $p \in S_\text{old}$ and $p = c_i$),
        we analyze the two possible cases of the algorithm.

        \begin{enumerate}
            \item Assume that the algorithm proceeds with case C1, and let $c'$ be the corresponding
            new zombie center. Also let $C_i$ be the corresponding updated zombie cluster after the point deletion.
            By the validity of $P_2$, there is a point $c^{(\text{reg})}_i \in S^{(\text{reg})}$ such that $\dist(q, c^{(\text{reg})}_i) \leq \hat{r}$ for every point $q \in C_i$. Recall that the point $c^{(\text{reg})}_i$ is 
            the most recent regular center of $C_i$ which was deleted at some 
            point in the past (including this update). Based on the way the algorithm selects the point $c'$, we have that $c' \in C_i$. 
            Thus by triangle inequality, it holds that $\dist(q, c') \leq \dist(q, c^{(\text{reg})}_i) + \dist(c', c^{(\text{reg})}_i) \leq 2\hat{r}$ for every point $q \in C_i$. Moreover, notice that every point $q \in P$ which switches clusters during the Regulating Operation in Line~\ref{algline:switch_cluster_1} satisfies the statement of $P_1$ by construction. Finally, consider an
            arbitrary point $q \in P \setminus C_i$ which remains in the same cluster $C_j$ after the point deletion. By construction
            we have that $S = (S_\text{old} \setminus \{p\}) \cup \{c'\}$, which means that the center point $c_j \in S_\text{old} \setminus \{p\}$ 
            of the cluster $C_j$ is its center before and after the point deletion. Since the point $q$ is in $C_j$ before and after the point update, using the induction hypothesis we get that $\dist(q, c_j) \leq 3\hat{r}_\text{old} = 3\hat{r}$, and so 
            the validity of $P_1$ follows.

            \item If the algorithm proceeds with case C2, then we analyze the two subcases separately:
            \begin{enumerate} 
                \item Assume that the set of centers is modified due to case C2a, and let 
                $p_{t_1}, c_{t_2}, p_{t_2}, \;\ldots\;, c_{t_l}, p_{t_l}$ be the corresponding sequence satisfying~(\ref{eq:seq_for_replac}). 
                If any point $q \in P$ switches clusters during the Regulating Operation in Line~\ref{algline:switch_cluster_1}, then by construction
                the point $q$ satisfies the statement of $P_1$. 
                As explained before in the proof about the validity of $P_2$, a point could switch clusters also in Line~\ref{algline:shift_centers_dec}, but this happens only for center points and thus the claim for $P_1$ is not affected.
                
                Hence it remains to examine points which do not switch clusters.
                Consider an arbitrary point $q \in P$, and let $C_{j'}$ be the cluster which contains it before and after the point deletion.
                If the most recent regular center of $C_{j'}$ is updated during the Regulating Operation either in Line~\ref{algline:cluster_reg_2} or in Line~\ref{algline:cluster_reg_1}, then by construction the point $q$ is within distance $\hat{r}$ from the center $c_{j'}$ of the cluster $C_{j'}$, as needed for $q$. Now regarding the point $q \in C_{j'}$, assume that the most recent regular center $c^{(\text{reg})}_{j'}$ of the cluster $C_{j'}$ does not change after the point deletion, namely
                $c^{(\text{reg})}_{j'} \in S^{(\text{reg})}_\text{old} \cap S^{(\text{reg})}$. Recall that we have $S = (S_\text{old} \setminus \{p\}) \cup \{p_{t_l}\}$,
                and consider the following two cases:
                \begin{itemize}
                \item If the index $j'$ is equal to an index $t_j$ of the sequence that satisfies~(\ref{eq:seq_for_replac}) where $1 \leq j \leq l$,
                then by the validity of $P_2$ it holds that $\dist(q, c^{(\text{reg})}_{t_j}) \leq \hat{r}$. 
                Based on the way the algorithm sets each $c_{t_j}$ and as the sequence satisfies the properties of (\ref{eq:seq_for_replac}), there must have been a point $p_{t_j}$ in $C_{t_j}$ such that $\dist(c_{t_j}, p_{t_j}) \leq \hat{r}_\text{old}$.\footnote{Note that for $j = l$, 
                the center point $c_{t_l}$ is actually the point $p_{t_l}$, and so $\dist(c_{t_l}, p_{t_l}) = 0$.} Since before the point deletion such a point $p_{t_j}$ was in $C_{t_j}$ and the point $c^{(\text{reg})}_{t_j} = c^{(\text{reg})}_{j'}$ was the most recent regular
                center of $C_{t_j}$, by induction hypothesis of $P_2$ we can conclude that $\dist(p_{t_j}, c^{(\text{reg})}_{t_j}) \leq \hat{r}_\text{old}$.
                In turn, by triangle inequality we have that $\dist(c_{t_j}, c^{(\text{reg})}_{t_j}) \leq \dist(c_{t_j}, p_{t_j}) + \dist(p_{t_j}, c^{(\text{reg})}_{t_j}) \leq 2\hat{r}$ (recall that $\hat{r} = \hat{r}_\text{old}$). Thus by using triangle inequality again, it holds that: 
                \[
                    \dist(q, c_{t_j}) \;\leq\; \dist(q, c^{(\text{reg})}_{t_j}) \;+\; \dist(c^{(\text{reg})}_{t_j}, c_{t_j}) \;\leq\; 3\hat{r}.
                \] 
                Therefore for any point $q \in P$ which belongs to a cluster whose index is equal to any of the previous $t_j$ as defined before in~(\ref{eq:seq_for_replac}), the statement of $P_1$ holds.
                
                \item If the index $j'$ is not equal to any of the previous $t_j$ as defined before in~(\ref{eq:seq_for_replac}), we have by construction that the center $c_{j'}$ belongs to both sets $S_\text{old}$ and $S$, and its position does not change in the metric space. 
                Thus by induction hypothesis, the point $q$ is within distance 
                $3\hat{r}$ from the center point $c_{j'} \in S$, and this concludes the validity of $P_1$.
                \end{itemize}
                \item Assume that the set of centers is modified due to case C2b, and recall that $p$ is the deleted point with index $i$ in $S_\text{old}$. To avoid confusion with the reset of the center $c_i$ and its cluster $C_i$ in Lines~\ref{algline:reset_ci_C2b_dec} and \ref{algline:reset_Cli_C2b_dec},
                let $C_i'$ be the cluster $C_i$ before it is updated. 
                The algorithm calls the Reassigning Operation, and note that Line~\ref{algline:find_c_dec} is executed properly based on Claim~\ref{clm:reassOper_q_atmost_r_c_dec}.
                Since the set $Q$ in the Reassigning Operation is initialized to $C_i' \setminus \{p\}$, 
                we can conclude that all the points in $C_i'$ are added to a different regular cluster in Line~\ref{algline:switch_cluster_2}. 
                If a point $q \in P$ ends up in a regular cluster due to either Line~\ref{algline:cluster_reg_2}, Line~\ref{algline:cluster_reg_1}, Line~\ref{algline:switch_cluster_1} or Line~\ref{algline:switch_cluster_2}, then $q$ satisfies the statement of $P_1$ by construction.
                Observe that every point that switches clusters is added to a regular cluster by construction.

                Hence, it remains to consider a point $q \in P \setminus C_i'$ that remains in the same cluster after the point deletion. By construction we have that $S = (S_\text{old} \setminus \{p\}) \cup \{c'\}$,
                and the center points with index different than $i$ in $S_\text{old}$ do not
                change position in the metric space. Therefore by induction hypothesis, such a point $q$ satisfies the statement of $P_1$, and this concludes the validity of $P_1$.
            \end{enumerate}
        \end{enumerate}
    \end{proof}
    
    By combining Lemma~\ref{lem:size_S_worst_recourse_dec}, Lemma~\ref{lem:inv_2}, and Lemma~\ref{lem:C_covers_dec}, 
    we can now finish the proof of Theorem~\ref{th:decr_consist} which we restate for convenience.

    \decrconsist*
    \begin{proof}
        Consider Algorithm~\ref{alg:dec_consist} and its analysis.
        Based on Lemma~\ref{lem:C_covers_dec}, we can infer that the set $S$ is a \distDS{$3\hat{r}$}, and by Lemma~\ref{lem:inv_2} 
        we have that $\hat{r} \leq 2R^*$. Therefore using also Lemma~\ref{lem:size_S_worst_recourse_dec}, the claim follows.
    \end{proof}

    Finally, we note that the worst-case update time of our decremental algorithm (i.e., Algorithm~\ref{alg:dec_consist}) is $O(|P| \cdot k)$, since it admits the same analysis as in Lemma~\ref{lem:runtime_fully}.

%% file: trunk/incr_consist_kcenter.tex
\section{Incremental Consistent $k$-Center Clustering} \label{sec:incr}
In the incremental setting, we are given as input a point set $P$ from an arbitrary metric space which undergoes point insertions. The goal is to incrementally maintain a feasible solution 
for the $k$-center clustering problem with small approximation ratio and low recourse. 
In this section, we develop a deterministic 
incremental algorithm that maintains a $6$-approximate solution for the $k$-center clustering problem
with worst-case recourse of $1$ per update, as demonstrated in the following theorem.
Recall that by $R^*$ we denote the current optimal radius of the given instance.

\begin{restatable}{theorem}{incrconsist} \label{th:incr_consist}
    There is a deterministic incremental algorithm that,
    given a point set $P$ from an arbitrary metric space subject to point insertions and an integer $k \geq 1$,
    maintains a subset of points $S \subseteq P$ such that: 
    \vspace{-0em}
    \begin{itemize}
        \setlength\itemsep{-0.3em}
        \item The set $S$ is a \distDS{$6R^*$} in $P$ of size $k$.
        \item The set $S$ changes by at most one point per update.
    \end{itemize}
\end{restatable}

\vspace{0.5em}
Our algorithm maintains a set $S = \{c_1, \ldots, c_k\}$ of $k$ \emph{center points} 
in some \emph{order} and a value $\delta$ called \emph{extension level}. 
Further, the algorithm associates a \emph{radius} $r^{(\delta)}$ for each value of $\delta$. 
Throughout the algorithm we want to satisfy the following three invariants. 

\vspace{1em}
\noindent
\textbf{Invariant~1}: The size of the set $S$ is exactly $k$.

\noindent
\textbf{Invariant~2}: There exists a point $p_S \in P \setminus S$ such that the set $S \cup \{p_S\}$ is a \distIS{$r^{(\delta-1)}$}.

\noindent
\textbf{Invariant~3}: The set $S$ is a \distDS{$(2R^* + r^{(\delta)})$}.

\subsection{Incremental Algorithm} 
At first we describe the preprocessing phase of the algorithm, and afterwards we describe
the way the algorithm handles point insertions. A pseudocode of the algorithm is provided in Algorithm~\ref{alg:incr_consist}.

\subsubsection{Preprocessing Phase}
In the preprocessing phase, we run the algorithm of Theorem~\ref{th:2appr_alg} to obtain a $2$-approximate solution $S$,
and let $r \coloneqq \max_{p \in P} \dist(p, S)$. 
Next, we repeat the following process with the goal to satisfy all the three invariants. 
While the set $S$ has size strictly less than $k$, we add to $S$ an arbitrary point $c$
that is at a distance greater than $\frac{r}{2}$ from the current centers (i.e., $\dist(c, S) > \frac{r}{2}$). 
If there is no such point, then the value of $r$ is halved, and the process continues until the size of $S$ reaches $k$. 
Once the size of $S$ becomes $k$, the value of $r$ continues to be halved as long as the set $S$ is a \distDS{$\frac{r}{2}$}.

In the special case where the point set $P$ contains fewer than $k$ points, we assume that the preprocessing phase is ongoing and the
set $S$ consists of all points of $P$. Notice that in this scenario the three invariants may not be satisfied. However, this is not an issue
as every point serves as a center. Eventually, when the size of $P$ becomes at least $k + 1$, we execute once the previous steps of the preprocessing phase.

We initialize $r^{(-1)} \coloneqq \frac{r}{2},\; r^{(0)} \coloneqq r$ where $r$ is the final value after the preprocessing step has finished.
Also, the extension level $\delta$ is initially set to $0$. 
\vspace{1em}

\paragraph{Analysis of the preprocessing phase.}
After the preprocessing phase has finished, the set $S$ has size $k$ and the set $S$ is a \distDS{$r^{(\delta)}$},
where $\delta$ is the extension level which is initially zero.
Additionally, by construction there exists a point $p_S \in P \setminus S$ such that the set $S \cup \{p_S\}$
is a \distIS{$r^{(\delta-1)}$}. Hence, in the beginning all the three invariants are satisfied.
Combining Invariant~1 with Invariant~2, we can infer that there exists a \distIS{$r^{(\delta-1)}$} of size at least $k + 1$. Thus it follows that $r^{(\delta-1)} < 2R^*$ by Corollary~\ref{cor:r_less2R},
and as $r^{(\delta)} = 2r^{(\delta-1)}$ we have that the set $S$ is initially a $4$-approximate solution.\footnote{Actually, the set $S$ is initially a $2$-approximate solution, because we run at first the $2$-approximation algorithm of Theorem~\ref{th:2appr_alg} and the addition of more points to $S$ can only reduce its ratio. However, we use this simplistic analysis as a warm-up to demonstrate the utility
of the invariants.}

\subsubsection{Point Insertion}
Let $\delta$ be the current extension level of the algorithm.
Under point insertions, the incremental algorithm operates as described below.

Consider the insertion of a point $p$ into the point set $P$, and let $P \coloneqq P \cup \{p\}$. If $p$ is within distance $r^{(\delta)}$ from the 
set of centers $S$ (i.e., $\dist(p, S) \leq r^{(\delta)}$), 
then the set $S$ remains unmodified and the recourse is zero. Otherwise if the distance between the new point $p$ and the set of centers $S$
is greater than $r^{(\delta)}$ (i.e., $\dist(p, S) > r^{(\delta)}$), the incremental algorithm proceeds as follows. Let $c_i$
be the center point that satisfies the following conditions:
\begin{equation} \label{choose_ci_cj}
    \dist(c_i, S \setminus \{c_i\}) = \min_{x, y \in S,\,x \neq y} \dist(x, y)\text{ and }i\text{ is the largest possible index in $S = \{c_1, \ldots, c_k\}$.}
\end{equation}
The incremental algorithm then continues based on the two following cases:

\begin{enumerate}[label=C\arabic*.]
    \item If the distance between $c_i$ and the other centers is at most $r^{(\delta)}$ (i.e., $\dist(c_i, S \setminus \{c_i\}) \leq
    r^{(\delta)}$), then the algorithm builds a new cluster with the new point $p$ as its center and 
    exempts the center point $c_i$ from being a center. In particular, the point $c_i$ is removed from $S$ and
    the point $p$ is added to $S$. Observe that all points previously closest to $c_i$, now have a new closest center in the updated set $S$. Since the new point $p$ becomes the $i$-th center of the ordered set $S$, the algorithm sets $c_i \coloneqq p$. The recourse in this case is $1$.
     
    \item Otherwise, the pairwise distances between the centers are greater than $r^{(\delta)}$. In this case, the algorithm first calls the Doubling Operation and then the Gonzalez Operation,
    both of which are described below. Next for the updated value of $\delta$, if the new point $p$ 
    is within distance $r^{(\delta)}$ from a center,
    then the set $S$ remains unmodified and the recourse is zero. Otherwise it holds that $\dist(p, S) > r^{(\delta)}$, and the algorithm finds the center point $c_i$ as defined before in~(\ref{choose_ci_cj}).
    Notice that $\dist(c_i, S \setminus \{c_i\}) \leq r^{(\delta)}$ by the construction of the Doubling Operation. Subsequently, the incremental algorithm continues as in case C1.

    \subparagraph{Doubling Operation.}
    Recall that $p$ is the new inserted point.
    During the Doubling Operation, as long as the set $S$ is a \distIS{$r^{(\delta)}$} and $\dist(p, S) > r^{(\delta)}$, the incremental algorithm increases the extension level $\delta$ by one and sets $r^{(\delta)}$ to $2r^{(\delta-1)}$. 

    \subparagraph{Gonzalez Operation.}
    During the Gonzalez Operation, the incremental algorithm reorders the set of centers $S$ by running the Gonzalez’s algorithm on $S$.
    In particular, the algorithm sets $c_1$ to be an arbitrary point from $S$,
    and performs the following step for $k - 1$ iterations. In the $i$-th iteration for $i \geq 2$, it
    sets $c_i$ to be a point in $S \setminus \{c_1, \ldots, c_{i-1}\}$ that is farthest from $\{c_1, \ldots, c_{i-1}\}$.
\end{enumerate}

\begin{algorithm}[ht!]
    \DontPrintSemicolon
    \caption{\textsc{incremental consistent $k$-center}{}}
    \label{alg:incr_consist}

    \SetAlgoLined
    \SetArgSty{textrm}

    \tcp{Let $S = \{c_1, \dots, c_k\}$ be the ordered set of $k$ centers}
    
    \tcp{Let $\delta$ be the current extension level}
    
    \vspace{0.5em}

    \SetKwFunction{FDoublingOp}{DoublingOp}
    \Procedure{\FDoublingOp{p}} {
        \While{$\min_{x, y \in S,\,x \neq y} \dist(x, y) > r^{(\delta)} \textbf{ and } \dist(p, S) > r^{(\delta)}$} {
            $\delta \gets \delta + 1$
        
            $r^{(\delta)} \gets 2 r^{(\delta-1)}$
        }
    }

    \vspace{1em}
    
    \SetKwFunction{FGonzalezOp}{GonzalezOp}
    \Procedure{\FGonzalezOp{}} {
        Let $c_1$ be an arbitrary center point of $S$
        \vspace{0.5em}
        
        \For{\texttt{$i \in \{2, \ldots, k\}$}} {
            $c_i \gets \argmax_{c \in S \setminus \{c_1, \ldots, c_{i-1}\}} \dist(c, \{c_1, \ldots, c_{i-1}\})$
        }

        \vspace{0.5em}
        $S \gets \{c_1, \ldots, c_k\}$ \tcp{Recall that $S$ is an ordered set}
    }

    \vspace{1em}
    
    \SetKwFunction{FInsertPoint}{InsertPoint}
    \Procedure{\FInsertPoint{p}} {
        Add $p$ to $P$
        \vspace{0.5em}
        
        \If{$\dist(p, S) > r^{(\delta)}$} {
            Find $c_i \in S$ such that
            $\dist(c_i, S \setminus \{c_i\}) = \min_{x, y \in S,\,x \neq y} \dist(x, y)$,\label{algline:pick_ci}
            
            and $i$ is the largest possible index in $S = \{c_1, \dots, c_k\}$ \label{algline:determ_i}
            \vspace{0.5em}
        
            \If(\tcp*[h]{case C1}){$\dist(c_i, S \setminus \{c_i\}) \leq r^{(\delta)}$} { \label{algline:ci_cj_leq_rdelta}
                Remove $c_i$ from $S$ \label{algline:c1a}
                
                Add $p$ to $S$

                $c_i \gets p$ \label{algline:c1b}
            }
            \vspace{0.5em}
            
            \Else(\tcp*[h]{case C2}){  
                \FDoublingOp{p}

                \FGonzalezOp{}{}
                \vspace{0.5em}
                
                \If{$\dist(p, S) > r^{(\delta)}$} {
                    Find $c_i \in S$ such that
                    $\dist(c_i, S \setminus \{c_i\}) = \min_{x, y \in S,\,x \neq y} \dist(x, y)$,
                    
                    and $i$ is the largest possible index in $S = \{c_1, \dots, c_k\}$
                    \vspace{0.5em}
                    
                    \tcp{case C1 after case C2}
                    
                    Remove $c_i$ from $S$ \label{algline:c1_a}
                
                    Add $p$ to $S$
        
                    $c_i \gets p$ \label{algline:c1_b}
                }
            }
        }
    }
\end{algorithm}

\subsection{Analysis of the Incremental Algorithm}
Our goal in this section is to prove Theorem~\ref{th:incr_consist} by analyzing
the incremental Algorithm~\ref{alg:incr_consist}. During a point insertion, let $S_\text{old}$ be the set of centers before the algorithm has processed the update,
and $S$ be the set of centers after the algorithm has processed the update.
Similarly, let $\delta_\text{old}$ be the value of the extension level before the algorithm has processed the update,
and $\delta$ be the value of the extension level after the algorithm has processed the update.

The proof of Theorem~\ref{th:incr_consist} almost immediately follows from the three invariants. Therefore, 
we prove that all of them are satisfied in Algorithm~\ref{alg:incr_consist} after a point insertion.

\begin{lemma} \label{lem:size_S_worst_recourse}
    After a point insertion, the size of $S$ is $k$ (i.e., Invariant~1 is satisfied). Moreover,
    the worst-case recourse per update is $1$.
\end{lemma}
\begin{proof}
    The proof follows an induction argument. In the preprocessing phase, the algorithm constructs a set $S$ of $k$ centers.
    Under a point insertion, the set of centers can be modified only due to case C1 in Lines~\ref{algline:c1a}-\ref{algline:c1b} or Lines~\ref{algline:c1_a}-\ref{algline:c1_b}.
    In this case, only the new point is added to the set of centers and only one point is removed 
    from the set of centers. As a result, the size of the set of centers remains the same (i.e., $|S| = k$), and the worst-case recourse after any point insertion is at most $1$.
\end{proof}

\begin{lemma} \label{lem:r_4R} 
    After a point insertion, there exists a point $p_S \in P \setminus S$ such that the set $S \cup \{p_S\}$ is a 
    \distIS{$r^{(\delta-1)}$} (i.e., Invariant~2 is satisfied). Moreover,
    the value of the radius $r^{(\delta)}$ is at most $4R^*$.
\end{lemma}
\begin{proof}
    We first prove by induction on the number of inserted points that Invariant~2 is satisfied. For the base case, the statement holds by the preprocessing phase. For the induction step, let $p$ be the inserted point into the point set $P$,
    and let $P$ be the updated point set including the new point $p$. Observe that if $p$ is within distance $r^{(\delta_\text{old})}$ from the set
    $S_\text{old}$ then $S = S_\text{old}$ and $\delta = \delta_\text{old}$. By the induction hypothesis, there exists a point $p_S \in (P \setminus \{p\}) \setminus S_\text{old}$ such that $S_\text{old} \cup \{p_S\}$ is a \distIS{$r^{(\delta_\text{old} - 1)}$}. In turn, we can conclude that $p_S \in P \setminus S$ and that $S \cup \{p_S\}$ is a \distIS{$r^{(\delta-1)}$}, as needed.
    
    Otherwise, assume that the distance between the new point $p$ and the set $S_\text{old}$ is greater than $r^{(\delta_\text{old})}$ (i.e., $\dist(p, S_\text{old}) > r^{(\delta_\text{old})}$). We continue the analysis based on whether the set of centers is modified, using also the following observation.
    
    \begin{observation} \label{obs:doublOper}
        During the Doubling Operation, the incremental algorithm increases the extension level to some $\delta'$ where $\delta_\text{old} < \delta' \leq \delta$, only when the set $S_\text{old}$ is a \distIS{$r^{(\delta'-1)}$} and $\dist(p, S_\text{old}) > r^{(\delta' - 1)}$.
    \end{observation}

    \begin{enumerate}
    \item If the set of centers is not modified (i.e., $S = S_\text{old}$), then the algorithm proceeds with case C2 and calls the Doubling Operation. Based on Observation~\ref{obs:doublOper}, there exists a point $p_S \in P \setminus S$ (e.g., the inserted point $p$) 
    at distance greater than $r^{(\delta-1)}$ from $S$, and the set $S$ is a \distIS{$r^{(\delta-1)}$}. Consequently, the set $S \cup \{p_S\}$ is a \distIS{$r^{(\delta-1)}$}, and thus Invariant~2 is satisfied.
    
    \item Otherwise if the set of centers is modified, then notice that this occurs only when $\dist(p, S_\text{old}) > r^{(\delta)}$. There are two subcases to consider in this scenario.
    \begin{enumerate}
    \item If the algorithm proceeds with case C1 in Lines~\ref{algline:c1a}-\ref{algline:c1b}, 
    an older center point $c_i$ is exempted from being a center. To avoid confusion with the reset of $c_i$ in Line~\ref{algline:c1b},
    let $c_i'$ be the old point $c_i$ up to Line~\ref{algline:c1a}. By construction, we have $S = (S_\text{old} \setminus \{c_i'\}) \cup
    \{p\}$ and $\delta = \delta_\text{old}$. By the induction hypothesis, the set $S_\text{old}$ is a \distIS{$r^{(\delta-1)}$}.
    Thus by setting $p_S = c_i'$ and since $\dist(p, S_\text{old}) > r^{(\delta)}$, we can infer that there exists a point $p_S \in P \setminus S$ at distance greater than $r^{(\delta-1)}$ from the set $S$
    and that $S$ is a \distIS{$r^{(\delta-1)}$}, as needed.
    
    \item Alternatively if the algorithm proceeds with case C2 and then with case C1, the Doubling Operation is called and Lines~\ref{algline:c1_a}-\ref{algline:c1_b} are executed. Based on Observation~\ref{obs:doublOper}, the set $S_\text{old}$ is a \distIS{$r^{(\delta-1)}$}. Thus by setting $p_S = c_i'$ where $c_i'$ is the old point $c_i$ up to Line~\ref{algline:c1_a}, the same argument as before holds.
    \end{enumerate}
    \end{enumerate}
    
    Regarding the upper bound on $r^{(\delta)}$, since Invariant~2 is satisfied
    there exists a point $p_S \in P \setminus S$ such that the set $S \cup \{p_S\}$ is a \distIS{$r^{(\delta-1)}$}.
    Based on Lemma~\ref{lem:size_S_worst_recourse}, the size of the set $S \cup \{p_S\}$ is $k + 1$,
    and using Corollary~\ref{cor:r_less2R} we get that $r^{(\delta-1)} < 2R^*$. Therefore as $r^{(\delta)} = 2r^{(\delta-1)}$ by construction, the claim follows.
\end{proof}

We continue with the proof that Invariant~3 is satisfied throughout the algorithm.
Let $S_{\text{init}}$ be the ordered set of centers resulting from the last execution of the Gonzalez Operation, and
$S_l$ be the ordered subset of $S_{\text{init}}$ up to center point $c_l$ for any $1 \leq l \leq k$. 
Observe that the incremental algorithm calls the Gonzalez Operation once per extension level, 
and this happens right after the extension level is increased. Hence,
the set $S_\text{init}$ is the set of centers at the beginning of the extension level $\delta_\text{old}$, and it may be different
from the current set of centers $S_\text{old}$ just before the algorithm processes a point insertion. We note that the set $S_\text{init}$ is an auxiliary set which is used in the following lemma.

\begin{lemma} \label{lem:C_covers}
    After a point insertion, it holds that:
    \begin{enumerate}[label=\subscript{P}{\arabic*}.]
        \item The distance between any point $q \in P$ and the set $S$ is at most $\min(2R^* + r^{(\delta)},\; 2r^{(\delta)})$.
        In other words, it holds that $\max_{q \in P} \dist(q, S) \leq \min(2R^* + r^{(\delta)},\; 2r^{(\delta)})$ 
        (i.e., Invariant~3 is satisfied).

        \item The distance between any point $q \in P$ and the set $S_\text{init} \cup S$ is at most $r^{(\delta)}$. In other words, it holds that $\max_{q \in P} \dist(q, S_\text{init} \cup S) \leq r^{(\delta)}$.
    \end{enumerate}  
\end{lemma}
\begin{proof}
    The validity of $P_1$ and $P_2$ is established by induction on the number of inserted points. For the base case once the preprocessing phase has finished, we have $\delta = 0$ and that every point is within distance $r^{(0)}$ from a center, and thus both $P_1$ and $P_2$ hold. 
    For the induction step, let $p$ be the inserted point into the point set $P$, and let $P$ be the updated point set including the new point $p$.
    If $p$ is within distance $r^{(\delta_\text{old})}$ from the set
    $S_\text{old}$ (i.e., $\dist(p, S_\text{old}) \leq r^{(\delta_\text{old})}$), then $S = S_\text{old}, \delta = \delta_\text{old}$ and the set $S_\text{init}$ remains unchanged. Thus, the claims for $P_1$ and $P_2$ follow by the induction hypothesis.

    Otherwise, assume that the distance between the new point $p$ and the set $S_\text{old}$ is greater than $r^{(\delta_\text{old})}$ 
    (i.e., $\dist(p, S_\text{old}) > r^{(\delta_\text{old})}$). We analyze the two possible cases of the incremental algorithm.

    \begin{enumerate}
        \item Assume that the algorithm proceeds with case C1 in Lines~\ref{algline:c1a}-\ref{algline:c1b}.
        In this case, we have $\delta = \delta_\text{old}$ and that the set $S_\text{init}$ remains unchanged. Let $c_i$ be the corresponding
        center point selected in~\cref{algline:pick_ci}. 
        To avoid confusion with the reset of $c_i$ in Line~\ref{algline:c1b},
        let $c_i'$ be the old point $c_i$ up to Line~\ref{algline:c1a}.

        \subparagraph{Validity of $P_2$.}
        By the induction hypothesis, every point $q \in P \setminus \{p\}$ is within distance $r^{(\delta_\text{old})}$ from its closest point in $S_\text{init} \cup S_\text{old}$. 
        Since the set $S_\text{init}$ is not modified and $S = (S_\text{old} \setminus \{c_i'\}) \cup \{p\}$, for the validity of $P_2$ it suffices to show the following claim.

        \antonis{Review this.}
        \begin{claim} \label{clm:c_in_Sinit}
            The point $c_i'$ belongs to $S_\text{init}$.
        \end{claim} \vspace{-1.5em}
        \begin{proof}
            By construction in Line~\ref{algline:ci_cj_leq_rdelta}, the distance between the point $c_i'$ and the set $S_\text{old}$ is at 
            most $r^{(\delta_\text{old})}$. Observe that whenever a new point $p'$ is added to the set of centers in Line~\ref{algline:c1b} or Line~\ref{algline:c1_b}, its distance from the other centers is greater than the corresponding radius. Consequently, it holds that $\dist(p', S_\text{old}) > r^{(\delta_\text{old})}$. Hence, the point $c_i'$ cannot have been added to $S_\text{old}$ after the last execution of the 
            Gonzalez Operation, which implies that $c_i'$ must belong to $S_\text{init}$. \end{proof} 
        
        Therefore, every point $q \in P$ is within distance $r^{(\delta)}$ from the set $S_\text{init} \cup S$, and so $P_2$ is satisfied.

        \vspace{-0.7em} \subparagraph{Validity of $P_1$.}
        Remember that for a point $q \in P$, its closest point in $S_\text{init}$ may not belong to the set of centers $S$ after the algorithm processes the point insertion. 
        To that end, we show that there exists another point in $S$ which is close enough to prove the statement. 
        
        We need the following claim which essentially says that all the center points which have been exempted already from being centers
        for the extension level $\delta_{\text{old}}$ since the last execution of the Gonzalez Operation (i.e., points in $S_\text{init}$ 
        and not in $S_\text{old}$), must have index higher than $i$ as determined in Line~\ref{algline:determ_i} with respect to $S_\text{init}$.
        \begin{claim} \label{clm:exempted_centers}
            The set $S_{i-1}$ is a subset of $S_\text{old}$ and of $S$.
        \end{claim}
        \vspace{-1.4em}
        \begin{proof}
            Note that in this scenario, by construction of the algorithm we have that $S = (S_\text{old} \setminus \{c_i'\}) \cup \{p\}$.
            Since $p \notin S_{i-1}$, it is sufficient to show that $S_{i-1} \subseteq S$.
            The proof is using the minimal counter-example argument. Consider the first time that the set $S_{i-1}$
            is not a subset of $S$, and let $c_j$ be the point with the smallest index $j$ in $S_{i-1}$ which 
            is missing from $S$ because it was replaced by another point
            (i.e., $c_j \in S_{i-1} \setminus S$ with the smallest possible index $j$).
            By the minimality of the counter-example, it holds that $S_{j-1} \subseteq S$. 
            Since the point $c_j$ belongs to $S_{i-1}$, 
            we have that $j < i$ by definition. Hence the point $c_i'$ belongs to $S$ but not to $S_{j-1}$.
            By the manner Gonzalez Operation works which always selects the furthest point,
            it holds that $\dist(c_i', S_{j-1}) \leq \dist(c_j, S_{j-1})$.
            Hence, the minimum pairwise distance at the moment the point $c_j$ was removed from $S$, 
            could also be achieved with $c_i'$. Thus as $j < i$, the algorithm should have removed the point 
            $c_i'$ from the set $S$ instead, which yields a contradiction.
        \end{proof}
        \vspace{-0.5em}

        Let $c'$ be any point which has been exempted already from being a center for the extension level $\delta_{\text{old}}$ since the last execution of the Gonzalez Operation. 
        In other words, the point $c'$ belongs to $S_{\text{init}}$ but $c'$ is not contained in $S_\text{old}$ (i.e., $c' \in S_\text{init} \setminus S_\text{old}$).
        Based on Claim~\ref{clm:exempted_centers}, we can infer that the point $c'$ does not belong to $S_{i-1}$.
        By the manner Gonzalez Operation works which always selects the furthest point, it holds that
        $\dist(c', S_{i-1}) \leq \dist(c_i', S_{i-1})$, and by Claim~\ref{clm:exempted_centers} it follows that $\dist(c', S) \leq \dist(c', S_{i-1})$. Since $i$ is the largest such index,
        we have that $\dist(c_i', S_{i-1}) \leq \dist(c_i', S_\text{old} \setminus \{c_i'\})$. By combining all these together, we can conclude that $\dist(c', S) \leq \dist(c_i', S_\text{old} \setminus \{c_i'\})$.
        
        By construction in Line~\ref{algline:ci_cj_leq_rdelta}, it holds that $\dist(c_i', S_\text{old} \setminus \{c_i'\}) \leq r^{(\delta_\text{old})} < \dist(p, S_\text{old})$. Moreover, since
        the distance between $c_i'$ and $S_\text{old} \setminus \{c_i'\}$ is the minimum over all the pairwise distances of points
        in $S_\text{old}$, we have that the pairwise distances of points in the set $S_\text{old} \cup \{p\}$ are at least $\dist(c_i', S_\text{old} \setminus \{c_i'\})$.
        By Lemma~\ref{lem:size_S_worst_recourse}, the size of $S_\text{old} \cup \{p\}$ is $k + 1$.
        Hence based on Lemma~\ref{lem:k+1_atleastr_lesseq2R}, it follows that $\dist(c_i', S_\text{old} \setminus \{c_i'\}) \leq 2R^*$. Finally, by combining all these together we can conclude that:
        \begin{equation} \label{eq:dist_exempt_to_S}
            \dist(c', S) \;\leq\; \dist(c_i', S_\text{old} \setminus \{c_i'\}) \;\leq\; \min(2R^*,\; r^{(\delta_\text{old})}).
        \end{equation}

        From the validity of $P_2$ in this induction step, every point $q \in P$ is within distance $r^{(\delta)}$ from the set $S_\text{init} \cup S$, and recall that $\delta = \delta_\text{old}$.
        Therefore by Inequality~\ref{eq:dist_exempt_to_S} and triangle inequality, all the points are within distance $\min(2R^* + r^{(\delta)},\; 2r^{(\delta)})$ from another center point of the updated set of centers $S$, and so $P_1$ is satisfied.

        \item Assume that the algorithm proceeds with case C2. In this case, the algorithm calls the Doubling
        Operation, which means that $\delta_\text{old} \leq \delta - 1$ and in turn $2r^{(\delta_\text{old})} \leq r^{(\delta)}$.
        If the set of centers is not modified (i.e., $S = S_\text{old})$, then the new point $p$ must be within distance
        $r^{(\delta)}$ from the set $S$ (i.e., $\dist(p, S) \leq r^{(\delta)}$). 
        Moreover by induction hypothesis, any point $q \in P \setminus \{p\}$
        is within distance $2r^{(\delta_\text{old})}$ from the set $S_\text{old}$.
        Since we have that $r^{(\delta)} \geq 2r^{(\delta_\text{old})}$, it follows that the set $S$ is a \distDS{$r^{(\delta)}$}. By construction of the algorithm, for this case we have that
        $S_\text{init} = S$. Therefore, both $P_1$ and $P_2$ are satisfied.
        
        Otherwise if the set of centers is modified, it holds that $\dist(p, S_\text{old}) > r^{(\delta)}$ and the set $S_\text{old}$ cannot be a \distIS{$r^{(\delta)}$} by the way the Doubling Operation works. Observe that after the Doubling Operation,
        the algorithm calls the Gonzalez Operation, and so we have that $S_\text{init} = S_\text{old}$.
        Let $c_i$ be the corresponding center point, and to avoid confusion with the reset of $c_i$
        in Line~\ref{algline:c1_b}, let $c_i'$ be the old point $c_i$ up to Line~\ref{algline:c1_a}.
        By induction hypothesis, every point $q \in P \setminus \{p\}$ is within
        distance $2r^{(\delta_\text{old})}$ from its closest point in $S_\text{old}$. 
        Since we have that $r^{(\delta)} \geq 2r^{(\delta_\text{old})}$, it follows that
        the distance between any point $q \in P \setminus \{p\}$ and the set $S_\text{init}$ (or $S_\text{old}$) is at most $r^{(\delta)}$. 
        Also by construction in Line~\ref{algline:c1_b}, we have that $p \in S$. Thus, we can conclude
        that $\dist(q, S_\text{init} \cup S) \leq r^{(\delta)}$ for any point $q \in P$, and so $P_2$ is satisfied.
        In turn, every point $q \in P \setminus \{p\}$ whose center was the old center point $c_i'$, is within distance 
        $r^{(\delta)}$ from the point $c_i'$. Moreover by the same argument as before in Inequality~\ref{eq:dist_exempt_to_S}, it holds that 
        $\dist(c_i', S_\text{old} \setminus \{c_i'\}) \leq 2R^*$ and $\dist(c_i', S_\text{old} \setminus \{c_i'\}) \leq r^{(\delta)}$.             
        Therefore by triangle inequality and as $S = (S_\text{old} \setminus \{c_i'\}) \cup \{p\}$, all the points are within 
        distance $\min(2R^* + r^{(\delta)},\; 2r^{(\delta)})$ 
        from another center point of the updated set of centers $S$, as needed for the validity of $P_1$.
    \end{enumerate}    
\end{proof}

By combining Lemma~\ref{lem:size_S_worst_recourse}, Lemma~\ref{lem:r_4R}, and Lemma~\ref{lem:C_covers}, 
we can now finish the proof of Theorem~\ref{th:incr_consist} which we restate for convenience.

\incrconsist*
\begin{proof}
    Consider Algorithm~\ref{alg:incr_consist} and its analysis.
    By Lemma~\ref{lem:C_covers}, we can infer that the set $S$ is a \distDS{$(2R^* + r^{(\delta)})$}. 
    Based on Lemma~\ref{lem:r_4R}, we have that $r^{(\delta)} \leq 4R^*$. Therefore, by setting $r = 2R^* + r^{(\delta)}$ 
    and by Lemma~\ref{lem:size_S_worst_recourse}, the claim follows.
\end{proof}

Notice that the number of operations performed by our incremental algorithm (i.e., Algorithm~\ref{alg:incr_consist}) 
after a point insertion, is a polynomial in the input parameter $k$. Therefore, our incremental algorithm is faster than our decremental algorithm (i.e., Algorithm~\ref{alg:dec_consist}).

\begin{lemma}
    Algorithm~\ref{alg:incr_consist} performs $O(k^2)$
    operations per point insertion.
\end{lemma}
\begin{proof}
    After a point insertion, the algorithm needs $O(k^2)$ time
    for the Gonzalez Operation, and observe that every other task can be performed in $O(k)$ operations.
\end{proof}

Finally, we note that our incremental algorithm (i.e., Algorithm~\ref{alg:incr_consist}) achieves the tight $O(k \log n)$ total recourse bound as well.

\subsection{Doubling Algorithm and Worst-Case Recourse} \label{sec:doubling_alg}
Lattanzi and Vassilvitskii~\cite{LattanziV17} argued that the $8$-approximation
``doubling algorithm'' by Charikar, Chekuri, Feder, and Motwani~\cite{CharikarCFM97}
has $O(k \log n)$ total recourse, where $n$ is the total number of points inserted by the adversary. In~\cite{LattanziV17} the authors provide a lower bound of $\Omega(k \log n)$ for the total recourse, which implies that the total recourse of the ``doubling algorithm'' is optimal. In turn, the amortized recourse
of the ``doubling algorithm'' is $O(k \log n / n)$. Interestingly enough, 
a \emph{lazy} version of the ``doubling algorithm'' has also worst-case recourse of $1$ as demonstrated in the following observation.
Therefore, this justifies our claim that the approximation ratio of our incremental algorithm is an improvement over prior work.

\begin{observation}
    The $8$-approximation ``doubling algorithm'' by \cite{CharikarCFM97} can be adjusted such that its worst-case recourse is at most $1$.
\end{observation}
\begin{proof}
    During the merging stage of the ``doubling algorithm'', the algorithm picks an arbitrary cluster and merges all its neighbors to it.
    In turn, more than one center could be removed from the set of centers, resulting in a high worst-case recourse.
    However, observe that at every update only one point (i.e., the new inserted point) can be added to the set of centers.
    Hence by adjusting the previous step such that the algorithm removes only one center from the set of centers per update, yields an algorithm with worst-case recourse of $1$. 
    
    There is one subtle thing to explain regarding the way the algorithm picks clusters to merge, as this affects the approximation ratio. Once a cluster is picked so that its neighbors are merged into it,
    the algorithm should remember it, and let us call such a cluster a \emph{leader}.
    Then in the future updates, a leader cluster is never merged to other clusters, but rather only
    the neighbors of a leader cluster are merged to it. In the original algorithm this merging happens at once,
    but since here we adopt a lazy approach, this slight modification is needed to preserve the approximation ratio.
\end{proof}

%% file: trunk/acknowledgements.tex
\subsection*{Acknowledgments}
Antonis Skarlatos would like to thank Christoph Grunau for taking the time to meet and engage in thoughtful conversations in Boston.